%% file: main.tex
\documentclass{article}
\usepackage[letterpaper,top=1in,bottom=1in,left=1in,right=1in,marginparwidth=0.5in]{geometry}
\usepackage[utf8]{inputenc}
\PassOptionsToPackage{hyphens}{url}\usepackage{hyperref}
\usepackage{multicol}
\usepackage{amsmath}
\usepackage{amsthm}
\usepackage{amssymb}
\usepackage{csquotes}
\usepackage[english]{babel}
\usepackage{braket}
\usepackage{comment}
\usepackage{booktabs}
\usepackage{tabularx}
\usepackage{ragged2e}
\usepackage{caption}
\usepackage{hyperref}
\usepackage{lscape}
\usepackage{array}
\usepackage{authblk}

\usepackage{tikz}
\usepackage{tikz-cd}
\usetikzlibrary{shapes.geometric, intersections, through}
\usetikzlibrary{matrix,positioning,decorations.pathreplacing}
\usetikzlibrary{calc}

\newcommand{\sym}[1]{\mathfrak{S}_{#1}}

\newcommand{\seclass}[3]{[[#1,#2:{#3}]]} 
\newcommand{\eclass}[4]{[[#1,#2,#3\,:\,#4]]} 

\makeatletter
\def\bb{bb}
\@tfor\next:=ABCDEFGHIJKLMNOPQRSTUVWXYZ\do
 {\begingroup\edef\x{\endgroup
    \noexpand\@namedef{\bb\next}{\noexpand\mathbb{\next}}%
  }\x}
\makeatother

\makeatletter
\def\cal{cal}
\@tfor\next:=ABCDEFGHIJKLMNOPQRSTUVWXYZ\do
 {\begingroup\edef\x{\endgroup
    \noexpand\@namedef{\cal\next}{\noexpand\mathcal{\next}}%
  }\x}
\makeatother

\newcommand{\aut}[1]{\mathrm{Aut}(#1)}

\newcommand{\Pauli}{\mathcal{P}}
\newcommand{\PauliN}{\Pauli_n}

\newcommand{\baseCliff}{\mathcal{C}\ell}
\newcommand{\iCl}[1]{\baseCliff_{#1}^{\infty}}
\newcommand{\Cl}[1]{\baseCliff_{#1}}
\newcommand{\Cln}[2]{\baseCliff_{#1}^{\otimes #2}}

\newcommand{\gcoeff}[2]{\genfrac{[}{]}{0pt}{}{#1}{#2}}
\newcommand{\db}{\backslash\mkern-5mu\backslash}

\newcolumntype{L}[1]{>{\hsize=#1\hsize\RaggedRight} X}

\newif\ifphysics
\physicstrue
\newif\ifmath
\mathfalse

\ifphysics
    
\fi
\ifmath
    
\fi

\newtheorem{theorem}{Theorem}[section]
\newtheorem{lemma}[theorem]{Lemma}
\newtheorem{proposition}[theorem]{Proposition}
\newtheorem{corollary}[theorem]{Corollary}

\newtheorem{definition}{Definition}[section]

\makeatletter
\renewenvironment{proof}[1][\proofname]{\par
  \vspace{-\topsep}
  \pushQED{\qed}%
  \normalfont
  \topsep3pt \partopsep3pt 
  \trivlist
  \item[\hskip\labelsep
        \itshape
    #1\@addpunct{.}]\ignorespaces
}{%
  \popQED\endtrivlist\@endpefalse
  \addvspace{6pt plus 6pt} 
}

\makeatother

\title{Small Binary Stabilizer Subsystem Codes}
\author[1]{Andrew Cross}
\author[2]{Drew Vandeth}
\affil[1]{IBM Quantum, IBM T.J. Watson Research Center}
\affil[2]{IBM Quantum, IBM Brisbane}

\begin{document}
\maketitle

\begin{abstract}
    We establish a database consisting of a representative of every binary quantum stabilizer code under local Clifford permutation equivalence for $n\leq 9$. 
\end{abstract}

\section{Introduction}

Several databases and tables exist that list quantum codes. Due to the sheer volume of possible codes, databases such as \texttt{codetables.de}~\cite{Grassl:codetables} by Markus Grassl or \texttt{quantumcodes.info}~\cite{aydin2021databasequantumcodes} by Nuh Aydin et. al. or Magma's QECC database~\cite{MagmaQECC} tend to only list the best or extremal codes~\footnote{An extremal code is a $[[n,k,d]]$ code where either $k$ is maximal  or $d$ is maximal (with possibly $k$ maximal secondly) given $n$.}. Lars Eirik Danielsen's database of self-dual quantum codes~\cite{DanielsenDatabase} on the other hand is a comprehensive database but is restricted to equivalence classes of zero-dimensional quantum codes. Databases like the Error Correction Zoo~\cite{ErrorCorrectionZoo} organized by Victor Albert and Philippe Faist on the other hand focus on listing different types of codes that have been studied in published papers and theses. Other databases exist that list specific types of codes such as twisted codes by Jürgen Bierbrauer and Yves Edel~\cite{Bierbrauer2000},  triorthogonal codes by Jain and Albert~\cite{Jain2024HighdistanceCW} or triorthogonal codes by Nezami and Haah~\cite{Nezami2021ClassificationOS}.

This article is concerned with cataloging and enumerating binary stabilizer codes~\cite{CalderbankEtAl1997} and their generalization to binary stabilizer subsystem codes~\cite{Poulin2005}. Binary stabilizer and subsystem codes are subspaces or subsystems within the Hilbert space $\calH\simeq (\bbC^2)^{\otimes n}$ that are defined via subgroups of the $n$-qubit Pauli group $\Pauli_n$. This relationship allows us to more compactly describe these codes and search for new codes. To reduce the number of codes to catalogue and enumerate we take the approach in~\cite{DanielsenDatabase} or~\cite{YCO07} and only consider equivalence classes of codes under some reasonable definition of equivalence. In particular, we find a representative of each equivalence class of binary stabilizer code under Local Clifford with permutation equivalence. From this catalogue we can then study, catalogue or enumerate subsystem codes. Recent work on the equivalence of codes, and in particular canonical forms of codes, can be found in~\cite{Khesin2024EquivalenceCO}.

This paper is organized as follows. First, Section~\ref{sec:code-equivalence} reviews what it means for two stabilizer codes to be equivalent. Section~\ref{sec:mass-formulas} reviews mass formulas and how they can be used to verify counting results. Section~\ref{sec:canonical-representations} defines a canonical representation for stabilizer codes that can be used to test code equivalence. Section~\ref{sec:search-strategies} describes our approaches for enumerating inequivalent stabilizer codes. Section~\ref{sec:code-properties} defines several code properties that can be computed. Section~\ref{sec:results} presents the results of our search.

\section{Code Equivalence}\label{sec:code-equivalence}

We first briefly discuss Pauli and Clifford groups and their relationship to various definitions of code equivalence. The various definitions of code equivalence can make it easier to discover codes with useful properties. 

\subsection{Pauli and Clifford Groups} 
The Pauli group $\PauliN$ on $n$ qubits is defined as $\langle i, X, Z\rangle^{\otimes n}$ and the complex Clifford group $\iCl{1}$ on a single qubit  is defined as the normalizer of $\Pauli_1$ over the unitary group $\calU(2)$. As the Clifford group $\iCl{1}$ is infinite it causes unnecessary computational problems and so we instead transfer to using the finite Clifford group $\Cl{1} = \langle H, S \rangle$. Since $(HS)^3 = \zeta_8 I$ we can further reduce this group to $\Cl{1}/\langle\zeta_8I,\Pauli_1\rangle$ which is isomorphic to $\sym{3}$ where $\sym{3}$ is the symmetric group on $3$ letters. The effective action of the is group, via conjugation, on the Pauli group $\Pauli_1$ is permutation of $X$, $Z$ and $Y$. 

For notation purposes it will be useful to pick six Clifford operators that, when acting via conjugation, behave as the elements of $\sym{3}$. We take 
\begin{align*}
    I &\longleftrightarrow () & H &\longleftrightarrow (XZ) \\
    R=HS &\longleftrightarrow (XYZ) & S &\longleftrightarrow (XY) \\
    R^{-1}=SH &\longleftrightarrow (ZYX) & V=HSH &\longleftrightarrow (YZ)
\end{align*}

\noindent When we refer to local finite Clifford group $\Cl{1}$ we will in general be referring to $(\Cl{1}/\langle\zeta_8I,\Pauli_1\rangle)^{\otimes n}$ but the context will make this clear.

\subsection{Stabilizer Subsystem Codes}
We quickly revisit the definition of stabilizer groups and codes as it will be convenient to extend the definition of what constitutes as a stabilizer code associated to a given stabilizer group. A subgroup $\mathcal{S}$ of $\PauliN$ such that $-I\not\in \mathcal{S}$ is called a \textit{stabilizer group}\footnote{Often the definition includes the added restriction that $\calS$ be abelian group but this is redundant as the condition $-I\not\in\calS$ implies that $\calS$ is abelian. The reverse condition is not true since $\langle-I\rangle$ is an abelian group that contains $-I$.} 
A subgroup $\mathcal{G}$ of $\PauliN$ such that $\mathcal{G} = \langle iI, \mathcal{G}'\rangle$ for some subgroup $\mathcal{G}'$ of $\PauliN$ is called a \textit{gauge group}. Any subgroup of $\PauliN$ can be extended into a gauge group by simply adjoining the generator $iI$. Instead of associating a single code space to a given group we can instead associate a set of code spaces. Let $\calA$ be any subgroup of $\PauliN$ with $\calG_\calA = \langle\calA,iI\rangle$.  Let $\calS_\calA$ any stabilizer group such that $Z(\calG_\calA)= \langle iI, \calS_\calA\rangle$ where $Z(W)$ represents the center of a group $W$. Then the set of code spaces associated to the group $\calG_\calA$ (and $\calA$) can be defined by
\[
         CS(\calA) := CS(\calG_\calA) := \{\calH_\chi:\chi\in\widehat{\calS_\calA}\} 
\]
\noindent where 
\[
\calH_\chi =\{\ket{\psi}\in\calH:R\ket{\psi}=\chi(R)\ket{\psi},\,\forall R\in \calS_\calA\}.
\]
\noindent The definition of $CS(\calA)$ is independent of the choice of stabilizer group $\calS_\calA$ satisfying the condition $Z(\calA)= \langle iI, \calS_\calA\rangle$. If $\calS$ is a stabilizer group then the code space of $\calS$, as it is typically defined, is the code space in $CS(\calS)$ that is associated with the principal character of $\calS$. Such a code space is called the \textit{principal code space} of $\calS$.

A subsystem code is a stabilizer code in which a subset of the logical operators are chosen not to store information~\cite{poulin-subsystem, HiggottBreuckmannPhysRevX.11.031039}. In this case, the code subspace is decomposed as $\calH_\calS = \calH_\calL \otimes \calH_\calG$ where only $\calH_\calL$ stores information and operations on $\calH_\calG$ are ignored. Here the Pauli operators that act trivially on $\calH_\calL$ form the \textit{gauge group} $\calG$ of the code. 

Let $C_{\PauliN}(A)$ and $N_{\PauliN}(A)$ denote the centeralizer and normalizer of the the group $A$ relative to the group $\PauliN$ respectively. These groups are equal if, and only if, $-I\not\in A$.
If $\calG$ is a gauge group then 
\[
Z(\calG) = \langle iI, \calS_\calG\rangle =C_{\PauliN}(\calG)\cap\calG
\]
\noindent and the operators from $\calG$, up to phase factors, either act trivially on $\calH_\calL\otimes\calH_\calG$ or act non-trivially on only $\calH_\calG$. Logical operators that act non-trivially only on $\calH_\calL$ are called \textit{bare} logical operators and are described by $C_{\PauliN}(\calG)\backslash\calG=N_{\PauliN}(\calG)\backslash{Z(\calG)}=N_{\PauliN}(\calG)\backslash{\calG}$. The \textit{dressed} logical operators, that is bare logical operators that are multiplied by a gauge operator in $\calG\backslash Z(\calG)$, are described by $C_{\PauliN}(\calS_\calG)\backslash\calG$ and act non-trivially on both $\calH_\calL$ and $\calH_\calG$. 

We now introduce some notation that will be needed later. We can represent stabilizer subgroups of $\PauliN$ by minimal generating sets which in turn can be represented as a symplectic binary matrix via the map $\kappa$:
\[
\langle G_1,G_2,\dots,G_k \rangle \xrightarrow{\kappa} [\delta|\sigma] \in \bbF_2^{v\times 2n}
\]
\noindent where $G_j=X_1^{\delta_{j1}}Z^{\sigma_{j1}}\otimes\cdots\otimes X_n^{\delta_{jn}}Z_n^{\sigma_{jn}}$, with $\delta,\sigma\in\bbF_2^{v\times n}$. Since $G_1,G_2,\dots,G_k$ is a minimal generating\footnote{Since $\PauliN$ is a $2$-group the Burnside basis theorem tells us that minimal generating sets (that is, generating sets for which no proper subset also generates $\PauliN$) all have the same cardinality~\cite{MCDOUGALLBAGNALL2011332}} set we have that
$[\delta,\sigma]\in\bbF_2^{v\times 2n,v}$ where $\bbF_2^{v\times 2n, v}$ is the set of $v\times 2n$ binary matrices of rank $v$.

\subsection{LP and CS Equivalent Codes}

Our aim is to find equivalence class representatives for each stabilizer subspace code subject to LC equivalence from which we can then expand that to find the subsystem codes. To do this we first need to look at code set (CS) equivalence and local Pauli (LP) equivalence.

Two different stabilizer groups can generate the same code set and thus those two groups are naturally considered equivalent. In particular, we can use the equivalence relation $\sim_{cs}$ that is defined by $\calA\sim_{cs}\calB$ for $\calA,\calB\leq\PauliN$ if, and only if, $CS(\calA)=CS(\calB)$ to partition the set of stabilizer groups. This is called code set (CS) equivalence and it is equivalent to LP equivalence. That is $\calA\sim_{cs}\calB$ if, and only if, there exists a $W\in\PauliN$ such that $\calA=W\calB W^{-1}$. See Lemma~\ref{appendix:lemma:LPeqSE} in the appendix. Note that the symplectic matrix representation of two CS/LP equivalent codes is the same. While CS/LP equivalence is not that useful for reducing the number of codes we will use CS/LP equivalence to help us count the number of LC equivalent stabilizer code classes. To do this we first note that counting the number of LP equivalence classes of stabilizer codes is known. Let $NLP_{n,k}$ denote the number of different $[[n,k]]$ codes under LP equivalence. Then by Theorem~20 of~\cite{G06} we have that 
\begin{equation}
NLP_{n,k} = \gcoeff{n}{k}_2 \prod_{i=0}^{n-k-1}(2^{n-i}+1) = \prod_{i=0}^{k-1}\frac{2^{n-i}-1}{2^{k-i}-1}\prod_{i=0}^{n-k-1}(2^{n-i}+1) , \label{Nnk}
\end{equation}
\noindent where the square brackets denote the Gaussian coefficients.

\subsection{LC and LU Equivalent Codes}

Both Local Clifford (LC) equivalence and local unitary (LU) equivalence are good equivalence definitions to use when enumerating codes and we know that they are in fact different~\cite{10.5555/2011438.2011446}. Both LC and LU equivalences preserve the distance of codes yet LC equivalence has the nice property that stabilizer groups are mapped to stabilizer groups. Further, the depolarizing noise channel is invariant under the LC equivalence which will be useful when looking for naturally fault-tolerant gates. Of course if you are interested in biased noise then this would be a disadvantage. 

LU equivalence considers two codes different under permutation of qubits thus we can further reduce the number of equivalence classes by including the permutation equivalence ($\Pi$) relation $\sim_{\Pi}$ which is defined as $\calA\sim_{\Pi}\calB$ for $\calA,\calB \leq\PauliN$ if, and only if, there exists some permutation in $\pi\in S_n$ such that $\calA = \calB^\pi$. The union of LC and $\Pi$ equivalence relations will be denoted by LC$\Pi$.

Before defining LC$\Pi$ we introduce the following notation. Since $\Cl{1}/\langle\zeta_8I,\Pauli_1\rangle\simeq  \sym{3}$ we can define the group action of $\sym{3}^n$ on gauge groups of $\PauliN$ in the following way. Let $W\in \sym{3}^n$ and let $C_W$ be the corresponding operator in $(\Cl{1}/\langle\zeta_8I,\Pauli_1\rangle)^{\otimes n}$ under the isomorphism $\Cl{1}/\langle\zeta_8I,\Pauli_1\rangle\simeq  \sym{3}$. Now define the action of $W$ on a gauge group via the conjugation action of $W_C$. The combined group action 
\[
(W,\pi)\calL = W\calL^\pi
\]
\noindent then turns the set $\sym{3}^n\times \sym{n}$  into the group $\sym{3}^n \rtimes \sym{n} = \sym{3}\wr\sym{n}$.

\begin{definition}[LC$\Pi$ Equivalence for Gauge Groups]
Let $\calG$ and $\calL$ be two gauge groups of $\PauliN$. The groups $\mathcal{G}$ and $\mathcal{L}$ are said to be locally permutation Clifford (LC$\Pi$) equivalent, written as $\calG\sim \calL$ if, and only if,
\[
\calG = (W, \pi) \calL
\]
\noindent for some $(W, \pi)\in \sym{3}\wr \sym{n}$
\end{definition}

Our computations do not work directly with gauge groups or stabilizer groups but instead on symplectic representations of those groups. As such the above equivalence definition must be written in terms of symplectic matrices. For this we need three group actions, melded into one, that represent the action of matrices representing a local Clifford conjugation ($\sym{3}^n$), the action of matrices representing a change of basis ($\textrm{GL}_k(\bbF_2)$) (since using a symplectic matrix representation requires fixing a basis) and the action of matrices representing a permutation of qubits ($\sym{n}$). For our purpose the set $\sym{n} \times \sym{3}^n \times \textrm{GL}_k(\bbF_2)$ can be made into the suitable group $\textrm{GL}_k(\bbF_2) \times \sym{3}^n \rtimes_\psi \sym{n}$ and its action on $\bbF_2^{v\times 2n,v}$ can then be defined as
\[
((M,G);\pi) \Gamma := (M, G) (\pi\Gamma) = M \Gamma [\pi^{-1}] [G]
\]
\noindent where $[\pi^{-1}]$ and $[G]$ are the matrix representations of $\pi^{-1}$ and $G$ in $\textrm{GL}_{2n}(\bbF_2)$. The details on these defintions and actions can be found in the Appendix. We can now define LC$\Pi$ for symplectic matrices as follows:

\begin{definition}[LC$\Pi$ for Symplectic Matrices]
Let $\Gamma$ and $\Omega$ be matrices of $\bbF_2^{k\times 2n, v}$. Then $\Gamma$ and $\Omega$ are said to be locally permutation Clifford equivalent, written as $\Gamma\sim\Omega$ if, and only if, 
\[
\Gamma = ((m,G);\pi) \Omega,
\]
\noindent for some $((m,G);\pi) \in (\textrm{GL}_k(\bbF_2) \times \sym{3}^n) \rtimes_\psi \sym{n}$.
\end{definition}

\noindent It now follows that if $\mathcal{A},\mathcal{B}\leq\mathcal{P}_n$ then $\mathcal{A}\sim\mathcal{B}$ if, and only if, $\kappa(\mathcal{A})\sim\kappa(\mathcal{B})$ where $\kappa(\calA)$ is the symplectic matrix representation of $\calA$ relative to the standard basis for $\PauliN$.

\section{Mass Formulas}\label{sec:mass-formulas}

It is an open problem to be able to calculate analytically the number of equivalence classes for the LC$\Pi$ equivalence. However, we still wish to be able to verify that any search algorithm has indeed found the correct number of equivalence classes. To do this we make use of a mass formula relative to the quantity $NLP_{n,k}$ which has been calculated analytically in equation~\eqref{Nnk}. A mass formula is an expression of the form
\begin{equation}
M(n) = \sum_\calC \frac{|\calT|}{|\aut{\calC}|}
\end{equation}
that counts the number of combinatorial objects of size $n$ by summing over the set of equivalence classes of those objects. Here $\calT$ is the full group of allowed transformations that are used to define equivalence, and $\aut{\calC}$ is the automorphism group of an object $\calC$ relative to $\calT$. In our case the group of allowed transformations $\calT$ is $\sym{3}^n \rtimes_{\psi} \sym{n} = \sym{3}\wr \sym{n}$ and
\[
\aut{\calG} = \{\phi\in \sym{3}\wr \sym{n}: \phi\calG=\calG\}.
\]

\noindent where $\calG\leq\PauliN$. Let $\mathcal{SP}_{n,k}$ denote the set of $n$-qubit stabilizer groups that generate $[[n,k]]$ stabilizer codes. Let $\sym{3}\wr\sym{n}\db\mathcal{SP}_{n,k}$ denote the set of orbits (LC$\Pi$ equivalence classes) of $\mathcal{SP}_{n,k}$ under the action of $\sym{3}\wr\sym{n}$ and let $T_{n,k}$ be a transversal of $\sym{3}\wr\sym{n}\db\mathcal{SP}_{n,k}$ (so a set of equivalence class representatives). Then 
\[
NLP_{n,k} = \sum_{\mathcal{G}\in T_{n,k}}\frac{|\sym{3}\wr\sym{n}|}{|\aut{\mathcal{G}}|}
= \sum_{\mathcal{G}\in T_{n,k}}\frac{6^nn!}{|\aut{\mathcal{G}}|}.
\]

\noindent By comparing these two expressions for $NLP_{n,k}$ we have a way to determine if the equivalence count or the size of the automorphism groups are incorrect. If we are confident in the group size calculations then we also have a verification on the equivalence count.

\section{Canonical Representations}\label{sec:canonical-representations}

To calculate a transversal of $\sym{3}\wr\sym{n}\db\mathcal{SP}_{n,k}$ we need an efficient algorithm to determine if two given codes are equivalent. There are two standard approaches to this problem. You can ether develop an algorithm that will determine if two given codes are equivalent or you can develop an algorithm that calculates a canonical representation for each given code and then compare canonical representations. In this paper we use the canonical representation approach and do this by transferring the problem to graphs such that two codes are equivalent if, and only if, the corresponding graphs are isomorphic -- as is done for classical linear codes~\cite{KO06}. Then it is possible to use well-developed algorithms to test for graph isomorphism  \cite{MP13} to determine if two codes are isomorphic.

We use notations and definitions similar to those used in~\cite{MP13}. Let $\mathcal{G}_m$ denote the set of graphs with vertex set $V=\{1,2,...,m\}$.
A coloring of a graph $V$ is a surjective function $c$ from $V$ onto $\{1,2,...,l\}$ for some $l$. Let $\Pi_m$ be the set of colorings of $V$. The group $\sym{m}$ will act on $V$ and its induced structures via exponentiation so that for a $v\in V$ and a $g \in\sym{m}$, $v^g$ is the image of $v$ under $g$ and $v^{g_1g_2}=(v^{g_1})^{g_2}$. Two useful properties are 
\begin{enumerate}
    \item If $c$ is a coloring of $V$, the $c^g$ is the coloring with $c^g(v^g)=c(v)$.
    \item If $(\calG, c)$ is a colored graph,  then $(\calG,c)^g = (\calG^g, c^g)$.
\end{enumerate}

\noindent Two coloured graphs $(\calG,c_\calG)$ and $(\mathcal{H}, c_\mathcal{H})$ are isomorphic if there is some $g\in\sym{m}$ such that $(\mathcal{H},c_\mathcal{H})=(\calG,c_\calG)^g$. Such a $g$ is called an isomorphism. The automorphism group $\aut{\calG,c}$ is the group of isomorphisms of the colored graph $(\calG,c)$ onto itself; that is
\[
\aut{\calG,c} = \{ g\in\sym{m}: (\calG,c)^g = (\calG,c)\}.
\]
\noindent  A canonical form is a function 
$\mathcal{C}:\mathcal{G}_m \times \Pi_m\rightarrow \mathcal{G}_m \times \Pi_m$ such that for all $G\in\mathcal{G}$, $c\in\Pi_m$ and $g\in\sym{m}$,
\begin{enumerate}
    \item $\mathcal{C}(\calG,c)\simeq (\calG,c)$
    \item $\mathcal{C}(\calG^g,c^g)=\mathcal{C}(\calG,c)$.
\end{enumerate}

\noindent The canonical form constructed in~\cite{MP13} is used to determine if two codes are equivalent and we use the nauty software package by the same authors to compute this canonical form. To do this we need to map a code into a graph. Given any subgroup $\calA$ on of $\PauliN$, we construct a vertex-colored graph $\Gamma(\calA)$ in the following way. Let $c_\calA$ be a bi-coloring where we identify black as $1$ and white as $2$. Let $\xi$ be any bijection from the set $\calA$ to $\{1,2,...,|\calA|\}$. For each of the elements $G\in \calA$ create a black vertex $v_{G}=\xi(A)$ of $\Gamma(\calA)$. For each coordinate $j\in [n]$, construct the fully connected triangle in $\Gamma(\calA)$ on white vertices 
\begin{align*}
    u_{j,X}&=t+3(j-1)+1,\\
    u_{j,Y}&=t+3(j-1)+2, \\
    u_{j,Z}&=t+3(j-1)+3.
\end{align*} 
\noindent Finally, create the edge $(v_A,u_{j,P})$ of $\Gamma(\calA)$ between a black and white vertex if the $j$-th coordinate of $G\in \calA$ is $P\in\calP_1$.  Figure~\ref{fig:gammaa} depicts $\Gamma(\calA)$ for the subgroup $\calA=\langle X_1, X_2Y_3\rangle\leq\Pauli_3$.

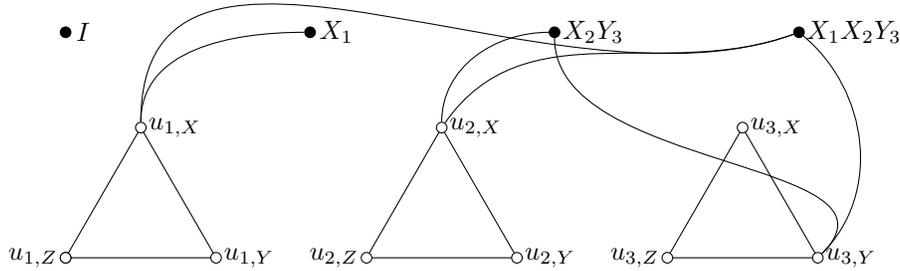
\begin{figure}[h]
    \centering
\begin{tikzpicture}

    \tikzstyle{every node}=[draw,circle,fill=white, minimum size=4pt,
                            inner sep=0pt]
    \draw[fill=white] (0,0) circle[radius=2pt] node (u1z) [label=left:$u_{1,Z}$] {}
        -- ++(0:2.0cm) node (u1y) [label=right:$u_{1,Y}$] {}
        -- ++(120:2.0cm) node (u1x) [label=right:$u_{1,X}$] {}
        -- (u1z);
        
    \draw (4,0) node (u2z) [label=left:$u_{2,Z}$] {}
        -- ++(0:2.0cm) node (u2y) [label=right:$u_{2,Y}$] {}
        -- ++(120:2.0cm) node (u2x) [label=right:$u_{2,X}$] {}
        -- (u2z);
        
    \draw (8,0) node (u3z) [label=left:$u_{3,Z}$] {}
        -- ++(0:2.0cm) node (u3y) [label=right:$u_{3,Y}$] {}
        -- ++(120:2.0cm) node (u3x) [label=right:$u_{3,X}$] {}
        -- (u3z);

    \tikzstyle{every node}=[draw,circle,fill=black, minimum size=4pt,
                            inner sep=0pt]
    
    \fill[black] (0,3) circle[radius=2pt] node (I) [label=right:$I$] {};
    
    \fill[black] (3.25,3) circle[radius=2pt] node (X1) [label=right:$X_1$] {};
    
    \fill[black] (6.5,3) circle[radius=2pt] node (X2Y3) [label=right:$X_2Y_3$] {};
    
    \fill[black] (9.75,3) circle[radius=2pt] node (X1X2Y3) [label=right:$X_1X_2Y_3$] {};

    \draw (X1) to [out=180,in=90] (u1x);
    \draw (X2Y3) to [out=180,in=90] (u2x);
    \draw (X2Y3) to [out=270,in=45] (u3y);
    \draw (X1X2Y3) to [out=200,in=90] (u1x);
    \draw (X1X2Y3) to [out=200,in=55] (u2x);
    \draw (X1X2Y3) to [out=320,in=45] (u3y);
    
\end{tikzpicture}
    \caption{Depiction of the graph $\Gamma(\calA)$ where $\calA = \langle X_1, X_2Y_3\rangle$. The integer labels would be $I=1, X_1=2, ...$ and $u_{1,X}=5, u_{1,Y}=6, ...$.}
    \label{fig:gammaa}
\end{figure}

For computational purposes we want the graphs $\Gamma(\calA)$ to be small. Unfortunately, the above definition of $\Gamma(\calA)$ is anything but small. Ideally, the we would like the graph to have $O(d(\calA)+n)$ vertices instead of the $O(|\calA|+n)$ vertices that $\Gamma(\calA)$ as. Here, $d(\calA)$ denotes the group distance of the group $\calA$ where the group distance of $\calA$ is the minimal number of generators of a group. It is a open problem to find such a small graph. We note that 

\begin{lemma}
Let $\calG$ and $\calL$ be gauge subgroups of $\PauliN$. Then $\calG\sim\calL$ if, and only if, $\Gamma(\calG)\simeq\Gamma(\calL)$ if, and only if, $\Gamma(\calG^+)\simeq\Gamma(\calL^+)$ where $\calG^+$ and $\calL^+$ are any two subgroups of $\PauliN$ such that $\langle \calG^+, iI\rangle = \calG$, $\langle \calL^+, iI\rangle = \calL$ with the elements of $\calG^+$ and $\calL^+$ all having order two.
\end{lemma}
\begin{proof}
Any isomorphism between the two graphs must preserve the coloring and basic structures such as the triangles. It follows that an isomorphism either permutes that black vertices, permutes the white vertices within a triangle or permutes the triangles. Permuting the white vertices has no effect on the group as they are not ordered. Permuting the white vertices within a triangle is equivalent to a local Clifford operation on the qubit associated to that triangle and permuting the triangles is equivalent to a permutation on the qubits. The first result follows from these observations. The last result follows from the observation that the set $\calG$ is a disjoint union of the sets $\calG^+, (-1) \cdot\calG^+, i\cdot\calG^+$ and $(-i)\cdot\calG^+$.
\end{proof}

The graph $\Gamma(\calG^+)$ has $1/4$ the number of black vertices and $1/4$ number of edges as compared to the graph $\Gamma(\calG)$. 

\section{CWS codes}\label{sec:stabcws}

We recall from~\cite{CSSZ09} that a Code Word Stabilized (CWS) code is completely specified by a graph state $|\calG\rangle$ and a classical code $\mathcal{C}$ and that those CWS codes that are also stabilizer codes correspond to when the classical code $\mathcal{C}$ is a linear code. When combined with Theorem~4 of~\cite{CSSZ09}, which states that the CWS framework includes all stabilizer codes, we have a way to construct all stabilizer codes from CWS codes. To do this we need to be able to transform a given CWS code that represents a stabilizer code into a stabilizer code (under local Clifford equivalence). There are efficient procedures to transform a stabilizer code to or from it's representation as a locally equivalent CWS code in standard form. We briefly outline both of these procedures. 

Suppose we are given an $[[n,k]]$ stabilizer code $S=\langle S_1, \dots, S_{n-k}\rangle$ that we wish to convert to a CWS representation. The encoded state $|\bar{0}^k\rangle$ is stabilized by $\langle S_1, \dots, S_{n-k}, \bar{Z}_1, \dots, \bar{Z}_k\rangle$, and the code space is spanned by the states $\bar{X}_j|\bar{0}^k\rangle$ for $j=1, \dots, k$. There is a local Clifford operator $U$ that maps $|\bar{0}^k\rangle$ to a graph state $|\calG\rangle$. The corresponding word operators are equivalent to $U\bar{X}_jU^\dagger$ for $j=1, \dots, k$ and can be reduced to CWS standard form if so desired.

Going the other direction, suppose we have an $[[n,k]]$ stabilizer code presented as a CWS code in standard form with graph state $|\calG\rangle$ and linear code $\mathcal{L}$. Our goal is to compute a generating set for the stabilizer group of the code. The following algorithm follows from Theorem~4 and Theorem~5 of~\cite{CSSZ09}.  Suppose that $S=\langle g_1,...,g_n\rangle$ defines the graph state $|\calG\rangle$ and that $\mathcal{L}$ is represented as a $k\times n$ matrix over $GF(2)$ with rows $l_1,...,l_k$. For each generator $l_i$ create the associated $Z$-type operator $r_i=Z^{l_i}$. Each $r_i$ anticommutes with at least one $g_j$ and so let $\Omega_i= \{g_{i_1}, g_{i_2}, ..., g_{i_t}\}$ be the set of all generators of $S$ that anticommute with the operator $r_i$. Then for each $g_{i_v}\in\Omega_i$ with $v\not=1$ set $g_{i_v} = g_{i_v}g_{i_1}$  and then set $g_{i_1}=0$. After running through all the linear code generators the stabilizer group of the code will be generated by those $g\in S$ that are non-zero.

\section{Search Strategies}\label{sec:search-strategies}

We approach the problem of finding and enumerating equivalent stabilizer groups in two different ways. The first approach iteratively constructs stabilizer groups and is suited to searching over groups that generate codes that encode many qubits. The second approach uses the codeword stabilized (CWS) code formalism \cite{CSSZ09} and is well-suited to stabilizer groups that generate codes encoding few qubits. Each approach yields a set of equivalence classes of $[[n,k]]$ stabilizer groups for some $n$ and $k$. From these groups the associated code sets can be generated and any associated gauge groups. Although it is not necessary, both approaches may be used together to find all of the equivalent $n$-qubit codes.

The first approach begins with codes that encode many qubits with distance one. There is exactly one $[[n,n,1]]$ principal code and it has a trivial stabilizer group. However, we can begin the search one step later by observing that there are $n$ classes of $[[n,n-1,1]]$ principal codes whose stabilizer groups are generated by a single element of each positive Hamming weight. We now iteratively generate $[[n,k-1]]$ stabilizer groups by appending stabilizers to representatives of each equivalence class of $[[n,k]]$ stabilizer groups and rejecting equivalent groups. For each $[[n,k-1]]$ principal code, we consider no more than $4^n-1$ Pauli operators to add to its stabilizer group. Operators that already belong to the stabilizer group or anticommute with one of its operators do not need to be considered. We test each candidate group against the set of discovered objects and accept the group if it is not equivalent to any previously discovered object.

For the second approach, recall that every stabilizer code is equivalent to a CWS code in standard form whose word operators form a classical linear code. Therefore, it is sufficient to iterate over all pairs $(G,C)$ of graph states and classical linear codes and reject equivalent codes. We reduce the search space by using Danielsen’s classification of graph states on up to 12 qubits \cite{DP06}. Our implementation naively iterates over classical linear codes. It is an open problem to find more efficient ways to iterate over classical codes $\cal{C}$.

\section{Code Properties}\label{sec:code-properties}

In this section we discuss various properties of codes that are calculated for each of the equivalence classes of code/stabilizer groups.
\subsection{Decomposability}
A stabilizer group $\calS\leq\PauliN$ is said to be \textit{strictly decomposable} if it can be factored into a tensor product of two stabilizer groups $\calA\leq\Pauli_{n_1}$ and $\calB\leq\Pauli_{n_2}$ with $n=n_1+n_2$ and $\calS=\calA\otimes\calB$. If a stabilizer group is strictly decomposable then codes in its code set are said to be strictly decomposable. If there exists a permutation $\pi\in\sym{n}$ such that $\calS^\pi$ is strictly decomposable then $\calS$ is said to be \textit{decomposable}. Not all decomposable stabilizer groups are strictly decomposable. For example, the stabilizer group (which corresponds to the principal 3-qubit bit-flip code with an extra qubit appended in the $|+\rangle$ state):
\[
\calS=\langle IIIX, ZZII, IZZI\rangle=\langle ZZI, IZZ\rangle\otimes \langle X\rangle
\]
\noindent is strictly decomposable where as the permutation equivalent group 
$\calS^\pi=\langle IXII, ZIZI, IIZZ\rangle$
\noindent with $\pi=(342)$ is decomposable but not strictly decomposable. Note the strict decomposabilty is invariant under conjugation by local Cliffords and so decomposability is invariant under permutation and conjugation by local Cliffords.

Let $\calS$ be a stabilizer group with symplectic generator matrix $G$ representing a $[[n,k]]$ code set. Extending~\cite{G97, nielsen00} we say that a symplectic generator matrix is in Revised Standard Form (RSF) if $G$ has the form
\begin{equation}
\left[\begin{array}{cccc|cccc}
I & A_1 & A_2 & 0 & B & 0 & C & 0 \\
0 & 0   & 0   & 0 & D & I & E & 0\end{array}\right].
\end{equation}
\noindent where the top row is of rank $r$, and the bottom row is of rank $n-k-r$, and the zeros represent blocks of $zeros$ (that are possibly empty). We note that the process to create the revised standard form involves a permutation of qubits, which may be the identity permutation, and so the stabilizer group for the revised standard form is permutation equivalent to the starting stabilizer group and therefore may not be equal to the original stabilizer group. We extend the terminology for revised standard form naturally to an ordered list of minimal generators for a given stabilizer group.

Proposition~\ref{prop:decomp} is a technical lemma which was inspired by~\cite{S60} that shows how a minimal generating set, that is in revised standard from, is structured when the underlying group is decomposable (See Figure~\ref{fig:decomdiag}). The first condition in Proposition~\ref{prop:decomp} is equivalent to saying that the revised standard form matrix has a zero column in the $n$-th and $2n$-th positions. In this case we say that the group, generators or code has \textit{trivial factors}. The second conditions states that, up to a permutation, the generators are products of generators from one of two disjoint sets but not both.

\begin{figure}[h] 
    \centering
    \begin{tikzcd}
        \calL\otimes\calM \arrow[d]
        & \calS  \arrow[l, "\sim_\pi"'] \arrow[d] \arrow[r, "\sim_\rho"]
        & \calG \arrow[d] \\
        \langle \calL\otimes I^{\otimes n_2},I^{\otimes n_1}\otimes\calM\rangle 
        & \langle S_1, S_2,...,S_r\rangle \arrow[r, "\textrm{RSF}"] 
        & \left(\langle G_1, G_2,...,G_r\rangle, \rho\right)
    \end{tikzcd}
    \caption{Let $S$ be a decomposable stabilizer group with minimal generating set $S_1, S_2,...,S_r$. Generating the revised standard form from the minimal generating set produces an new minimal generating list that is in revised standard form as well as the permutation $\rho$ needed to make the transformation. This means that $\calS\sim_\rho\calG$. Since $\calS$ is decomposable there exists some permutation $\pi$ such that $\calS \sim_\pi \calL\otimes \calM$. The group $\calS \sim_\pi \calL\otimes \calM$ can in turn be generated by $\langle \calL\otimes I^{\otimes n_2},I^{\otimes n_1}\otimes\calM\rangle$ for some positive integers $n_1$ and $n_2$ such that $n=n_1+n_2$. The arrows for the permutation equivalence indicate the application of the permutation, thus $\calS\xrightarrow{\sim_\pi}\calG$ means that $\calS^\pi = \calG$.} 
    \label{fig:decomdiag}
\end{figure}

\begin{proposition}\label{prop:decomp}
Let $\calS\leq\PauliN$ be a stabilizer group with minimal generating list $S_1, S_2,...,S_r$. Let $G_1, G_2,...,G_r$ be a minimal generating set that is in revised standard form that was generated from $S_1, S_2,...,S_r$. Let $\rho\in\sym{n}$ be the permutation such that $\calS\sim_\rho\calG$ where $\calG=\langle G_1, G_2,...,G_r\rangle$. Then $\calS$ is decomposable with $\calS\sim_\pi\calL\otimes\calM$, $\calL\leq\Pauli_{n_1}$, $\calM\leq\Pauli_{n_2}$, $\pi\in\sym{n}$ if, and only if, at least one of the following conditions hold:
\begin{enumerate}
    \item $G_j = G_j'\otimes I$ for some $G_j'\in\Pauli_{n-1}$ for each $j\in [n]$ 
    \item $G_j^{\rho^{-1}\pi}\in \calL\otimes\langle I^{\otimes n_2}\rangle$ or $G_j^{\rho^{-1}\pi}\in \langle I^{\otimes n_2}\rangle\otimes \calM$ for each $j\in [n]$ 
\end{enumerate}
\end{proposition}
\begin{proof}
The sufficiency of the conditions follows by definition. To prove the necessity of the conditions we assume without loss of generality that both $\pi$ and $\rho$ are the identity permutations and so $\calG = \calS = \calL\otimes\calM$. If $\calM = \calM'\otimes I$ then the first condition is true. We therefore suppose that $\calM$ is not equal to $\calM'\otimes \langle I^{\otimes t}\rangle$ for any positive integer $t$.

Before proceeding it will be useful to establish some nomenclature. If the X-part of the standard form generator matrix $G$ for $G_1, G_2,...,G_r$ has full rank, then $G$ has the form
\begin{equation}
\left[\begin{array}{cc|cc}
I & A & B & C\end{array}\right].
\end{equation}
\noindent If instead the X-part of the generator matrix is zero, the standard form is
\begin{equation}
\left[\begin{array}{cc|cc}
0 & 0 & I & E \end{array}\right].
\end{equation}
\noindent In both cases, we refer to the qubits in the support of the identity matrix as left qubits and the other qubits as right qubits. Otherwise, the revised standard form is generally (under current assumptions):
\begin{equation}
\left[\begin{array}{ccc|ccc}
I & A_1 & A_2 & B & 0 & C \\
0 & 0   & 0   & D & I & E\end{array}\right].
\end{equation}
\noindent In this general case, we have left qubits in the support of the X-part identity matrix, middle qubits in the support of the Z-part identity matrix, and right qubits are the rest.

Consider the first case where the X-part is full rank. Assume that there is a generator $g$ that is a product of generators from both $\calL$ and $\calM$, so $g=g^{\vphantom{1}}_\calL g^{\vphantom{1}}_\calM$ where 
$g^{\vphantom{1}}_\calL\in \calL\otimes\langle I^{\otimes n_2}\rangle$ and
$g^{\vphantom{1}}_{\calM}\in \langle I^{\otimes n_2}\rangle\otimes \calM$. The standard form implies that $g=X_iX_RZ_LZ_R$ where $X_i$ acts on the $i$th (left) qubit. Assume without loss of generality that the $i$th qubit is in the support of $\calL$ and so $g^{\vphantom{1}}_{\calM}=g^{\vphantom{1}}_\calL g=X'_RZ'_LZ'_R$. However, this is a contradiction because every element of $\calL\otimes \calM$ has X or Y operators on at least one left qubit. Therefore each generator belongs to either
$\calL\otimes\langle I^{\otimes n_2}\rangle$ or $\langle I^{\otimes n_2}\rangle\otimes \calM$. Essentially the same argument holds for the second case where the X-part is zero. Finally, in the general case, the argument is still essentially the same due to the presence of the identity blocks. The left qubits paritition into a set supported on $\calL$ and another set supported on $\calM$, which partitions the top generators into those belonging to $\calL\otimes\langle I^{\otimes n_2}\rangle$ and those belonging to $\langle I^{\otimes n_2}\rangle\otimes \calM$. The middle qubits partition as well, which partitions the bottom generators.
\end{proof}

By repeated application of Proposition~\ref{prop:decomp} we see that if $\calS$ is a stabilizer group then $\calS$ can be in some sense factored into indecomposable factors. This factorization is unique up to permutations of the qubits within each factor and permutation of the factors themselves. From this we can define the \textit{length} of a stabilizer group $\calS$ as the number of non-trivial indecomposable factors of $\calS$. We denote this by the integer $l$. Hence, a stabilizer group $\calS$ with no trivial factors is indecomposable if, and only if, $l(\calS)=1$. We now show how $l$ can be calculated.

Let $\calS\leq\PauliN$ be a stabilizer group with minimal generating list $S_1, S_2,...,S_r$. Let $G_1, G_2,...,G_r$ be a minimal generating set that is in revised standard form that was generated from $S_1, S_2,...,S_r$ with $\calG = \langle G_1, G_2,...,G_r \rangle$. If $\calG$ has trivial factors then it is decomposable.  Suppose that $\calG$ does not have trivial factors or if we want to know if the non-trivial factor is decomposable. In this case we then construct a graph $\Lambda=\Lambda(G_1, G_2,...,G_r)$. The vertices of $\Lambda$ correspond to the generators $G_j$. A pair of vertices are connected if the corresponding generators have a qubit in the intersection of their supports. Note that if we had generated a graph from an different revised standard form then the two graphs would be isomorphic. We can define $\Lambda(\calS)$ as some representative of this class of isomorphic graphs. With this notation we see that $\textrm{len}(\calS)=cc(\Lambda(\calS))$ where $cc(\Lambda(\calS))$ is the number of connected components of the graph $\Lambda(\calS)$.

\begin{corollary}
If $\calS$ is a stabilizer group then $\mathrm{len}(\calS)=cc(\Lambda(\calS))$.
\end{corollary}

\begin{corollary}
A stabilizer group $\calS$, with no trivial factors, is decomposable if, and only if, the graph $\Lambda(\calS)$ is disconnected.
\end{corollary}

\subsection{CSS Codes}

A stabilizer group $\calS\leq\PauliN$ is said to be \textit{CCS} if there exists a set of generators for $\calS$ that can be partitioned into only $X$-type stabilizers or $Z$-type stabilizers. If $\calS$ is CSS then so is $\calS^\pi$ for any $\pi\in\sym{n}$ and so the CSS property is invariant with respect to permutation equivalence. The CSS property of the group is also invariant under LP equivalence. Thus we can naturally extend the definition of CSS to code sets and individual codes. The CSS property is not however invariant with respect to LC equivalence and thus not invariant for LC$\Pi$ equivalence. Instead, we say that a LC$\Pi$ equivalence class is CSS if, and only if, there exists some LC$\Pi$ equivalent group that is CSS.

Given a list of $n-k$ independent stabilizer generators, we would like to decide if the corresponding stabilizer group is CSS, and if so, find a set of $n-k$ independent generators for the group in CSS form. The following lemma, thanks to Theodore Yoder, provides a way to quickly recognize stabilizer groups that are CSS.

\begin{lemma}\label{lem:TedCSS}
Let $G=\left[G_X|G_Z\right]$ be the generator matrix of stabilizer group. The group is CSS if, and only if,
\begin{equation*}
\mathrm{rank}(G_X) + \mathrm{rank}(G_Z) = n-k.
\end{equation*}
\end{lemma}

\begin{proof}
If the operators defined by $G$ generate a CSS group, then there is another generator matrix $G'$ with $n_x$ independent X-type generators and $n_z$ independent Z-type generators such that $n_x+n_z=n-k$. Therefore,
\begin{align*}
\mathrm{rank}(G_X) + \mathrm{rank}(G_Z) & = \mathrm{rank}(G'_X) + \mathrm{rank}(G'_Z) \\
& = n_x + n_z = n-k.
\end{align*}
Conversely, suppose $G$ satisfies
\begin{equation*}
\mathrm{rank}(G_X) + \mathrm{rank}(G_Z) = n-k.
\end{equation*}
Reduce $G$ to revised standard form 
\begin{equation}
G' := \left[\begin{array}{ccc|ccc}
I & A_1 & A_2 & B & 0 & C \\
0 & 0   & 0   & D & I & E\end{array}\right].
\end{equation}
As the CSS property is preserved under permutations we may do this and maintain the CSS property. The matrix $G'$ has $n-k$ rows. There are $\mathrm{rank}(G_X)$ non-zero rows of the left submatrix since these rows are clearly linearly independent. The remaining rows are purely Z-type stabilizers, and there are $\mathrm{rank}(G_Z)=n-k-\mathrm{rank}(G_X)$ of these. Therefore the rows of the submatrix $B0C$ must be linear combinations of the rows of the $DIE$ submatrix. We can take suitable linear combinations to reduce the first $\mathrm{rank}(G_X)$ rows to be purely X-type stabilizers. We conclude that the group is CSS.
\end{proof}

Not all stabilizer codes are CSS or equivalent to a CSS code. For example the codes associated with the group $\langle XZ, ZX\rangle$ are not CSS but are equivalent to a CSS code and the codes associated with the group $\langle ZIXZ, YXYI, IZZX\rangle$ are not CSS and are not equivalence to any CSS code. Currently we do not have a fast way to determine if a given LC$\Pi$ equivalence class is CSS or not. We can find a CSS representative of an equivalence class, if one exists, by iterating over all $6^n$ single-qubit Clifford equivalent groups. For each, we test the rank equation of Lemma~\ref{lem:TedCSS}. If a group satisfies the rank equation, its standard form generator matrix is in CSS form.

\subsection{GF(4)-linearity} 

Stabilizer codes can be viewed as additive codes over the Galois field GF(4) \cite{CRSS98}. We represent GF(4) as the quadratic extension $\bbF_2[\omega] = \{0, 1, \omega, \omega^2\}$ where $\omega^2+\omega+1=0$. There is a group homomorphism $\lambda:{\calP}_1\rightarrow \textrm{GF(4)}$ from the single-qubit Pauli group to GF(4) as an additive group. It is defined by setting $\lambda(i^a X)=1$ and $\lambda(i^a Z)=\omega$ and extending to all elements by $\lambda(P_1P_2)=\lambda(P_1)+\lambda(P_2)$. The homomorphism extends to the $n$-qubit Pauli group by acting on each term in the tensor product. That is $\lambda(A_1\otimes A_2\otimes\dots\otimes A_n) = (\lambda(A_1),\dots,\lambda(A_n))\in\textrm{GF}(4)^n$. In some cases, stabilizer codes correspond not just to additive but to linear codes over GF(4). 

\begin{definition}
A stabilizer group/code is GF(4)-linear if the image of the stabilizer group under $\lambda$ is a linear code, i.e. if it is closed under addition and scalar multiplication by elements of GF(4).
\end{definition}

Let $R=HS$ be the Clifford gate that acts as the permutation $(XYZ)$ via conjugation. The action of $R$ and its inverse by conjugation on the Pauli group corresponds to multiplication by $\omega^2$ and $\omega$, respectively, in GF(4). For example,
\begin{equation}
\lambda(RPR^{-1}) = \omega^2\lambda(P).
\end{equation}
The next lemma gives us a way to recognize GF(4)-linear stabilizer codes using the language of stabilizers.

\begin{lemma}
A stabilizer code is GF(4)-linear if and only if $R^{\otimes n}$ maps the stabilizer group to itself by conjugation.
\end{lemma}\label{lem:gf4}

\begin{proof}
Suppose we have a GF(4)-linear code with stabilizer group $\calS$. Since conjugation by $R$ corresponds to multiplication by $\omega^2\in\textrm{GF(4)}$, we see that $R^{\otimes n}\calS(R^{-1})^{\otimes n}=\calS$. Conversely suppose that $R^{\otimes n}\calS(R^{-1})^{\otimes n}=\calS$. Then for any $j$, 
\[
\omega^j\lambda(\calS)=\lambda((R^{-j})^{\otimes n} \calS (R^j)^{\otimes n})=\lambda((R^{2j})^{\otimes n}\calS(R^{-2j})^{\otimes n})=\lambda(\calS). \tag*{\qedhere}
\]
\end{proof}

\begin{corollary}
    If $S$ is LC$\Pi$ equivalent to an GF(4)-linear stabilizer group of type $[[n,k]]$ then $n\equiv k\pmod{2}$.
\end{corollary}
\begin{proof}
Let $S'$ be the GF(4)-linear LC$\Pi$ representative of $S$. Then by Lemma~\ref{lem:gf4} $S'$ is invariant under conjugation by $R^{\otimes n}$. Thus, if $g\in S'$ then the orbit of $g$ under conjugation by $R^{\otimes n}$ is of order 3 if $g$ is not the identity and is of order 1 if $g$ is the identity. Hence $|S'|\equiv 1\pmod{3}$. But $|S'|=|S|=2^{n-k}$ and so $2^{n-k} \equiv 1 \pmod{3}$. It follows by Euler's theorem~\cite{euler_theoremata_1763} that $n-k\equiv 0\pmod{2}$.
\end{proof}

The property of GF(4)-linearity is preserved under qubit permutations ($\Pi$) but it is not preserved under local Clifford conjugation (LC). However, if a $GF(4)$-linear representative exists, we can find it by iterating over all $6^n$ single-qubit local Clifford equivalent stabilizer groups. For each, we test if $R^{\otimes n}$ takes stabilizer generators to stabilizers by repeated conjugation. If the test passes for one of the LC equivalent stabilizer groups, we have found a $GF(4)$-linear representative. We will call an equivalence class GF(4)-linear at least one of the representatives of that class is GF(4)-linear.

Suppose we have a stabilizer group that is LC equivalent to a GF(4)-linear code by local equivalence $U_1$
and also to a CSS code by LC equivalence $U_2$. The following lemma tells us that finding a GF(4)-linear representative won't automatically give us the CSS representative.

\begin{lemma}
There is a non-CSS GF(4)-linear stabilizer group that is LC equivalent to a CSS group.
\end{lemma}

\begin{proof}
Consider any self-orthogonal CSS code with stabilizer group $\calS=\langle G_1,\dots,G_m\rangle$ such that each $G_i$ is either $X$-type
or $Z$-type. By definition, $\calS$ is GF(4)-linear. We obtain a locally equivalent
code by conjugating the generators by the Clifford gate $R_i$ where $i$ is any coordinate in the support of some generator $G_j$. Without loss of generality, suppose $G_j$ is an X-type Pauli operator. Then $R_iG_jR_i^{-1}=-iZ_iG_j$ is in the image of the generating set. If we assume that $\calS$ is non-trivial in the sense that $Z_i$ is not in the stabilizer, then the locally equivalent code $R_iSR_i^{-1}$ is not CSS. However, conjugating the stabilizer of a GF(4)-linear code by $R_i$ produces another GF(4)-linear code. Indeed,
\begin{equation}
R^{\otimes n}(R_i\calS R_i^{-1})(R^{-1})^{\otimes n} = R_iR^{\otimes n}S(R^{-1})^{\otimes n}R_i^{-1} = R_iSR_i^{-1},
\end{equation}
so the locally equivalent code is GF(4)-linear.
\end{proof}

Suppose we have a GF(4)-linear code that is locally equivalent (LC) to a CSS code. The following
lemma implies that we can put a GF(4)-linear code into CSS form by applying local
Clifford gates and the resulting code will still be GF(4)-linear.

\begin{lemma}
Any CSS code that is locally equivalent (LC) to a GF(4)-linear code is GF(4)-linear.
\end{lemma}

\begin{proof}
Let $\calS$ be the stabilizer of a CSS code and let $U$ be a local Clifford operator
that maps $\calS$ to a GF(4)-linear code $\calL$. Take any element $A\in \calS$. It's image
in the GF(4)-linear code is $B=UAU^{-1}$. Let $\bar{R}=R^{\otimes n}$ and then repeatedly conjugating $B$ by $\bar{R}$
we generate two additional elements $B'=\bar{R}B\bar{R}^{-1}$ and $B''=\bar{R}^{-1} B\bar{R}$ in $\calL$.
All of $B$, $B'$, and $B''$ have the same support as $A$. Therefore, there are
three distinct operators $A$, $A'$, and $A''$ in $\calS$ that have the same support.
Furthermore,
\begin{equation}
A'=U^{-1} B'U=U^{-1} \bar{R} B \bar{R}^{-1} U=\left(U^{-1} \bar{R}U\right) A \left(U^{-1} \bar{R}^{-1} U\right), 
\end{equation}
so the Clifford operation mapping $A$ to $A'$ to $A''$ is Clifford-conjugate to $R$
on each qubit. Suppose without loss of generality that $A$ is X-type, then $A'$ is
a tensor product of $Y$ and $Z$ on the support of $A$. Since $\calS$ is CSS,
$A'$ must be a product of $X$ and $Z$ type operators. Therefore, there is a Z-type
operator with the same support as $A$. Hence the corresponding classical code is
self-orthogonal and $\calS$ is GF(4)-linear.
\end{proof}

\section{Results}\label{sec:results}

\subsection{Stabilizer code database}

The search algorithms for LC$\Pi$ equivalence classes for $n<10$ yielded a total of $754,006$ classes and details of the breakdown can be found in Table~\ref{table:totalcount}. Table~\ref{table:totalcountfine} further breaks these totals down into the number of equivalence classes for each $[[n,k]]$ pair. Of the $754,006$ equivalence classes we find that $686,904$, or $91\%$, are indecomposable as shown in Table~\ref{table:totalcountfineindecom}. A generating set is said to be even in weight if all generators have even weight/support. Evenness is invariant under LC$\Pi$ equivalence and so a class is called even if its generating sets are even. About $3\%$, or 21,273, of the 686,904 indecomposable classes are even. Of those there are only $5121$ of distance $2$. Of the $686,904$ indecomposable equivalence classes, $271,910$, or $39.6\%$, encode at least one qubit and have either distance $2$ or $3$ as shown in Table~\ref{table:indecomp_kg0_dg1}. Of these $271,910$ equivalence classes only $4,328$ are CSS with only $21$ being distance $3$ as shown in Table~\ref{table:indecompcss_kg0_dg1}. The smallest code that is not equivalent to any CCS code is the indecomposable code $\eclass{4}{1}{2}{8}$ and the smallest error correcting code that is not equivalent to any CSS code is the indecompoasble code $\eclass{5}{1}{3}{21}$.

Of the $686,904$ indecomposable codes only $24$ are GF(4)-linear. A representative from each GF(4)linear class are given in Table~\ref{tab:all_gf_lin_indecom}. A detailed count table for GF(4)-linear classes can be found in Table~\ref{table:GF4countfineindecom} in the table appendix. The table appendix contains a range of tables detailing the database. Complete tables of the database for $n\leq 6$ are included in the Appendix and include sizes of the automorphism groups which can be used to verify the class count via the Mass formulas. The full database is available through the Qiskit-QEC package~\cite{qiskitqec}. See Section~\ref{qiskit-qec} in the Appendix on how to access the database. 

\begin{table}
	\centering
	\begin{tabular}{c|cccccccccc}
            n & 1 & 2 & 3 & 4 & 5 & 6 & 7 & 8 & 9 \\ 
            \midrule
            $N_n$ & 2 & 5 & 12 & 35 & 112 & 474 & 2757 & 28642 & 721967 \\ 
            $M_n$ & 2 & 2 & 4 & 13 & 46 & 245 & 1765 & 22773 & 662054 \\ 
\end{tabular}
        \captionsetup{width=0.55\linewidth}
	\caption{Number of equivalence classes $N_n$ of $n$-qubit stabilizer codes and the number $M_n$ that are indecomposable.}\label{table:totalcount}
\end{table}

\begin{table}
    \centering
    \begin{tabular}{c|cccccccccc}
        $n\backslash k$ & 0 & 1 & 2 & 3 & 4 & 5 & 6 & 7 & 8 & 9 \\ 
        \midrule
        1 & 1 & 1 & 0 & 0 & 0 & 0 & 0 & 0 & 0 & 0 \\ 
        2 & 2 & 2 & 1 & 0 & 0 & 0 & 0 & 0 & 0 & 0 \\ 
        3 & 3 & 5 & 3 & 1 & 0 & 0 & 0 & 0 & 0 & 0 \\ 
        4 & 6 & 13 & 11 & 4 & 1 & 0 & 0 & 0 & 0 & 0 \\ 
        5 & 11 & 36 & 40 & 19 & 5 & 1 & 0 & 0 & 0 & 0 \\ 
        6 & 26 & 115 & 185 & 109 & 32 & 6 & 1 & 0 & 0 & 0 \\ 
        7 & 59 & 448 & 1075 & 852 & 267 & 48 & 7 & 1 & 0 & 0 \\ 
        8 & 182 & 2371 & 10010 & 11422 & 3963 & 614 & 71 & 8 & 1 & 0 \\ 
        9 & 675 & 20128 & 181039 & 353569 & 146658 & 18445 & 1344 & 99 & 9 & 1 \\ 
    \end{tabular}
    \caption{Number of equivalence classes for each $[[n,k]]$.}\label{table:totalcountfine}
\end{table}

\begin{table}
    \centering
    \begin{tabular}{c|cccccccccc}
        $n\backslash k$ & 0 & 1 & 2 & 3 & 4 & 5 & 6 & 7 & 8 & 9 \\ 
        \midrule
        1 & 1 & 1 & 0 & 0 & 0 & 0 & 0 & 0 & 0 & 0 \\ 
        2 & 1 & 1 & 0 & 0 & 0 & 0 & 0 & 0 & 0 & 0 \\ 
        3 & 1 & 2 & 1 & 0 & 0 & 0 & 0 & 0 & 0 & 0 \\ 
        4 & 2 & 6 & 4 & 1 & 0 & 0 & 0 & 0 & 0 & 0 \\ 
        5 & 4 & 17 & 18 & 6 & 1 & 0 & 0 & 0 & 0 & 0 \\ 
        6 & 11 & 63 & 107 & 53 & 10 & 1 & 0 & 0 & 0 & 0 \\ 
        7 & 26 & 284 & 754 & 556 & 131 & 13 & 1 & 0 & 0 & 0 \\ 
        8 & 101 & 1767 & 8328 & 9417 & 2834 & 306 & 19 & 1 & 0 & 0 \\ 
        9 & 440 & 17143 & 167595 & 331296 & 131035 & 13852 & 668 & 24 & 1 & 0 \\ 
    \end{tabular}
    \caption{Number of equivalence classes of each $[[n,k]]$ that are indecomposable.}\label{table:totalcountfineindecom}
\end{table}

\begin{table}
\centering
\begin{tabular}{cccccccccc}
\multicolumn{1}{c}{} & \multicolumn{6}{c}{distance 2} & \multicolumn{3}{c}{distance 3} \\
$n\backslash k$ & 1 & 2 & 3 & 4 & 5 & 6 & 1 & 2 & 3 \\
\cmidrule(lr){2-7}\cmidrule(lr){8-10}
4 & 2 & 1 & - & - & - & -  & - & - & - \\
5 & 4 & 2 & - & - & - & -  & 1 & - & - \\
6 & 27 & 25 & 5 & 1 & - & - & 1 & - & - \\
7 & 128 & 209 & 62 & 6 & - & - & 16 & - & - \\
8 & 964 & 3450 & 2043 & 255 & 11 & 1 & 157 & 20 & 1 \\ 
9 & 9395 & 94048 & 128405 & 23844 & 757 & 12 & 3411 & 4425 & 221 \\
\end{tabular}
\captionsetup{width=0.6\linewidth}
\caption{Number of classes of indecomposable codes of distance 2 (left) and distance 3 (right). \label{table:indecomp_kg0_dg1}}
\end{table}

\begin{table}
\centering
\begin{tabular}{cccccccccc}
\multicolumn{1}{c}{} & \multicolumn{6}{c}{distance 2} & \multicolumn{3}{c}{distance 3} \\
$n\backslash k$ & 1 & 2 & 3 & 4 & 5 & 6 & 1 & 2 & 3 \\
\cmidrule(lr){2-7}\cmidrule(lr){8-10}
4 & 1 & 1 & - & - & - & -  & - & - & - \\
5 & 3 & 1 & - & - & - & -  & - & - & - \\
6 & 12 & 10 & 2 & 1 & - & - & - & - & - \\
7 & 41 & 37 & 13 & 2 & - & - & 1 & - & - \\
8 & 168 & 244 & 114 & 31 & 3 & 1 & 1 & - & - \\ 
9 & 717 & 1475 & 1082 & 305 & 40 & 3 & 19 & - & - \\
\end{tabular}
\captionsetup{width=0.5\linewidth}
\caption{Number of classes of indecomposable CSS codes of distance 2 (left) and distance 3 (right).} \label{table:indecompcss_kg0_dg1} 
\end{table}

\begin{table}
\renewcommand{\arraystretch}{1.25}
\centering
\begingroup
\footnotesize
\begin{tabularx}{\textwidth}{l|l|l|L{34}}
\toprule
$[[n,k]]$ & $\mathrm{Idx}$ & $|\mathrm{Aut}(S)|$ & $S$ \tabularnewline
\specialrule{1.5pt}{1pt}{1pt}
\midrule $[[2,0,2]]$  & 1 & 12 & $\langle Z_{0}Z_{1}$, $X_{0}X_{1}\rangle$ \\ 
\midrule $[[4,2,2]]$  & 9 & 144 & $\langle Z_{0}Z_{1}Z_{2}Z_{3}$, $X_{0}X_{1}X_{2}X_{3}\rangle$ \\ 
\midrule $[[5,1,3]]$  & 21 & 360 & $\langle Y_{0}Y_{1}Z_{2}Z_{3}$, $Y_{0}Z_{1}Y_{2}Z_{4}$, $X_{0}Z_{2}X_{3}Z_{4}$, $X_{0}Z_{1}Z_{3}X_{4}\rangle$ \\ 
\midrule $[[6,0,4]]$  & 19 & 2160 & $\langle Y_{0}Z_{1}Z_{2}Z_{5}$, $Z_{0}Y_{1}Z_{3}Z_{5}$, $Z_{0}Y_{2}Z_{4}Z_{5}$, $Z_{1}Y_{3}Z_{4}Z_{5}$, $Z_{2}Z_{3}Y_{4}Z_{5}$, $X_{0}Z_{3}Z_{4}Y_{5}\rangle$ \\ 
\midrule $[[6,2,2]]$  & 126 & 288 & $\langle Z_{0}Z_{1}Z_{2}Z_{3}$, $X_{0}X_{1}X_{2}X_{3}$, $Z_{0}Z_{1}Z_{4}Z_{5}$, $X_{0}X_{1}X_{4}X_{5}\rangle$ \\ 
\midrule $[[6,4,2]]$  & 29 & 4320 & $\langle Z_{0}Z_{1}Z_{2}Z_{3}Z_{4}Z_{5}$, $Y_{0}Y_{1}Y_{2}Y_{3}Y_{4}Y_{5}\rangle$ \\ 
\midrule $[[7,1,3]]$  & 226 & 1008 & $\langle Z_{0}Z_{1}Z_{3}Z_{6}$, $Z_{0}Z_{2}Z_{3}Z_{5}$, $Y_{1}Y_{2}Y_{3}Y_{4}$, $Z_{3}Z_{4}Z_{5}Z_{6}$, $Y_{0}Y_{1}Y_{4}Y_{5}$, $Y_{0}Y_{2}Y_{4}Y_{6}\rangle$ \\ 
\midrule $[[7,3,2]]$  & 499 & 432 & $\langle X_{3}X_{4}Z_{5}Z_{6}$, $Z_{3}Z_{4}Y_{5}Y_{6}$, $Y_{0}X_{1}Y_{2}Y_{3}Z_{4}Z_{5}$, $X_{0}Z_{1}X_{2}Z_{4}X_{5}Z_{6}\rangle$ \\ 
\midrule $[[8,0,4]]$  & 125 & 8064 & $\langle X_{0}X_{1}X_{2}X_{3}$, $Z_{0}Z_{1}Z_{5}Z_{6}$, $Z_{0}Z_{2}Z_{5}Z_{7}$, $Z_{0}Z_{3}Z_{6}Z_{7}$, $Z_{4}Z_{5}Z_{6}Z_{7}$, $X_{1}X_{2}X_{4}X_{5}$, $X_{1}X_{3}X_{4}X_{6}$, $X_{2}X_{3}X_{4}X_{7}\rangle$ \\ 
\midrule $[[8,2,2]]$  & 4584 & 288 & $\langle X_{2}X_{3}Z_{4}Z_{5}$, $Z_{2}Z_{3}Y_{4}Y_{5}$, $Z_{4}Z_{5}Y_{6}Z_{7}$, $Y_{2}Y_{3}Z_{6}X_{7}$, $X_{0}Z_{1}Y_{2}Z_{3}Z_{4}Z_{7}$, $Z_{0}Y_{1}Z_{2}Y_{4}Z_{6}Z_{7}\rangle$ \\ 
 & 4492 & 2160 & $\langle Y_{3}Y_{4}Z_{5}Z_{6}$, $Z_{3}Y_{5}Z_{6}Z_{7}$, $Z_{4}Z_{5}Y_{6}Z_{7}$, $X_{3}Z_{4}Z_{6}Y_{7}$, $Y_{0}Z_{1}Y_{2}Z_{3}Z_{4}Z_{7}$, $Z_{0}X_{1}Z_{2}Y_{3}Z_{5}Z_{7}\rangle$ \\ 
 & 4926 & 2304 & $\langle Z_{0}Z_{1}Z_{2}Z_{3}$, $X_{0}X_{1}X_{2}X_{3}$, $Z_{0}Z_{1}Z_{4}Z_{5}$, $X_{0}X_{1}X_{4}X_{5}$, $Z_{0}Z_{1}Z_{6}Z_{7}$, $X_{0}X_{1}X_{6}X_{7}\rangle$ \\ 
\midrule $[[8,2,3]]$  & 4947 & 1728 & $\langle Y_{0}Y_{1}Z_{4}Z_{7}$, $X_{2}X_{3}Z_{5}Z_{6}$, $Z_{2}Z_{3}Y_{5}Y_{6}$, $Z_{0}Z_{1}X_{4}X_{7}$, $Y_{0}Z_{1}Y_{2}Z_{3}Y_{4}Z_{6}$, $X_{0}Z_{1}X_{2}Z_{3}X_{5}Z_{7}\rangle$ \\ 
\midrule $[[8,4,2]]$  & 2206 & 1152 & $\langle Z_{0}Z_{1}Z_{2}Z_{3}$, $X_{0}X_{1}X_{2}X_{3}$, $X_{0}X_{1}X_{4}X_{5}X_{6}X_{7}$, $Z_{0}Z_{1}Z_{4}Z_{5}Z_{6}Z_{7}\rangle$ \\ 
 & 3041 & 1152 & $\langle Z_{0}Z_{1}X_{2}X_{3}X_{6}X_{7}$, $X_{0}X_{1}X_{2}X_{3}X_{4}X_{5}$, $Z_{0}Z_{1}Z_{2}Z_{3}Z_{4}Z_{5}$, $Y_{0}Y_{1}Z_{2}Z_{3}Z_{6}Z_{7}\rangle$ \\ 
\midrule $[[8,6,2]]$  & 67 & 241920 & $\langle Z_{0}Z_{1}Z_{2}Z_{3}Z_{4}Z_{5}Z_{6}Z_{7}$, $X_{0}X_{1}X_{2}X_{3}X_{4}X_{5}X_{6}X_{7}\rangle$ \\ 
\midrule $[[9,1,3]]$  & 10201 & 1152 & $\langle X_{0}X_{1}Z_{4}Z_{5}$, $X_{0}X_{2}Z_{4}Z_{6}$, $X_{0}X_{3}Z_{4}Z_{7}$, $Z_{0}Z_{1}Y_{4}Y_{5}$, $Z_{0}Z_{2}Y_{4}Y_{6}$, $Z_{0}Z_{3}Y_{4}Y_{7}$, $Y_{0}Y_{4}Z_{5}Z_{6}Z_{7}Z_{8}$, $Y_{0}Z_{1}Z_{2}Z_{3}Z_{4}Y_{8}\rangle$ \\ 
 & 9652 & 4320 & $\langle Y_{0}Y_{1}Z_{7}Z_{8}$, $Y_{2}Y_{3}Z_{4}Z_{5}$, $Y_{2}Z_{3}Y_{4}Z_{6}$, $X_{2}Z_{4}X_{5}Z_{6}$, $X_{2}Z_{3}Z_{5}X_{6}$, $Z_{0}Z_{1}X_{7}X_{8}$, $X_{0}Z_{1}X_{2}Z_{3}Z_{4}Z_{8}$, $Z_{0}Y_{2}Z_{5}Z_{6}Y_{7}Z_{8}\rangle$ \\ 
\midrule $[[9,3,2]]$  & 118847 & 144 & $\langle X_{0}X_{1}X_{2}X_{3}$, $Z_{0}Z_{1}Z_{2}Z_{3}$, $Z_{0}Y_{1}X_{2}Y_{4}Z_{5}X_{6}$, $X_{0}X_{1}Z_{4}Y_{5}X_{7}X_{8}$, $Y_{0}Y_{2}X_{4}Y_{6}Z_{7}Z_{8}$, $Y_{0}X_{1}Z_{2}X_{4}Y_{5}Z_{6}\rangle$ \\ 
 & 163595 & 576 & $\langle Y_{0}X_{1}Y_{2}Y_{3}$, $Y_{5}Y_{6}Z_{7}Z_{8}$, $Z_{0}Y_{1}Z_{2}Z_{3}$, $Z_{5}Z_{6}X_{7}X_{8}$, $Z_{0}Z_{1}Y_{2}Y_{4}Y_{5}Y_{6}$, $X_{0}X_{1}Z_{2}Z_{4}X_{7}X_{8}\rangle$ \\ 
 & 131752 & 864 & $\langle Z_{0}Z_{1}Z_{2}Z_{3}$, $X_{0}X_{2}X_{4}X_{8}$, $Y_{0}Y_{1}Y_{2}Y_{3}$, $Y_{1}Y_{3}Y_{4}Y_{8}$, $Y_{0}Y_{1}Y_{4}Y_{5}Y_{6}Y_{7}$, $Z_{1}Z_{2}Z_{5}Z_{6}Z_{7}Z_{8}\rangle$ \\ 
\midrule $[[9,3,3]]$  & 170235 & 1296 & $\langle Y_{0}X_{1}X_{2}Y_{3}X_{4}X_{5}$, $Z_{0}Z_{1}X_{2}X_{6}Y_{7}X_{8}$, $Z_{0}Z_{2}Z_{3}Z_{4}Z_{6}Y_{7}$, $X_{1}Z_{2}Y_{3}Z_{5}Z_{7}Z_{8}$, $X_{0}Z_{1}Z_{2}X_{3}Z_{4}Z_{5}$, $X_{0}X_{2}X_{3}X_{4}X_{6}Z_{7}\rangle$ \\ 
\midrule $[[9,5,2]]$  & 14986 & 1296 & $\langle Y_{0}X_{1}X_{2}X_{3}X_{4}X_{5}$, $Z_{0}Z_{1}X_{2}X_{6}X_{7}X_{8}$, $X_{0}Z_{1}Z_{2}Z_{3}Z_{4}Z_{5}$, $Y_{0}Y_{1}Z_{2}Z_{6}Z_{7}Z_{8}\rangle$ \\ 
 & 8643 & 8640 & $\langle X_{0}X_{1}X_{2}X_{3}$, $Z_{0}Z_{1}Z_{2}Z_{3}$, $Y_{0}Z_{1}X_{2}X_{4}X_{5}X_{6}X_{7}X_{8}$, $X_{0}Y_{1}Z_{2}Z_{4}Z_{5}Z_{6}Z_{7}Z_{8}\rangle$ \\ 
\bottomrule
\end{tabularx}
\endgroup
\caption{All indecomposable GF(4)-linear codes with $n<10$. \label{tab:all_gf_lin_indecom}}
\end{table}

\paragraph{Distance distribution} Table~\ref{table:Grassl} from \texttt{codetables.de}~\cite{Grassl:codetables} by Markus Grassl shows the bounds on the minimum distance of binary qubit quantum codes $[[n,k,d]]$. We expand on this table by seeing how the distance of a code is distributed. Figure~\ref{fig:distdistall} in the appendix shows the distribution of the distances of all binary non-equivalent stabilizer codes. Figure~\ref{fig:distdistindec} in the appendix shows the distribution of the distances of all binary non-equivalent indecomposable stabilizer codes. Figure~\ref{fig:maxdprobind} show the probability that you will find a code of maximum minimum distance for each given $[[n,k]]$ and Figure~\ref{fig:mindisthighprobind} shows the minimum code distances with highest probability of being found for indecomposable codes.

\begin{figure}
\centering
  \includegraphics[width=.5\textwidth]{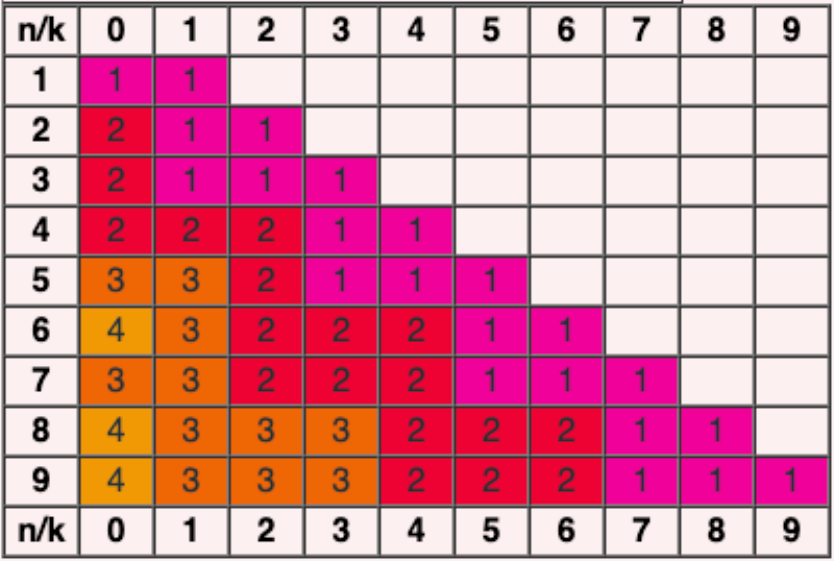}
  \caption{Bounds on the minimum distance of binary qubit quantum codes $[[n,k,d]]$ from~\cite{Grassl:codetables}
  \label{table:Grassl}}
\end{figure}

\begin{figure}
\centering
    \begin{minipage}{0.45\textwidth}
        \centering
        \includegraphics[width=\textwidth]{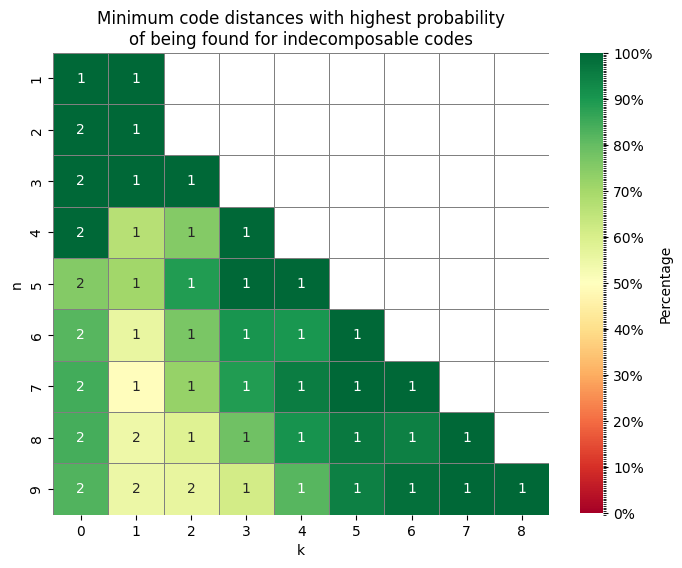}
        \caption{Minimum code distances with highest probability of being found for indecomposable codes.
        \label{fig:mindisthighprobind}}
    \end{minipage}
    \hspace{0.05\textwidth} 
    \begin{minipage}{0.45\textwidth}
        \centering
        \includegraphics[width=\textwidth]{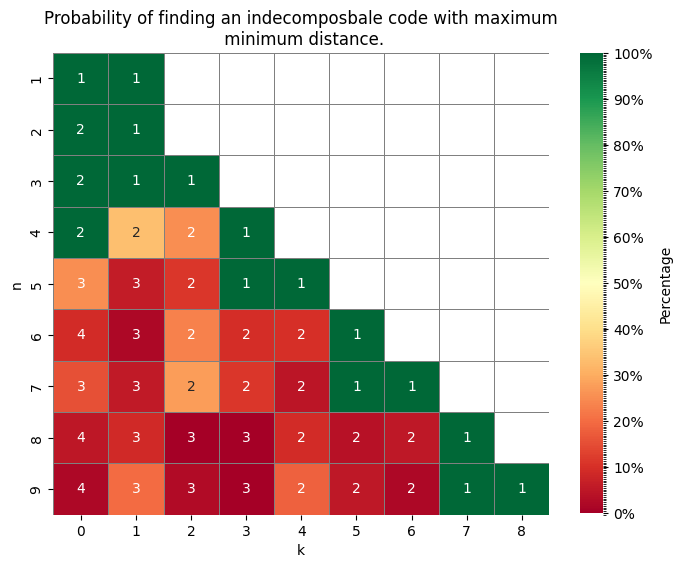}
        \caption{The probabilities of finding an indecomposable code with the maximum minimum distance.
        \label{fig:maxdprobind}}
    \end{minipage}
\end{figure}

We now describe the smallest codes and highlight some families of larger codes. In what follows, $S$ denotes a stabilizer group, $|\aut{S}|$ is the order of the code's automorphism group, and $w(x)=\sum_{P\in S} x^{\mathrm{wt}(P)}$ is the code's weight enumerator, where $\mathrm{wt}(P)$ counts the number of non-identity factors in the Pauli operator $P$. In order to specify a particular code in the database we will use the notation
$\seclass{n}{k}{i}$ or $\eclass{n}{k}{d}{i}$ where $i$ is the index of the code in the set of $[[n,k]]$ codes. 

\paragraph{Error-detecting codes on 4 qubits} The four-qubit error-detecting codes are listed within Table~\ref{tab:codes4}. These well-known codes \cite{vgw96,gbp97,kpeb18} are the smallest examples of color and surface codes. Each code's stabilizer generators correspond to faces of a planar graph as shown in Fig.~\ref{fig:codes4}. An arrow $P\rightarrow C$ indicates that a parent-child relationship exists between the child class $C$ and the parent class $P$. Such as relationship exists if given any code $R_C$ representing class $C$ there exists a code $R_P$ representing class $P$ such that $R_C\subset R_P$. In terms of stabilizer groups this implies that the associated stabilizer groups satisfy the relationship $S_P \leq S_C$. That is, there exists some set of generators $G$ such that $\langle S_P, G\rangle = S_C$. For example, if  
$S_P =\langle Z_{0}Z_{1}Z_{2}Z_{3}$, $X_{0}X_{1}X_{2}X_{3}\rangle$ then $S_P$ is a class representative of the parent class $[[4,2,2]]:9$ with child class
$[[4,1,2]]:8$  represented by the group $\langle S_P, X_{0}Y_{1}Z_{2}\rangle$ as shown in~\cite{Jordan_2006}. Figure~\ref{fig:pcrel-4qubits-indecom} shows the parent child relationship between all indecomposable four qubit equivalence classes.

\begin{table}
\renewcommand{\arraystretch}{1.25}
\centering
\begin{tabularx}{\textwidth}{l|l|l|l|l}
\toprule
$[[n,k,d]]$ & $\mathrm{Idx}$ & $|\mathrm{Aut}(S)|$ & $S$ & $w(x)$ \\ 
\specialrule{1.5pt}{1pt}{1pt}
\midrule $[[4,0,2]]$  & 2 & 32 & $\langle X_{0}X_{2}$, $Z_{1}Z_{3}$, $Z_{0}Z_{2}Z_{3}$, $X_{1}X_{2}X_{3}\rangle$ & $1 + 2x^{2} + 8x^{3} + 5x^{4}$ \\ 
 & 3 & 192 & $\langle Z_{0}Z_{3}$, $Z_{1}Z_{3}$, $Z_{2}Z_{3}$, $X_{0}X_{1}X_{2}X_{3}\rangle$ & $1 + 6x^{2} + 9x^{4}$ \\ 
\midrule $[[4,1,2]]$  & 8 & 24 & $\langle Z_{0}X_{2}Z_{3}$, $Y_{0}X_{1}Y_{2}$, $Z_{1}Z_{2}X_{3}\rangle$ & $1 + 4x^{3} + 3x^{4}$ \\ 
 & 6 & 32 & $\langle X_{0}X_{1}$, $X_{2}X_{3}$, $Z_{0}Z_{1}Z_{2}Z_{3}\rangle$ & $1 + 2x^{2} + 5x^{4}$ \\ 
\midrule $[[4,2,2]]$  & 9 & 144 & $\langle Z_{0}Z_{1}Z_{2}Z_{3}$, $X_{0}X_{1}X_{2}X_{3}\rangle$ & $1 + 3x^{4}$ \\ 
\bottomrule
\end{tabularx}
\caption{All $[[4,k,2]]$ indecomposable equivalence classes with the $[[4,1,2]]$ and $[[4,2,2]]$ classes being the only error-detecting equivalence classes with four qubits. We have given a CSS representative if one exists. The one four qubit decomposable class with $d=2$ (not shown) is equivalent to two Bell pairs. \label{tab:codes4}}
\end{table}

\definecolor{cZstab}{rgb}{1.0, 0.13, 0.32}
\definecolor{cXstab}{rgb}{0.55, 0.71, 0.0}
\definecolor{cYstab}{rgb}{0.13, 0.67, 0.8}

\tikzset{
    strong/.style={
        line width=1.4pt 
    }
}
\begin{figure}
\centering
\begin{tikzpicture}


\begin{scope}[shift={(0,2)}] 
    \fill[cZstab] (0.5,0.5) circle (1);
    \draw[strong] (0.5,0.5) circle (1);

    \fill[cXstab] (0,0) rectangle (1,1);
    \draw[strong] (0,0) rectangle (1,1);

    \foreach \corner in {(0,0), (1,0), (0,1), (1,1)}
        \filldraw[black] \corner circle (2pt);

        \node at (0.5,0.5) {$X$};
        
        \node at (0.5,1.25) {$Z$};
    
\end{scope}

\begin{scope}[shift={(4,0)}] 
    
    \fill[cXstab] (0,0) rectangle (1,1);
    
    
    \fill[cZstab] (1,1) arc[start angle=90, end angle=-90, radius=0.5];
    \draw[strong] (1,1) arc[start angle=90, end angle=-90, radius=0.5];
    
    
    \fill[cZstab] (0,0) arc[start angle=270, end angle=90, radius=0.5];
    \draw[strong] (0,0) arc[start angle=270, end angle=90, radius=0.5];
    
    \draw[strong] (0,0) rectangle (1,1);
    
    \foreach \corner in {(0,0), (1,0), (0,1), (1,1)}
        \filldraw[black] \corner circle (2pt);
\end{scope}

\begin{scope}[shift={(8,0)}]
    \fill[cXstab] (0,0) rectangle (1,1);
    \draw[strong] (0,0) rectangle (1,1);
    
    
    \fill[cZstab] (1,1) arc[start angle=90, end angle=-90, radius=0.5];
    \draw[strong] (1,1) arc[start angle=90, end angle=-90, radius=0.5];
    
    
    \fill[cZstab] (0,0) arc[start angle=270, end angle=90, radius=0.5];
    \draw[strong] (0,0) arc[start angle=270, end angle=90, radius=0.5];
    
    
    \fill[cZstab] (0,1) arc[start angle=180, end angle=0, radius=0.5];
    \draw[strong] (0,1) arc[start angle=180, end angle=0, radius=0.5];
    
    \foreach \corner in {(0,0), (1,0), (0,1), (1,1)}
        \filldraw[black] \corner circle (2pt);
\end{scope}

\begin{scope}[shift={(3.5,4)}]
    \coordinate (A) at (0,0);
    \coordinate (B) at (2,0);
    \coordinate (C) at (1,{sqrt(3)}); 

    \coordinate (center) at (barycentric cs:A=1,B=1,C=1);

    \coordinate (H) at (barycentric cs:A=1,C=1,center=1);
    \coordinate (I) at (barycentric cs:C=1,B=1,center=1);
    \coordinate (J) at (barycentric cs:A=1,B=1,center=1);

    \coordinate (W) at ($(A)!0.5!(C)$);
    \coordinate (L) at ($(C)!0.5!(center)$);
    \coordinate (X) at ($(C)!0.5!(B)$);
    \coordinate (M) at ($(B)!0.5!(center)$);
    \coordinate (Y) at ($(A)!0.5!(B)$);
    \coordinate (N) at ($(A)!0.5!(center)$);

    \fill[cYstab] (A) -- (W) -- (H) -- (N) -- cycle;
    \fill[cYstab] (N) -- (H) -- (L) -- (center) -- cycle;
    \fill[cXstab] (W) -- (H) -- (L) -- (C) -- cycle;
    \fill[cZstab] (C) -- (L) -- (I) -- (X) -- cycle;
    \fill[cZstab] (L) -- (I) -- (M) -- (center) -- cycle;
    \fill[cXstab] (X) -- (B) -- (M) -- (I) -- cycle;
    \fill[cZstab] (B) -- (M) -- (J) -- (Y) -- cycle;
    \fill[cXstab] (M) -- (J) -- (N) -- (center) -- cycle;
    \fill[cZstab] (Y) -- (J) -- (N) -- (A) -- cycle;
    
    \draw[thick] (A) -- (B) -- (C) -- cycle;
    
    \draw[thick] (A) -- (center);
    \draw[thick] (B) -- (center);
    \draw[thick] (C) -- (center);
    
    \foreach \point in {A,B,C}
        \filldraw[black] (\point) circle (2pt); 
    
    \filldraw[black] (center) circle (2pt); 
\end{scope}

\node at (4,4.5) {$Y$};
\node at (4.35,5.15) {$X$};
\node at (4.4,4.75) {$Y$};

\begin{scope}[shift={(7.5,4)}]
    \coordinate (A) at (0,0);
    \coordinate (B) at (2,0);
    \coordinate (C) at (1,{sqrt(3)}); 
    
    \coordinate (center) at (barycentric cs:A=1,B=1,C=1);
    
    \fill[cXstab] (A) to[bend left=60] (C);
    \fill[cZstab] (A) -- (C) -- (center) -- cycle;
    \fill[cXstab] (A) -- (B) -- (center) -- cycle;
    \fill[cZstab] (B) to[bend right=80] (center);
    
    \draw[thick] (B) -- (A) -- (C);
    
    \draw[thick] (B) to[bend right=80] (center);
    
    \draw[thick] (A) to[bend left=60] (C);

    \draw[thick] (A) -- (center);
    \draw[thick] (B) -- (center);
    \draw[thick] (C) -- (center);

    \foreach \point in {A,B,C}
        \filldraw[black] (\point) circle (2pt); 
    
    \filldraw[black] (center) circle (2pt); 
\end{scope}


\node at (0.5,1.1) {$[[4,2,2]]:9$};
\node at (4.5,3.5) {$[[4,1,2]]:8$};
\node at (4.5,-0.5) {$[[4,1,2]]:6$};
\node at (8.5,3.5) {$[[4,0,2]]:2$};
\node at (8.5,-0.5) {$[[4,0,2]]:3$};

\draw[->, thick] (2,2.5) -- (3,1.5); 
\draw[->, thick] (2,2.5) -- (3,3.5);
\draw[->, thick] (5.5,5) -- (7,5);
\draw[->, thick] (6,0.5) -- (7,0.5);
\draw[->, thick] (6,2) -- (7,3);

\end{tikzpicture}

\caption{Representatives of all equivalence classes of indecomposable four qubit $d=2$ codes drawn as planar graphs whose faces correspond to stabilizer generators. An arrow $P\rightarrow C$ indicates that a parent-child relationship exists between the child class $C$ and the parent class $P$. Such as relationship exists if given any code $R_C$ representing class $C$ there exists a code $R_P$ representing class $P$ such that $R_C\subset R_P$.\label{fig:codes4}}
\end{figure}

\begin{figure}
\centering
  \includegraphics[width=.7\textwidth]{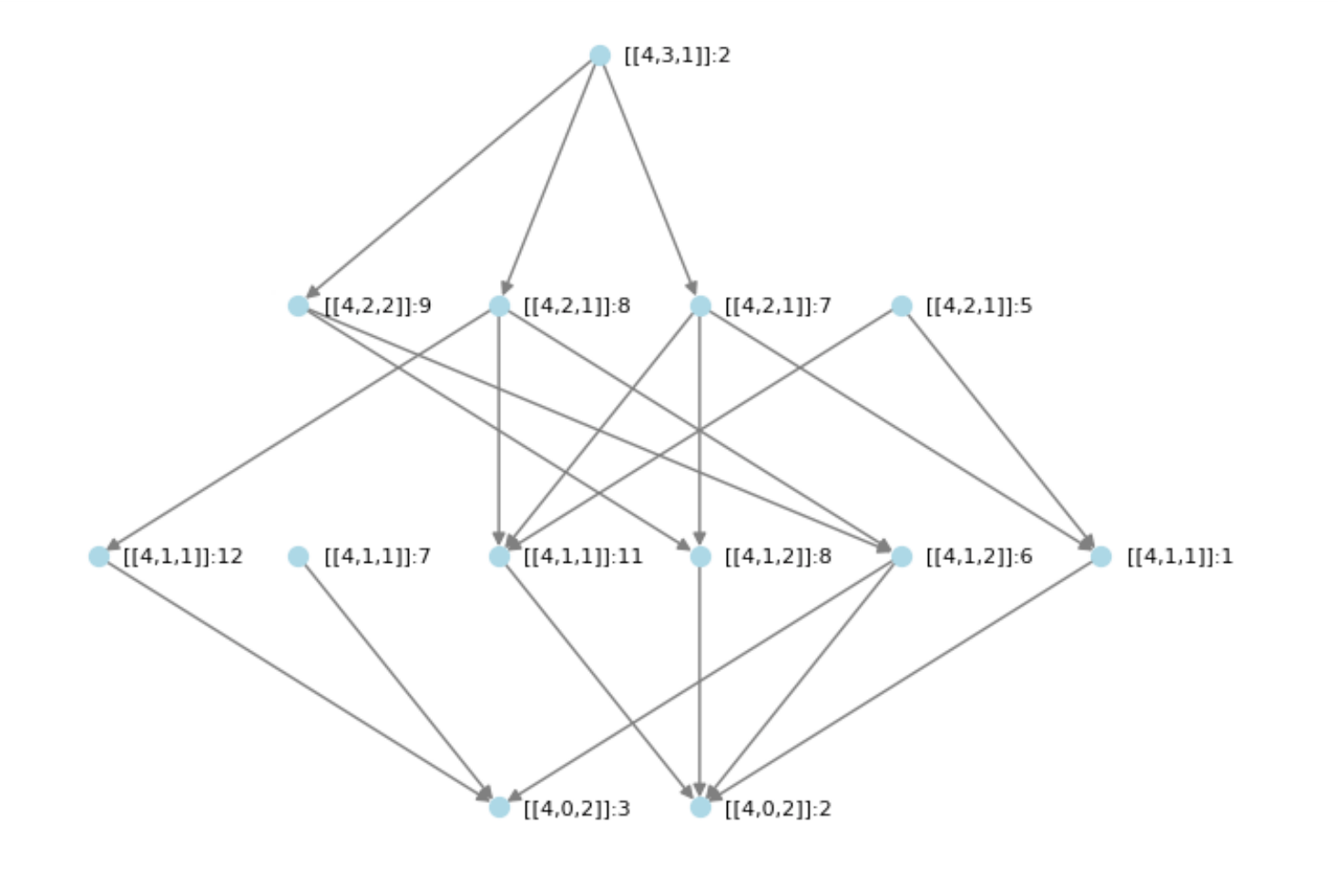}
  \caption{Parent child relationship between all indecomposable four qubit codes.\label{fig:pcrel-4qubits-indecom}}
\end{figure}

\paragraph{Error-detecting codes on 5 qubits} All the seven qubit equivalence classes with $d\geq 2$ are listed in Table~\ref{tab:codes5}. Seven of these are error-detecting codes. Two decomposable $[[5,1,2]]$ codes are not listed. They are obtained by appending a single qubit onto each of the $[[4,1,2]]$ codes and have $|\aut{S}|=48$ and $64$. One decomposable $[[5,2,2]]$ code is not listed. It is obtained by appending a single qubit onto the $[[4,2,2]]$ and has $|\aut{S}|=288$. The $\eclass{5}{1}{3}{7}$ code is famously the smallest stabilizer code correcting one error \cite{lmpz96,bdsw96}, but the other codes are less studied. With exception of the $[[5,1,3]]$ code, each code's stabilizer generators correspond to faces of a planar graph as shown in Fig.~\ref{fig:codes5}. All five qubit codes are contained in the Tables in the Appendix.

\begin{table}
\renewcommand{\arraystretch}{1.25}
\centering
\begingroup
\footnotesize
\begin{tabularx}{\textwidth}{l|l|l|l|l}
\toprule
$[[n,k,d]]$ & $\mathrm{Idx}$ & $|\mathrm{Aut}(S)|$ & $S$ & $w(x)$ \\ 
\specialrule{1.5pt}{1pt}{1pt}
\midrule $[[5,0,2]]$  & 4 & 32 & $\langle Z_{0}Z_{4}$, $Z_{1}Z_{2}$, $Z_{1}Z_{3}Z_{4}$, $X_{1}X_{2}X_{3}$, $X_{0}X_{3}X_{4}\rangle$ & $1 + 2x^{2} + 8x^{3} + 13x^{4} + 8x^{5}$ \\ 
 & 2 & 96 & $\langle Z_{0}Z_{4}$, $Z_{1}Z_{4}$, $X_{2}X_{3}$, $Z_{2}Z_{3}Z_{4}$, $X_{0}X_{1}X_{3}X_{4}\rangle$ & $1 + 4x^{2} + 6x^{3} + 11x^{4} + 10x^{5}$ \\ 
 & 6 & 1920 & $\langle Z_{0}Z_{4}$, $Z_{1}Z_{4}$, $Z_{2}Z_{4}$, $Z_{3}Z_{4}$, $X_{0}X_{1}X_{2}X_{3}X_{4}\rangle$ & $1 + 10x^{2} + 5x^{4} + 16x^{5}$ \\ 
\midrule $[[5,0,3]]$  & 7 & 120 & $\langle X_{0}Z_{1}Z_{2}$, $Z_{0}X_{1}Z_{3}$, $Z_{0}X_{2}Z_{4}$, $Z_{1}X_{3}Z_{4}$, $Z_{2}Z_{3}X_{4}\rangle$ & $1 + 10x^{3} + 15x^{4} + 6x^{5}$ \\ 
\midrule $[[5,1,2]]$  & 18 & 8 & $\langle Z_{0}Z_{1}Z_{3}$, $Z_{2}Z_{3}Z_{4}$, $X_{1}X_{2}X_{3}$, $X_{0}X_{3}X_{4}\rangle$ & $1 + 4x^{3} + 7x^{4} + 4x^{5}$ \\ 
 & 20 & 8 & $\langle X_{1}X_{2}$, $X_{1}Z_{3}Z_{4}$, $Y_{0}X_{3}Y_{4}$, $X_{0}Y_{1}Z_{2}Y_{3}\rangle$ & $1 + x^{2} + 3x^{3} + 6x^{4} + 5x^{5}$ \\ 
 & 14 & 16 & $\langle Z_{1}Z_{4}$, $X_{2}X_{3}$, $X_{0}X_{1}X_{4}$, $Z_{0}Z_{2}Z_{3}Z_{4}\rangle$ & $1 + 2x^{2} + 2x^{3} + 5x^{4} + 6x^{5}$ \\ 
 & 9 & 96 & $\langle Z_{0}Z_{4}$, $Z_{1}Z_{4}$, $Z_{2}Z_{3}$, $X_{0}X_{1}X_{2}X_{3}X_{4}\rangle$ & $1 + 4x^{2} + 3x^{4} + 8x^{5}$ \\ 
\midrule $[[5,1,3]]$  & 21 & 360 & $\langle Y_{0}Y_{1}Z_{2}Z_{3}$, $Y_{0}Z_{1}Y_{2}Z_{4}$, $X_{0}Z_{2}X_{3}Z_{4}$, $X_{0}Z_{1}Z_{3}X_{4}\rangle$ & $1 + 15x^{4}$ \\ 
\midrule $[[5,2,2]]$  & 27 & 12 & $\langle Z_{0}Z_{1}X_{4}$, $X_{0}X_{2}Z_{3}Z_{4}$, $X_{1}Z_{2}X_{3}Z_{4}\rangle$ & $1 + x^{3} + 3x^{4} + 3x^{5}$ \\ 
 & 26 & 48 & $\langle Z_{1}Z_{4}$, $Z_{0}Z_{2}Z_{3}Z_{4}$, $X_{0}X_{1}X_{2}X_{3}X_{4}\rangle$ & $1 + x^{2} + 2x^{4} + 4x^{5}$ \\ 
\bottomrule
\end{tabularx}
\endgroup
\caption{All $[[5,k]]$ indecomposable equivalence classes with $d\geq 2$ with the $[[5,1,2]]$, $[[5,2,3]]$ and $[[5,2,2]$ classes being the only error-detecting equivalence classes with five qubits. We have given a CSS representative if one exists.  \label{tab:codes5}}
\end{table}

\begin{figure}
    \centering
\begin{tikzpicture}
\begin{scope}[shift={(3.5,5)}]
    \node[
       regular polygon,
       regular polygon sides=5,
       shape border rotate=0,
       minimum size=1.5cm, 
       overlay 
    ] (P) {};

    \coordinate (O) at (0,0);
    \coordinate (M) at ($(P.corner 3)!0.5!(P.corner 4)$);
    \coordinate (C) at ($(O)!0.5!(M)$);
    \coordinate (d) at (0,-0.32);

    \fill[cXstab] (d) ellipse (1 and 0.8);
    \draw[strong] (d) ellipse (1 and 0.8);

    \fill[cXstab] (P.corner 1) to[bend left=90] (P.corner 5);
    \draw[strong] (P.corner 1) to[bend left=90] (P.corner 5);
    
    \fill[cZstab] (P.corner 1) -- (P.corner 2) -- (P.corner 3) -- (P.corner 4) -- (P.corner 5) -- cycle;
    \draw[strong] (P.corner 1) -- (P.corner 2) -- (P.corner 3) -- (P.corner 4) -- (P.corner 5) -- cycle;

    \foreach \point in {P.corner 1, P.corner 2, P.corner 3, P.corner 4, P.corner 5}
        \filldraw[black] (\point) circle (2pt); 

    \node at (0,-1.5) {$[[5,2,2]]:26$};

\end{scope}

\begin{scope}[shift={(0.5,1)}]
\node[
       regular polygon,
       regular polygon sides=5,
       shape border rotate=0,
       minimum size=1.5cm, 
       overlay 
    ] (P) {};

    \coordinate (O) at (0,0);
    \coordinate (M) at ($(P.corner 3)!0.5!(P.corner 4)$);
    \coordinate (C) at ($(O)!0.5!(M)$);
    \coordinate (d) at (0,-0.32);

    \fill[cXstab] (P.corner 1) to[bend left=90] (P.corner 5);
    \draw[strong] (P.corner 1) to[bend left=90] (P.corner 5);

    \fill[cXstab] (P.corner 1) to[bend right=90] (P.corner 2);
    \draw[strong] (P.corner 1) to[bend right=90] (P.corner 2);

    \fill[cXstab] (P.corner 3) to[bend right=90] (P.corner 4);
    \draw[strong] (P.corner 3) to[bend right=90] (P.corner 4);
    
    \fill[cZstab] (P.corner 1) -- (P.corner 2) -- (P.corner 3) -- (P.corner 4) -- (P.corner 5) -- cycle;
    \draw[strong] (P.corner 1) -- (P.corner 2) -- (P.corner 3) -- (P.corner 4) -- (P.corner 5) -- cycle;

    \foreach \point in {P.corner 1, P.corner 2, P.corner 3, P.corner 4, P.corner 5}
        \filldraw[black] (\point) circle (2pt); 

    \node at (0,-1.5) {$[[5,1,2]]:9$};

\end{scope}

\begin{scope}[shift={(3,0.5)}]
    \fill[cXstab] (0,0) rectangle (1,1);


    \fill[cZstab] (1,1) -- (0.5,1.5) -- (0,1) -- cycle;
    \draw[strong] (1,1) -- (0.5,1.5) -- (0,1) -- cycle;


    \fill[cZstab] (0,0) to[bend right=90] (1,0);
    \draw[strong] (0,0) to[bend right=90] (1,0);


    \fill[cXstab] (0.5,1.5) to[bend left=90] (1,1);
    \draw[strong] (0.5,1.5) to[bend left=90] (1,1);

    \draw[strong] (0,0) rectangle (1,1);
    
    \foreach \corner in {(0,0), (1,0), (0,1), (1,1), (0.5,1.5)}
        \filldraw[black] \corner circle (2pt);

    \node at (0.5,-1) {$[[5,1,2]]:14$};
\end{scope}

\begin{scope}[shift={(6,0.5)}]
    \fill[cZstab] (0.5,0.5) -- (0.5,1) -- (0,1) -- (0, 0.5) -- cycle;
    \fill[cXstab] (0.5,0.5) -- (0, 0.5) -- (0,0) -- (0.5,0) -- cycle;
    \fill[cYstab] (0.5,1) -- (1,1) -- (1,0) -- (0.5,0) -- cycle;

    \fill[cXstab] (0,1) to[bend left=90] (1,1);

    \coordinate (L) at (1.5, -0.5);
    \coordinate (A) at (0,0);
    \coordinate (B) at (1,0);
    \coordinate (C) at (1,1);
    \coordinate (ro) at (barycentric cs:C=1,B=1,L=1);
    \coordinate (lo) at (barycentric cs:A=1,B=1,L=1);
    \coordinate (D) at ($(A)!0.5!(L)$);
    \coordinate (E) at ($(C)!0.5!(L)$);
    \coordinate (F) at ($(C)!0.5!(B)$);
    \coordinate (G) at ($(A)!0.5!(B)$);
    \coordinate (H) at ($(B)!0.5!(L)$);

    \fill[cXstab] (C) -- (E) -- (ro) -- (F) -- cycle;
    \fill[cZstab] (E) -- (L) -- (H) -- (ro) -- cycle;
    \fill[cZstab] (H) -- (ro) -- (F) -- (B) -- cycle;

    \fill[cYstab] (L) -- (H) -- (lo) -- (D) -- cycle;
    \fill[cYstab] (D) -- (lo) -- (G) -- (A) -- cycle;
    \fill[cXstab] (H) -- (B) -- (G) -- (lo) -- cycle;

    \draw[strong] (A) -- (L) -- (C);
    \draw[strong] (L) -- (B);
    \draw[strong] (0,1) to[bend left=90] (C);
    \draw[strong] (A) rectangle (C);

    \foreach \corner in {(A), (B), (C), (L)}
        \filldraw[black] \corner circle (2pt);

    \node at (0.5,-1) {$[[5,1,2]]:20$};

\end{scope}

\begin{scope}[shift={(5.5,4)}]
    \coordinate (A) at (0,0);
    \coordinate (B) at (2,0);
    \coordinate (C) at (1,2.5); 

    \coordinate (center) at (barycentric cs:A=1,B=1,C=1);

    \coordinate (J) at (barycentric cs:A=1,B=1,center=1);

    \coordinate (M) at ($(B)!0.5!(center)$);
    \coordinate (Y) at ($(A)!0.5!(B)$);
    \coordinate (N) at ($(A)!0.5!(center)$);

    \coordinate (D) at ($(C)!0.5!(center)$);
    \coordinate (E) at ($(C)!0.5!(D)$);
    \coordinate (F) at ($(D)!0.5!(center)$);
    \coordinate (G) at ($(C)!0.5!(B)$);
    \coordinate (H) at ($(E)!0.5!(M)$);
    \coordinate (K) at ($(A)!0.5!(C)$);
    \coordinate (I) at ($(K)!0.5!(F)$);

    \fill[cZstab] (B) -- (M) -- (J) -- (Y) -- cycle;
    \fill[cXstab] (M) -- (J) -- (N) -- (center) -- cycle;
    \fill[cZstab] (Y) -- (J) -- (N) -- (A) -- cycle;

    \fill[cXstab] (B) -- (G) -- (H) -- (M) -- cycle;
    \fill[cZstab] (M) -- (center) -- (F) -- (H) -- cycle;
    \fill[cZstab] (F) -- (H) -- (E) -- cycle;
    \fill[cXstab] (C) -- (G) -- (H) -- (E) -- cycle;

    \fill[cZstab] (C) -- (K) -- (I) -- (E) -- cycle;
    \fill[cXstab] (E) -- (F) -- (I) -- cycle;
    \fill[cXstab] (K) -- (I) -- (N) -- (A) -- cycle;
    \fill[cZstab] (N) -- (center) -- (F) -- (I) -- cycle;

    \draw[thick] (A) -- (B) -- (C) -- cycle;
    
    \draw[thick] (A) -- (center);
    \draw[thick] (B) -- (center);
    \draw[thick] (C) -- (center);
    
    \foreach \point in {A, B, C, D, center}
        \filldraw[black] (\point) circle (2pt); 

    \node at (1,-0.5) {$[[5,2,2]]:27$};
     
\end{scope}

\begin{scope}[shift={(9,0.5)}]

    \coordinate (A) at (0,0);
    \coordinate (B) at (1,0);
    \coordinate (C) at (1,1);
    \coordinate (D) at (0,1);
    \coordinate (center) at (0.5,0.5);

    \fill[cZstab] (D) -- (A) -- (center) -- cycle;
    \fill[cZstab] (C) -- (B) -- (center) -- cycle;
    \fill[cXstab] (D) -- (C) -- (center) -- cycle;
    \fill[cXstab] (A) -- (B) -- (center) -- cycle;

    \draw[strong] (A) rectangle (C);
    \draw[strong] (D) -- (B);
    \draw[strong] (C) -- (A);

    \foreach \corner in {A, B, C, D, center}
        \filldraw[black] (\corner) circle (2pt);

    \node at (0.5,-1) {$[[5,1,2]]:18$};

\end{scope}
\draw[->, thick] (3.5,3.25) -- (1.5,2.25);
\draw[->, thick] (3.5,3.25) -- (3.5,2.25); 
\draw[->, thick] (3.5,3.25) -- (5.5,2.25);

\draw[->, thick] (6.5,3.25) -- (4.5,2.25);
\draw[->, thick] (6.5,3.25) -- (6.5,2.25);
\draw[->, thick] (6.5,3.25) -- (8.5,2.25);

\end{tikzpicture}

\caption{Representatives of all equivalence classes of indecomposable five qubit $d=2$  codes with $k\geq 1$ drawn as planar graphs whose faces correspond to stabilizer generators. An arrow $P\rightarrow C$ indicates that a parent-child relationship exists between the child class $C$ and the parent class $P$. Such as relationship exists if given any code $R_C$ representing class $C$ there exists a code $R_P$ representing class $P$ such that $R_C\subset R_P$.\label{fig:codes5}}

\end{figure}
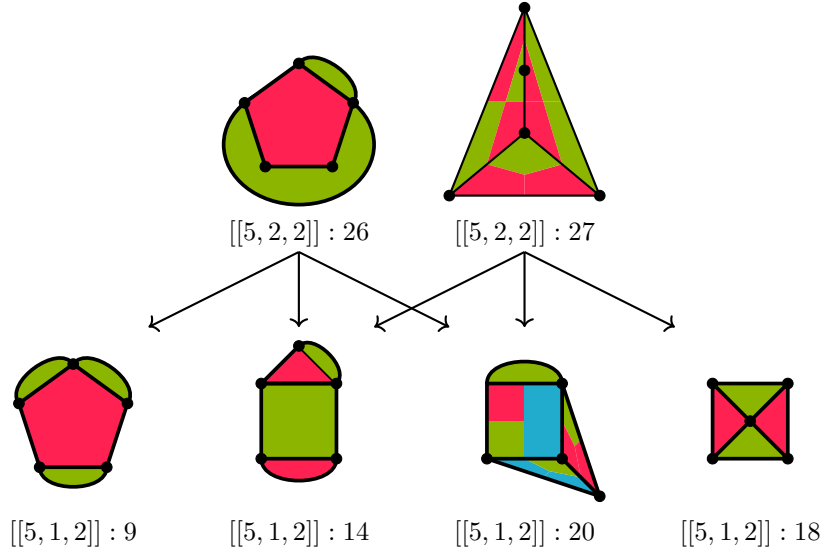

\paragraph{Error-correcting codes on 6 qubits}

There are 59 indecomposable error-detecting codes on six qubits, a number that is already too large to present in a concise manner. Therefore, we focus attention on the error-correcting codes from this point forward. There are exactly two inequivalent $[[6,1,3]]$ codes \cite{CRSS98,swokl08} $[[6,1,3]]:68$ and $[[6,1,3]]:87$. The class $[[6,1,3]]:87$ is decomposable and can be obtained from $[5,1,3]]:21$ by appending a single qubit. The other is the degenerate indecomposable class $[[6,1,3]]:68$ with the properties
\begin{align}
    S & = \langle X_{0}Z_{5}, X_{1}X_{2}Z_{3}Z_{4}, Y_{1}Y_{3}Z_{4}Z_{5}, X_{1}Z_{2}X_{4}Z_{5}, Z_{0}Y_{1}Z_{2}Z_{3}Y_{5}\rangle, \\
    w(x) & = 1 + x^{2} + 11x^{4} + 16x^{5} + 3x^{6}, \qquad |\aut{S}| = 96.
\end{align}
\noindent The tables in the Appendix include a complete list of all 6 qubit indecomposable classes.

\subsection{Error-correcting codes on 7 qubits}

Table~\ref{tab:codes7} lists the $16$ indecomposable $[[7,1,3]]$ codes classes. These classes were previously known \cite{YCO07}, and have been presented as codeword stabilized codes \cite{CSSZ09}. We briefly comment on the five most symmetric codes, i.e. those with the largest automorphism groups. The Steane code \cite{S96} $\eclass{7}{1}{3}{226}$ is the only $[[7,1,3]]$ CSS code class. It has the largest automorphism group among 7-qubit error-correcting codes, is the smallest error-correcting CSS code, and has no gauge freedom. The code class $\eclass{7}{1}{3}{108}$ with $|\aut{S}|=768$ has been called the bare 7-qubit code \cite{LGDHB17}. The code class $\eclass{7}{1}{3}{115}$ with $|\aut{S}|=576$ appears on an ascending code conversion path between the $[[5,1,3]]$ and Steane codes \cite{HFWH13}:
\begin{align*}
\eclass{5}{1}{3}{21} &\rightarrow 
\eclass{6}{1}{3}{68}\otimes X^{\otimes 4} \rightarrow
\eclass{7}{1}{3}{115}\otimes X^{\otimes 3} \rightarrow \\ 
&\quad\rightarrow
\eclass{8}{1}{3}{524}\otimes X^{\otimes 3} 
\rightarrow \eclass{9}{1}{3}{3528}\otimes X^{\otimes 1} \rightarrow
\cdots\rightarrow
\eclass{7}{1}{3}{226}\otimes X^{\otimes 3}.
\end{align*}
\noindent The code class $\eclass{7}{1}{3}{190}$ with $|\aut{S}|=192$ is the smallest example of a triangle code \cite{YK17}.

\begin{table}
\renewcommand{\arraystretch}{1.25}
\centering
\begingroup
\begin{tabularx}{\textwidth}{l|l|L{20}|l}
\toprule
$\mathrm{Idx}$ & $|\mathrm{Aut}(S)|$ & $S$ & $w(x)$ \\ 
\specialrule{1.5pt}{1pt}{1pt}
185 & 4 & $\langle Z_{3}X_{4}Z_{6}$, $Z_{2}Z_{4}X_{6}$, $Y_{0}Y_{1}Z_{2}Z_{5}$, $Y_{0}Z_{1}Y_{2}Z_{6}$, $Z_{0}X_{1}X_{3}Z_{4}$, $X_{0}Z_{2}Z_{3}X_{5}\rangle$ & $1 + 2x^{3} + 9x^{4} + 24x^{5} + 22x^{6} + 6x^{7}$ \\ 
200 & 6 & $\langle Z_{1}Z_{3}X_{5}$, $Z_{2}Z_{4}X_{6}$, $Y_{0}Y_{1}Z_{2}Z_{5}$, $Y_{0}Z_{1}Y_{2}Z_{6}$, $Z_{0}X_{1}X_{3}Z_{4}$, $Z_{0}X_{2}Z_{3}X_{4}\rangle$ & $1 + 2x^{3} + 9x^{4} + 24x^{5} + 22x^{6} + 6x^{7}$ \\ 
221 & 16 & $\langle X_{0}X_{1}Z_{3}Z_{4}$, $X_{0}X_{2}Z_{3}Z_{5}$, $Z_{0}Z_{1}Y_{3}Y_{4}$, $Z_{2}Z_{3}Z_{4}X_{5}$, $Z_{0}Z_{1}Z_{2}X_{6}$, $Y_{0}Y_{3}Z_{4}Z_{5}Z_{6}\rangle$ & $1 + 13x^{4} + 24x^{5} + 18x^{6} + 8x^{7}$ \\ 
240 & 16 & $\langle X_{0}Z_{6}$, $Z_{1}Z_{3}X_{4}$, $Z_{2}Z_{3}X_{5}$, $X_{1}X_{2}Z_{4}Z_{5}$, $X_{1}X_{3}Z_{5}Z_{6}$, $Z_{0}Y_{1}Z_{2}Z_{4}Y_{6}\rangle$ & $1 + x^{2} + 2x^{3} + 7x^{4} + 24x^{5} + 23x^{6} + 6x^{7}$ \\ 
255 & 32 & $\langle X_{0}Z_{4}$, $Y_{1}Y_{2}Z_{3}Z_{5}$, $Y_{1}Z_{2}Y_{3}Z_{6}$, $Z_{0}X_{4}Z_{5}Z_{6}$, $Z_{2}Z_{3}Y_{5}Y_{6}$, $X_{1}Z_{3}Z_{4}X_{5}Z_{6}\rangle$ & $1 + x^{2} + 11x^{4} + 24x^{5} + 19x^{6} + 8x^{7}$ \\ 
257 & 32 & $\langle X_{0}Z_{6}$, $X_{1}X_{2}Z_{4}Z_{5}$, $X_{1}X_{3}Z_{5}Z_{6}$, $Y_{1}Z_{3}Y_{4}Z_{6}$, $Y_{2}Z_{3}Y_{5}Z_{6}$, $Z_{0}Z_{1}Z_{2}X_{6}\rangle$ & $1 + x^{2} + 19x^{4} + 43x^{6}$ \\ 
227 & 42 & $\langle Y_{0}Y_{1}Z_{2}Z_{5}$, $Y_{0}Z_{1}Y_{2}Z_{6}$, $Z_{0}X_{1}X_{3}Z_{4}$, $Z_{0}X_{2}Z_{3}X_{4}$, $X_{0}Z_{2}Z_{3}X_{5}$, $X_{0}Z_{1}Z_{4}X_{6}\rangle$ & $1 + 21x^{4} + 42x^{6}$ \\ 
209 & 48 & $\langle X_{0}X_{1}Z_{3}Z_{4}$, $X_{0}X_{2}Z_{3}Z_{5}$, $Z_{0}Z_{1}Y_{3}Y_{4}$, $Z_{0}Z_{2}Y_{3}Y_{5}$, $Z_{0}Z_{1}Z_{2}X_{6}$, $Y_{0}Y_{3}Z_{4}Z_{5}Z_{6}\rangle$ & $1 + 13x^{4} + 24x^{5} + 18x^{6} + 8x^{7}$ \\ 
164 & 64 & $\langle X_{0}Z_{5}$, $X_{1}Z_{6}$, $Y_{2}Y_{3}Z_{4}Z_{5}$, $Y_{2}Z_{3}Y_{4}Z_{6}$, $Z_{0}X_{2}Z_{4}X_{5}Z_{6}$, $Z_{1}X_{2}Z_{3}Z_{5}X_{6}\rangle$ & $1 + 2x^{2} + 9x^{4} + 24x^{5} + 20x^{6} + 8x^{7}$ \\ 
166 & 64 & $\langle X_{0}Z_{6}$, $X_{1}Z_{3}$, $Z_{3}X_{4}Z_{5}Z_{6}$, $Z_{1}X_{2}Y_{3}Y_{4}$, $Z_{2}Z_{4}X_{5}Z_{6}$, $Z_{0}Y_{2}Z_{4}Y_{6}\rangle$ & $1 + 2x^{2} + 17x^{4} + 44x^{6}$ \\ 
239 & 96 & $\langle X_{0}Z_{1}$, $Z_{0}X_{1}Z_{6}$, $X_{2}X_{3}Z_{4}Z_{5}$, $Y_{2}Y_{4}Z_{5}Z_{6}$, $X_{2}Z_{3}X_{5}Z_{6}$, $Z_{1}Y_{2}Z_{3}Z_{4}Y_{6}\rangle$ & $1 + x^{2} + 2x^{3} + 7x^{4} + 24x^{5} + 23x^{6} + 6x^{7}$ \\ 
228 & 144 & $\langle Y_{0}Y_{1}Z_{5}Z_{6}$, $Z_{0}X_{1}X_{2}Z_{3}$, $X_{0}Z_{1}Z_{2}X_{3}$, $X_{0}Z_{1}X_{4}Z_{6}$, $Z_{0}Z_{3}Z_{4}X_{5}$, $Z_{1}Z_{2}Z_{4}X_{6}\rangle$ & $1 + 21x^{4} + 42x^{6}$ \\ 
190 & 192 & $\langle X_{0}Z_{6}$, $X_{1}Z_{5}$, $X_{2}Z_{4}$, $Z_{2}X_{3}X_{4}Z_{5}$, $Z_{1}Y_{3}Y_{5}Z_{6}$, $Z_{0}Z_{3}Z_{4}X_{6}\rangle$ & $1 + 3x^{2} + 15x^{4} + 45x^{6}$ \\ 
115 & 576 & $\langle X_{0}Z_{6}$, $X_{1}Z_{6}$, $X_{2}X_{3}Z_{4}Z_{5}$, $Y_{2}Y_{4}Z_{5}Z_{6}$, $X_{2}Z_{3}X_{5}Z_{6}$, $Z_{0}Z_{1}Y_{2}Z_{3}Z_{4}Y_{6}\rangle$ & $1 + 3x^{2} + 15x^{4} + 45x^{6}$ \\ 
108 & 768 & $\langle X_{0}Z_{4}$, $X_{1}Z_{4}$, $X_{2}Z_{5}$, $X_{3}Z_{6}$, $Z_{2}Z_{3}Y_{5}Y_{6}$, $Z_{0}Z_{1}Z_{2}X_{4}X_{5}Z_{6}\rangle$ & $1 + 5x^{2} + 11x^{4} + 47x^{6}$ \\ 
226 & 1008 & $\langle Z_{0}Z_{1}Z_{3}Z_{6}$, $Z_{0}Z_{2}Z_{3}Z_{5}$, $Y_{1}Y_{2}Y_{3}Y_{4}$, $Z_{3}Z_{4}Z_{5}Z_{6}$, $Y_{0}Y_{1}Y_{4}Y_{5}$, $Y_{0}Y_{2}Y_{4}Y_{6}\rangle$ & $1 + 21x^{4} + 42x^{6}$ \\ 
\bottomrule
\end{tabularx}
\endgroup
\caption{All 16 classes if indecomposable $[[7,1,3]]$ codes sorted by automorphism group order. The last four classes are as follows: $[[7,1,3]]:226$ represents the Steane code~\cite{S96}; $[[7,1,3]]:108$ represents the bare 7-qubit code~\cite{LGDHB17}; $[[7,1,3]]:115$ represents the 7-qubit 5-7 conversion code~\cite{HFWH13}; and $[[7,1,3]]:190$ represents the $d=3$ triangle code~\cite{YK17}.\label{tab:codes7} }
\end{table}

\subsection{Error-correcting codes on 8 qubits}

There are 157 indecomposable $[[8,1,3]]$ codes. The classes with the largest automorphisms groups are listed in Table~\ref{tab:codes813inde1024}. There are 21 indecomposable codes for which we found generating sets with planar Tanner graphs; the four with the largest automorphism groups are listed in Table~\ref{tab:codes813inde1024} and have parent subsystem codes.  There are only two $[[8,1,3]]$ codes with $|\aut{S}|\geq 1024$ for which we did not find planar embeddings. These are
\begin{itemize}
	\item \eclass{8}{1}{3}{984} : $S=\langle X_1Z_8, X_2Z_8, X_3Z_8, X_4X_5Z_6Z_7, Y_4Y_6Z_7Z_8, X_4Z_5X_7Z_8, Y_1Y_2Z_3X_4Y_5Y_6X_8\rangle$, \\ $w(x)=1+6x^2+20x^4+34x^6+64x^7+3x^8$, $|\aut{S}|=4608$,
	\item \eclass{8}{1}{3}{602} : $S=\langle X_1Z_8, X_2Z_8, Y_3Y_4Z_5Z_6, Y_3Z_4Y_5Z_7, X_3Z_5X_6Z_7, X_3Z_4Z_6X_7, Y_1Y_2Z_3X_5Z_7X_8\rangle$, \\ $w(x)=1+3x^2+15x^4+85x^6+24x^8$, $|\aut{S}|=2880$.
\end{itemize}
The first code has an $[[8,1,3]]$ parent subsystem code with three gauge qubits whose X gauge operators are the first three generators of the code. The corresponding Z gauge operators are $Z_1$, $Z_2$, and $Z_3$. Therefore the Z-basis gauge-fixing is the $[[5,1,3]]$ code. The second code has an $[[8,1,3]]$ parent subsystem code with two gauge qubits
\begin{align}
G = \langle & Y_1Z_2X_3Z_4Z_5Y_8, Z_1Z_2Z_3X_4Z_6X_8, Z_1Z_2Z_3X_5Z_7X_8, Z_1Y_2Z_4X_6Z_7Y_8, Z_1Z_2Z_5Z_6X_7X_8 \\
& \tilde{X}_1=X_1Z_8, \tilde{Z}_1=Z_1Z_3, \tilde{X}_2=X_2Z_8, \tilde{Z}_2=Z_2Z_6\rangle.
\end{align}
In both cases these are the unique parent subsystem codes (up to code equivalence) with the maximal number of gauge qubits. The \eclass{8}{1}{3}{855} code class with the largest automorphism group is the class at the top of Table~\ref{tab:codes813inde1024} with $|\textrm{Aut}(S)|=6144$. The entries of Table~\ref{tab:codes813inde1024} that have planar Tannar graphs are presented graphically in Figure~\ref{fig:codes8}.

\begin{table}
\renewcommand{\arraystretch}{1.55}
\centering
\begingroup
\begin{tabularx}{\textwidth}{l|l|L{24}|>{\hsize=4cm\RaggedRight} X}
\toprule
$\mathrm{Idx}$ & $|\mathrm{Aut}(S)|$ & $S$ & $w(x)$ \\ 
\specialrule{1.5pt}{1pt}{1pt}
855  & 6144 &  $\langle X_{0}Z_{7}, X_{1}Z_{7}, X_{2}Z_{7}, X_{3}Z_{5},$ $X_{3}Z_{5}, X_{4}Z_{6}, Z_{3}Z_{4}Y_{5}Y_{6}, Z_{0}Z_{1}Z_{2}Z_{3}X_{5}Z_{6}X_{7}\rangle$ & $1 + 8x^{2} + 18x^{4} + 32x^{6} + 64x^{7} + 5x^{8}$\\ 
894 & 4608 &  $\langle X_{0}Z_{7}, X_{1}Z_{7}, X_{2}Z_{7}, X_{3}X_{4}Z_{5}Z_{6},$ $X_{3}X_{4}Z_{5}Z_{6}, Y_{3}Y_{5}Z_{6}Z_{7}, X_{3}Z_{4}X_{6}Z_{7}, Z_{0}Z_{1}Z_{2}Y_{3}Z_{4}Z_{5}Y_{7}\rangle$ & $1 + 6x^{2} + 20x^{4} + 34x^{6} + 64x^{7} + 3x^{8}$\\ 
602 & 2880 &  $\langle X_{0}Z_{7}, X_{1}Z_{7}, Y_{2}Y_{3}Z_{4}Z_{5}, Y_{2}Z_{3}Y_{4}Z_{6},$ $Y_{2}Z_{3}Y_{4}Z_{6}, X_{2}Z_{4}X_{5}Z_{6}, X_{2}Z_{3}Z_{5}X_{6}, Z_{0}Z_{1}X_{2}Z_{3}Z_{4}X_{7}\rangle$ & $1 + 3x^{2} + 15x^{4} + 85x^{6} + 24x^{8}$\\ 
631 & 2304 &  $\langle X_{0}Z_{7}, X_{1}Z_{7}, X_{2}Z_{6}, X_{3}Z_{6},$ $X_{3}Z_{6}, X_{4}Z_{5}, Z_{0}Z_{1}Z_{4}Y_{5}Y_{7}, Z_{2}Z_{3}Z_{4}X_{5}X_{6}Z_{7}\rangle$ & $1 + 7x^{2} + 15x^{4} + 8x^{5} + 33x^{6} + 56x^{7} + 8x^{8}$\\ 
1090 & 2048 &  $\langle X_{0}Z_{4}, X_{1}Z_{5}, X_{2}Z_{6}, X_{3}Z_{7},$ $X_{3}Z_{7}, Z_{0}Z_{1}Y_{4}Y_{5}, Z_{2}Z_{3}Y_{6}Y_{7}, Z_{0}Z_{2}X_{4}Z_{5}X_{6}Z_{7}\rangle$ & $1 + 4x^{2} + 14x^{4} + 84x^{6} + 25x^{8}$\\ 
1094 & 1024 &  $\langle X_{0}Z_{6}, X_{1}Z_{7}, X_{2}Z_{5}, X_{3}X_{4},$ $X_{3}X_{4}, Z_{0}Y_{3}Z_{4}Z_{5}Y_{6}, Z_{0}Z_{2}X_{5}X_{6}Z_{7}, Z_{1}Z_{3}Z_{4}Z_{6}X_{7}\rangle$ & $1 + 4x^{2} + 6x^{4} + 32x^{5} + 36x^{6} + 32x^{7} + 17x^{8}$\\ 
\bottomrule
\end{tabularx}
\endgroup
\caption{All indecomposable $[[8,1,3]]$ code classes with $|\mathrm{Aut}(S)|\geq 1024$. All of the above codes are degenerate. \label{tab:codes813inde1024}}
\end{table}

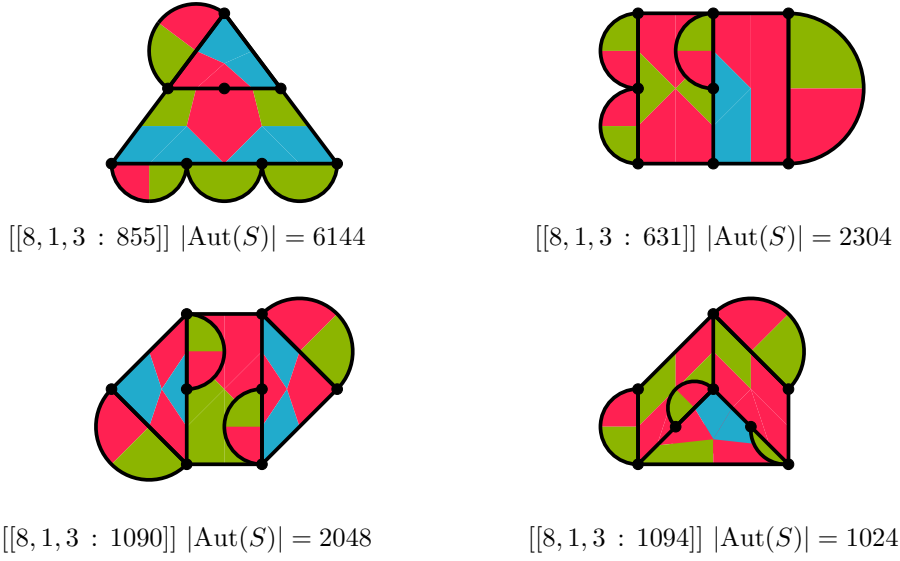
\begin{figure}[h!]
    \centering
    \begin{tikzpicture}

        \node at (1,3) {$\eclass{8}{1}{3}{855}$ $|\aut{S}|=6144$};
        \begin{scope}[shift={(0,4)}]
            \coordinate (A) at (0,0);
            \coordinate (B) at (1.5,2);
            \coordinate (C) at (3,0);
            \coordinate (X) at (1,0);
            \coordinate (Y) at (2,0);
            \coordinate (D) at ($(A)!0.5!(B)$);
            \coordinate (E) at ($(B)!0.5!(C)$);
            \coordinate (F) at ($(D)!0.5!(E)$);
            \coordinate (G) at ($(A)!0.5!(X)$);
            \coordinate (H) at ($(X)!0.5!(Y)$);
            \coordinate (I) at ($(Y)!0.5!(C)$);
            \coordinate (J) at ($(E)!0.5!(C)$);
            \coordinate (K) at ($(B)!0.5!(E)$);

            \coordinate (L) at ($(B)!0.5!(D)$);
            \coordinate (M) at ($(A)!0.5!(D)$);
            \coordinate (N) at (1.5,1);

            \coordinate (R) at ($(N)!0.5!(E)$);
            \coordinate (S) at ($(D)!0.5!(N)$);

            \coordinate (TBY) at (barycentric cs:D=1,B=1,E=1);

            \fill[cYstab] (B) -- (K) -- (TBY) -- (L) -- cycle;
            \fill[cYstab] (K) -- (E) -- (R) -- (TBY) -- cycle;
            \fill[cZstab] (R) -- (TBY) -- (S) -- cycle;
            \fill[cZstab] (S) -- (D) -- (L) -- (TBY) -- cycle;

            \fill[cXstab] (D) -- (S) -- (1,0.5) -- (M) -- cycle;
            \fill[cYstab] (M) -- (1,0.5) --  (G) -- (A) -- cycle;
            \fill[cYstab] (G) -- (1,0.5) -- (H) -- cycle;
            \fill[cZstab] (S) -- (1,0.5) -- (H) -- (2,0.5) -- (R) -- cycle;
            \fill[cYstab] (H) -- (2,0.5) -- (I) -- cycle;
            \fill[cYstab] (C) -- (I) -- (2,0.5) -- (J) -- cycle;
            \fill[cXstab] (J) -- (2,0.5) -- (R) -- (E) -- cycle;

            \fill[cXstab] (X) arc (0:-90:0.5) -- (G) -- cycle;
            \fill[cZstab] (0.5,-0.5) arc (-90:-180:0.5) -- (0.5,0) -- cycle;

            \fill[cXstab] (0.75,1) arc (233.13:143.13:0.625) -- (1.125, 1.5) -- cycle;
            \fill[cZstab] (1.5,2) arc (53.13:143.13:0.625) -- (1.125, 1.5) -- cycle;

            \fill[cXstab] (Y) arc (0:-180:0.5);
            \fill[cXstab] (C) arc (0:-180:0.5);
        
            \draw[strong] (D) -- (B) -- (E) -- cycle;
            \draw[strong] (A) -- (B) -- (C) -- cycle;

            \draw[strong] (X) arc (0:-90:0.5);
            \draw[strong] (0.5,-0.5) arc (-90:-180:0.5);

            \draw[strong] (0.75,1) arc (233.13:143.13:0.625);
            \draw[strong] (1.5,2) arc (53.13:143.13:0.625);

            \draw[strong] (Y) arc (0:-180:0.5);
            \draw[strong] (C) arc (0:-180:0.5);

            \foreach \corner in {A,B,C,D,E,N,X,Y}
                \filldraw[black] (\corner) circle (2pt);
        \end{scope}

        \node at (8,3) {$\eclass{8}{1}{3}{631}$ $|\aut{S}|=2304$};
        \begin{scope}[shift={(7,4)}]
            \coordinate (A) at (0,0);
            \coordinate (B) at (0,2);
            \coordinate (C) at (2,2);
            \coordinate (D) at (2,0);
            \coordinate (E) at (1,1);
            \coordinate (F) at (1,2);
            \coordinate (G) at (0,1);
            \coordinate (H) at ($(C)!0.5!(D)$);
            \coordinate (I) at (1.5,1);
            \coordinate (L) at (1,0);
            \coordinate (J) at ($(D)!0.5!(L)$);
            \coordinate (M) at ($(A)!0.5!(L)$);
            \coordinate (N) at ($(A)!0.5!(G)$);
            \coordinate (O) at ($(B)!0.5!(G)$);
            \coordinate (P) at ($(B)!0.5!(F)$);
            \coordinate (Q) at ($(F)!0.5!(C)$);
            \coordinate (R) at (0.5,1);
            \coordinate (S) at ($(F)!0.5!(E)$);
            \coordinate (T) at ($(E)!0.5!(L)$);

        \fill[cZstab] (B) -- (P) -- (R) -- (O) -- cycle;
        \fill[cZstab] (P) -- (F) -- (S) -- (R) -- cycle;
        \fill[cXstab] (O) -- (R) -- (N) -- cycle;
        \fill[cZstab] (A) -- (N) -- (R) -- (M) -- cycle;
        \fill[cZstab] (M) -- (L) -- (T) -- (R) -- cycle;
        \fill[cXstab] (T) -- (R) -- (S) -- cycle;

        \fill[cZstab] (F) -- (S) -- (I) -- (Q) -- cycle;
        \fill[cZstab] (Q) -- (I) -- (H) -- (C) -- cycle;
        \fill[cZstab] (H) -- (D) -- (J) -- (I) -- cycle;
        \fill[cYstab] (J) -- (L) -- (T) -- (I) -- cycle;
        \fill[cYstab] (T) -- (S) -- (I) -- cycle;

        \fill[cXstab] (0,2) arc (90:180:0.5) -- (0,1.5) -- cycle;
        \fill[cZstab] (-0.5,1.5) arc (180:270:0.5) -- (0,1.5) -- cycle;

        \fill[cZstab] (0,1) arc (90:180:0.5) -- (0,0.5) -- cycle;
        \fill[cXstab] (-0.5,0.5) arc (180:270:0.5) -- (0,0.5) -- cycle;

        \fill[cXstab] (1,2) arc (90:180:0.5) -- (1,1.5) -- cycle;
        \fill[cZstab] (0.5,1.5) arc (180:270:0.5) -- (1,1.5) -- cycle;

        \fill[cXstab] (2,2) arc (90:0:1) -- (2,1) -- cycle;
        \fill[cZstab] (3,1) arc (0:-90:1) -- (2,1) -- cycle;

        \draw[strong] (A) -- (B) -- (C) -- (D) -- cycle;
        \draw[strong] (F) -- (L);
        \draw[strong] (0,2) arc (90:270:0.5) -- cycle;
        \draw[strong] (0,1) arc (90:270:0.5) -- cycle;
        \draw[strong] (1,2) arc (90:270:0.5) -- cycle;
        \draw[strong] (2,2) arc (90:-90:1) -- cycle;

        \foreach \corner in {A,B,C,D,E, F, G, L}
                \filldraw[black] (\corner) circle (2pt);  
            
        \end{scope}

        \node at (1,-1) {$\eclass{8}{1}{3}{1090}$ $|\aut{S}|=2048$};
        \begin{scope}[shift={(1,0)}]
            \coordinate (A) at (0,0);
            \coordinate (B) at (0,2);
            \coordinate (C) at (1,2);
            \coordinate (D) at (1,0);
            \coordinate (E) at (1,1);
            \coordinate (F) at (0,1);
            \coordinate (G) at (2,1);
            \coordinate (H) at (-1,1);
            \coordinate (I) at ($(H)!0.5!(B)$);
            \coordinate (J) at ($(B)!0.5!(C)$);
            \coordinate (K) at ($(C)!0.5!(G)$);
            \coordinate (L) at ($(D)!0.5!(G)$);
            \coordinate (M) at ($(A)!0.5!(D)$);
            \coordinate (N) at ($(A)!0.5!(H)$);
            \coordinate (O) at ($(A)!0.5!(F)$);
            \coordinate (P) at ($(D)!0.5!(E)$);
            \coordinate (Q) at ($(C)!0.5!(E)$);
            \coordinate (R) at ($(B)!0.5!(F)$);
            \coordinate (S) at ($(F)!0.5!(E)$);
            \coordinate (T) at (barycentric cs:A=1,H=1,B=1);
            \coordinate (U) at (barycentric cs:C=1,D=1,G=1);

            \fill[cZstab] (B) -- (J) -- (S) -- (R) -- cycle;
            \fill[cXstab] (R) -- (S) -- (O) -- cycle;
            \fill[cXstab] (O) -- (A) -- (M) -- (S) -- cycle;
            \fill[cXstab] (M) -- (D) -- (P) -- (S) -- cycle;
            \fill[cZstab] (P) -- (S) -- (Q) -- cycle;
            \fill[cZstab] (Q) -- (S) -- (J) -- (C) -- cycle;
            \fill[cYstab] (C) -- (Q) -- (U) -- (K) -- cycle;
            \fill[cZstab] (Q) -- (U) -- (P) -- cycle;
            \fill[cYstab] (P) -- (D) -- (L) -- (U) -- cycle;
            \fill[cZstab] (U) -- (L) -- (G) -- (K) -- cycle;
            \fill[cZstab] (B) -- (R) -- (T) -- (I) -- cycle;
            \fill[cYstab] (R) -- (T) -- (O) -- cycle;
            \fill[cZstab] (T) -- (O) -- (A) -- (N) -- cycle;
            \fill[cYstab] (H) -- (I) -- (T) -- (N) -- cycle;

            \fill[cZstab] (-1,1) arc (135:225:0.7071) -- (-0.5,0.5) -- cycle;
            \fill[cXstab] (-1,0) arc (225:315:0.7071) -- (-0.5,0.5) -- cycle;

            \fill[cXstab] (0,2) arc (90:0:0.5) -- (0,1.5) -- cycle;
            \fill[cZstab] (0.5, 1.5) arc (0:-90:0.5) -- (0,1.5) -- cycle;

            \fill[cXstab] (1,1) arc (90:180:0.5) -- (1,0.5) -- cycle;
            \fill[cZstab] (0.5,0.5) arc (180:270:0.5) -- (1,0.5) -- cycle;

            \fill[cZstab] (1,2) arc (135:45:0.7071) -- (1.5,1.5) -- cycle;
            \fill[cXstab] (2,2) arc (45:-45:0.7071) -- (1.5,1.5) -- cycle;

            \draw[strong] (A) -- (B) -- (C) -- (D) -- cycle;
            \draw[strong] (C) -- (G) -- (D);
            \draw[strong] (B) -- (H) -- (A);

            \draw[strong] (-1,1) arc (135:315:0.7071) -- cycle;
            \draw[strong] (0,2) arc (90:-90:0.5) -- cycle;
            \draw[strong] (1,1) arc (90:270:0.5) -- cycle;
            \draw[strong] (1,2) arc (135:-45:0.7071) -- cycle;

            \foreach \corner in {A, B, C, D, E, F, G, H}
                \filldraw[black] (\corner) circle (2pt);
        \end{scope}

        \node at (8,-1) {$\eclass{8}{1}{3}{1094}$ $|\aut{S}|=1024$};
        \begin{scope}[shift={(7,0)}]
            \coordinate (A) at (0,0);
            \coordinate (B) at (0,1);
            \coordinate (C) at (1,2);
            \coordinate (D) at (2,1);
            \coordinate (E) at (2,0);
            \coordinate (F) at (1,1);
            \coordinate (G) at (0.5,0.5);
            \coordinate (H) at (1.5,0.5);
            \coordinate (I) at (1,0);
            \coordinate (J) at ($(A)!0.5!(B)$);
            \coordinate (K) at ($(B)!0.5!(C)$);
            \coordinate (L) at ($(C)!0.5!(D)$);
            \coordinate (M) at ($(D)!0.5!(E)$);
            \coordinate (N) at ($(I)!0.5!(E)$);
            \coordinate (O) at ($(A)!0.5!(I)$);
            \coordinate (P) at ($(A)!0.5!(G)$);
            \coordinate (Q) at ($(C)!0.5!(F)$);
            \coordinate (R) at ($(G)!0.5!(F)$);
            \coordinate (S) at ($(F)!0.5!(H)$);
            \coordinate (T) at ($(H)!0.5!(E)$);
            \coordinate (U) at (barycentric cs:A=1,E=1,F=1);
            \coordinate (V) at (barycentric cs:F=1,E=1,C=1,D=1);
            \coordinate (W) at (barycentric cs:A=1,B=1,C=1,F=1);

            \fill[cXstab] (B) -- (K) -- (W) -- (J) -- cycle;
            \fill[cZstab] (K) -- (C) -- (Q) -- (W) -- cycle;
            \fill[cXstab] (Q) -- (W) -- (R) -- (F) -- cycle;
            \fill[cZstab] (R) -- (W) -- (P) -- (G) -- cycle;
            \fill[cZstab] (A) -- (P) -- (W) -- (J) -- cycle;
            \fill[cXstab] (C) -- (Q) -- (V) -- (L) -- cycle;
            \fill[cZstab] (L) -- (V) -- (M) -- (D) -- cycle;
            \fill[cZstab] (M) -- (V) -- (T) -- (E) -- cycle;
            \fill[cZstab] (T) -- (V) -- (S) -- (A) -- cycle;
            \fill[cZstab] (S) -- (F) -- (Q) -- (V) -- cycle;
            \fill[cYstab] (F) -- (S) -- (U) -- (R) -- cycle;
            \fill[cZstab] (R) -- (U) -- (P) -- (G) -- cycle;
            \fill[cXstab] (A) -- (P) -- (U) -- (I) -- cycle;
            \fill[cZstab] (I) -- (U) -- (T) -- (E) -- cycle;
            \fill[cYstab] (U) -- (S) -- (H) -- (T) -- cycle;

            \fill[cZstab] (0,1) arc (90:180:0.5) -- (0,0.5) -- cycle;
            \fill[cXstab] (-0.5,0.5) arc (180:270:0.5) -- (0,0.5) -- cycle;

            \fill[cZstab] (1,1) arc (45:135:0.3536) -- (R) -- cycle;
            \fill[cXstab] (0.5,1) arc (135:225:0.3536) -- (R) -- cycle;

            \fill[cZstab] (1,2) arc (135:45:0.7071) -- (L) -- cycle;
            \fill[cXstab] (2,2) arc (45:-45:0.7071) -- (L) -- cycle;

            \fill[cXstab] (H) arc (175.298:274.702:0.4638) -- cycle;

            \draw[strong] (A) -- (B) -- (C) -- (D) -- (E) -- cycle;
            \draw[strong] (A) -- (F) -- (C);
            \draw[strong] (F) -- (E);

            \draw[strong] (0,1) arc (90:270:0.5) -- (0,0.5) -- cycle;

            \draw[strong] (1,1) arc (45:225:0.3536) -- (R) -- cycle;

            \draw[strong] (1,2) arc (135:-45:0.7071) -- (L) -- cycle;

            \draw[strong] (H) arc (175.298:274.702:0.4638);
            
            \foreach \corner in {A, B, C, D, E, F, G, H}
                \filldraw[black] (\corner) circle (2pt);
        \end{scope}

    \end{tikzpicture}
    \caption{All planar eight qubit $d=3$ codes with $|\aut{S}|\geq 1024$. There are two other $[[8,1,3]]$ classes listed in Table~\ref{tab:codes813inde1024} with $|\aut{S}|\geq 1024$ for which we did not find a planar embedding.\label{fig:codes8}}
\end{figure}

Table~\ref{tab:codes823} presents all 20 of the $[[8,2,3]]$ codes. All the code classes are nondegenerate except for the class \eclass{8}{2}{3}{5277} which  has $|\aut{S}|=4$. The $[[8,2,3]]$ code with the largest automorphism group, $|\aut{S}|=1728$, is the \eclass{8}{2}{3}{4947} class. The $[[8,2,3]]$ code with the next largest automorphism group has the same weight enumerator but a significantly smaller group order of 48.  Class \eclass{8}{2}{3}{3310} is the representative for the $[[8,2,3]]$ code in Grassl's code tables \cite{Grassl:codetables}. Only three of the codes in Table~\ref{tab:codes823} have parent subsystem codes.

\begin{table}
\renewcommand{\arraystretch}{1.55}
\centering
\begingroup
\begin{tabularx}{\textwidth}{l|l|L{24}|>{\hsize=4cm\RaggedRight} X}
\toprule
$\mathrm{Idx}$ & $|\mathrm{Aut}(S)|$ & $S$ & $w(x)$ \\ 
\specialrule{1.5pt}{1pt}{1pt}
4947 & 1728 &  $\langle Y_{0}Y_{1}Z_{4}Z_{7}, X_{2}X_{3}Z_{5}Z_{6}, Z_{2}Z_{3}Y_{5}Y_{6}, Z_{0}Z_{1}X_{4}X_{7},$ $Z_{2}Z_{3}Y_{5}Y_{6}, Z_{0}Z_{1}X_{4}X_{7}, Y_{0}Z_{1}Y_{2}Z_{3}Y_{4}Z_{6}, X_{0}Z_{1}X_{2}Z_{3}X_{5}Z_{7}\rangle$ & $1 + 6x^{4} + 48x^{6} + 9x^{8}$\\ 
4948 & 48 &  $\langle Y_{0}Z_{1}Y_{2}Z_{7}, Y_{3}Y_{4}Z_{5}Z_{6}, X_{1}X_{2}X_{4}X_{5}, Z_{0}Y_{1}Z_{4}Y_{6},$ $X_{1}X_{2}X_{4}X_{5}, Z_{0}Y_{1}Z_{4}Y_{6}, Z_{2}Z_{3}Y_{5}Y_{7}, X_{0}Z_{1}Z_{2}X_{3}Z_{4}Z_{5}\rangle$ & $1 + 6x^{4} + 48x^{6} + 9x^{8}$\\ 
4519 & 24 &  $\langle X_{1}X_{2}Z_{4}Z_{5}, Z_{1}Z_{2}X_{4}X_{5}, Y_{0}Y_{3}Z_{6}X_{7}, X_{0}X_{3}Z_{4}Z_{5}Z_{6},$ $Y_{0}Y_{3}Z_{6}X_{7}, X_{0}X_{3}Z_{4}Z_{5}Z_{6}, Y_{1}Z_{3}Y_{4}Z_{6}Z_{7}, Y_{0}Z_{2}Z_{3}X_{4}Y_{6}\rangle$ & $1 + 4x^{4} + 12x^{5} + 24x^{6} + 20x^{7} + 3x^{8}$\\ 
3745 & 16 &  $\langle X_{3}Z_{5}Z_{6}Z_{7}, Y_{0}Y_{1}Z_{2}X_{4}, X_{0}Z_{1}Z_{2}Z_{3}X_{5}, X_{2}Z_{3}Y_{4}Y_{5}Z_{6},$ $X_{0}Z_{1}Z_{2}Z_{3}X_{5}, X_{2}Z_{3}Y_{4}Y_{5}Z_{6}, X_{0}Z_{3}Z_{4}X_{6}Z_{7}, Z_{0}Z_{2}Z_{4}Y_{5}Y_{7}\rangle$ & $1 + 2x^{4} + 16x^{5} + 24x^{6} + 16x^{7} + 5x^{8}$\\ 
3744 & 12 &  $\langle Y_{0}Y_{1}Z_{6}Z_{7}, Z_{2}X_{3}X_{4}Z_{5}, X_{0}Z_{1}Z_{2}Z_{3}X_{5}, Y_{2}Y_{3}Z_{4}X_{5}Z_{6},$ $X_{0}Z_{1}Z_{2}Z_{3}X_{5}, Y_{2}Y_{3}Z_{4}X_{5}Z_{6}, X_{0}Z_{3}Z_{4}X_{6}Z_{7}, Z_{0}Z_{2}Z_{4}Y_{5}Y_{7}\rangle$ & $1 + 2x^{4} + 16x^{5} + 24x^{6} + 16x^{7} + 5x^{8}$\\ 
3310 & 6 &  $\langle X_{0}Z_{1}Z_{2}, X_{5}Z_{6}Z_{7}, X_{1}X_{2}X_{3}Z_{4}Z_{7}, Z_{0}X_{2}Y_{3}Y_{4}Z_{6},$ $X_{1}X_{2}X_{3}Z_{4}Z_{7}, Z_{0}X_{2}Y_{3}Y_{4}Z_{6}, Y_{1}X_{2}X_{4}Z_{5}Y_{6}, Z_{0}Y_{2}Z_{4}Z_{5}Y_{7}\rangle$ & $1 + 2x^{3} + 12x^{5} + 28x^{6} + 18x^{7} + 3x^{8}$\\ 
4525 & 6 &  $\langle Y_{0}Z_{1}Y_{2}Z_{7}, Z_{0}X_{1}X_{3}Z_{4}, Y_{3}Y_{4}Z_{6}Z_{7}, Z_{0}X_{2}X_{5}Z_{6},$ $Y_{3}Y_{4}Z_{6}Z_{7}, Z_{0}X_{2}X_{5}Z_{6}, Z_{1}Z_{3}Z_{5}X_{6}, Z_{2}Z_{4}Z_{5}X_{7}\rangle$ & $1 + 6x^{4} + 48x^{6} + 9x^{8}$\\ 
5277 & 4 &  $\langle X_{0}Z_{4}, X_{2}X_{3}Z_{5}Z_{6}, Z_{1}Z_{2}Z_{3}X_{7}, X_{1}Y_{2}Z_{4}Y_{5}Z_{6},$ $Z_{1}Z_{2}Z_{3}X_{7}, X_{1}Y_{2}Z_{4}Y_{5}Z_{6}, Y_{2}Z_{3}X_{5}Y_{6}Z_{7}, Z_{0}X_{4}Y_{5}Y_{6}X_{7}\rangle$ & $1 + x^{2} + 2x^{4} + 12x^{5} + 25x^{6} + 20x^{7} + 3x^{8}$\\ 
4149 & 3 &  $\langle X_{0}X_{1}Z_{4}Z_{7}, X_{0}X_{2}X_{3}Z_{4}, Z_{1}Z_{2}X_{4}X_{5}, Y_{1}Y_{3}X_{4}Z_{5}Z_{6},$ $Z_{1}Z_{2}X_{4}X_{5}, Y_{1}Y_{3}X_{4}Z_{5}Z_{6}, Z_{0}Z_{1}Y_{2}Z_{5}Y_{6}, Z_{0}Z_{3}Z_{4}Z_{5}X_{7}\rangle$ & $1 + 4x^{4} + 12x^{5} + 24x^{6} + 20x^{7} + 3x^{8}$\\ 
3710 & 2 &  $\langle Z_{1}Z_{2}X_{4}Z_{7}, Z_{0}Y_{2}Y_{5}X_{7}, Y_{0}Y_{2}X_{3}Z_{4}Z_{5}, Z_{0}Y_{1}X_{3}Y_{4}Z_{6},$ $Y_{0}Y_{2}X_{3}Z_{4}Z_{5}, Z_{0}Y_{1}X_{3}Y_{4}Z_{6}, X_{0}Z_{2}Y_{3}Y_{5}Z_{6}, Y_{1}Z_{4}Z_{5}Y_{6}Z_{7}\rangle$ & $1 + 2x^{4} + 16x^{5} + 24x^{6} + 16x^{7} + 5x^{8}$\\ 
3831 & 2 &  $\langle Z_{1}Z_{2}X_{4}Z_{7}, Z_{0}Y_{2}Y_{5}X_{7}, Y_{0}Z_{2}Y_{3}Z_{5}Z_{7}, Z_{0}X_{1}Y_{2}Y_{3}Z_{6},$ $Y_{0}Z_{2}Y_{3}Z_{5}Z_{7}, Z_{0}X_{1}Y_{2}Y_{3}Z_{6}, X_{1}Z_{2}X_{3}Z_{4}X_{5}, X_{0}Z_{1}Z_{3}X_{6}Z_{7}\rangle$ & $1 + 2x^{4} + 16x^{5} + 24x^{6} + 16x^{7} + 5x^{8}$\\ 
4091 & 2 &  $\langle Y_{0}Y_{1}X_{2}Z_{4}, X_{3}Z_{5}Z_{6}Z_{7}, X_{2}Z_{3}X_{5}Z_{6}, Z_{0}X_{1}Z_{2}X_{3}X_{4},$ $X_{2}Z_{3}X_{5}Z_{6}, Z_{0}X_{1}Z_{2}X_{3}X_{4}, X_{0}Z_{3}X_{4}Z_{5}X_{6}, X_{1}Z_{3}Y_{4}Z_{6}Y_{7}\rangle$ & $1 + 4x^{4} + 12x^{5} + 24x^{6} + 20x^{7} + 3x^{8}$\\ 
4154 & 2 &  $\langle Z_{2}X_{4}Z_{6}Z_{7}, Z_{1}Z_{3}Z_{4}X_{6}, Z_{0}Y_{2}X_{5}Y_{7}, Y_{0}X_{3}Y_{6}X_{7},$ $Z_{0}Y_{2}X_{5}Y_{7}, Y_{0}X_{3}Y_{6}X_{7}, X_{0}Z_{1}X_{3}Z_{5}Z_{6}, Y_{0}Y_{1}Y_{2}Y_{4}Z_{5}\rangle$ & $1 + 4x^{4} + 12x^{5} + 24x^{6} + 20x^{7} + 3x^{8}$\\ 
3354 & 1 &  $\langle Z_{1}Z_{2}X_{4}Z_{7}, Z_{0}X_{5}Z_{6}Z_{7}, X_{0}X_{2}Z_{4}Z_{5}Z_{7}, Z_{0}Y_{1}X_{3}Y_{4}Z_{6},$ $X_{0}X_{2}Z_{4}Z_{5}Z_{7}, Z_{0}Y_{1}X_{3}Y_{4}Z_{6}, Y_{1}X_{2}Z_{3}Z_{5}Y_{6}, X_{0}Y_{2}Z_{3}Z_{6}Y_{7}\rangle$ & $1 + 2x^{4} + 16x^{5} + 24x^{6} + 16x^{7} + 5x^{8}$\\ 
3829 & 1 &  $\langle Y_{0}Y_{1}Z_{6}Z_{7}, Z_{2}Z_{3}X_{5}Z_{7}, Z_{1}Z_{3}Z_{4}X_{6}, X_{0}X_{2}Y_{3}Y_{6},$ $Z_{1}Z_{3}Z_{4}X_{6}, X_{0}X_{2}Y_{3}Y_{6}, Y_{2}Y_{4}Z_{5}Z_{6}Z_{7}, Z_{0}Z_{3}Z_{4}Z_{5}X_{7}\rangle$ & $1 + 4x^{4} + 12x^{5} + 24x^{6} + 20x^{7} + 3x^{8}$\\ 
3952 & 1 &  $\langle Z_{0}Z_{2}X_{3}Z_{7}, Y_{2}Z_{5}Z_{6}Y_{7}, X_{1}X_{2}Z_{3}Z_{6}Z_{7}, X_{0}Z_{1}Y_{2}Y_{4}Z_{5},$ $X_{1}X_{2}Z_{3}Z_{6}Z_{7}, X_{0}Z_{1}Y_{2}Y_{4}Z_{5}, Y_{0}X_{2}Z_{4}Y_{5}Z_{6}, Y_{1}Z_{4}Z_{5}Y_{6}Z_{7}\rangle$ & $1 + 2x^{4} + 16x^{5} + 24x^{6} + 16x^{7} + 5x^{8}$\\ 
3979 & 1 &  $\langle X_{3}Z_{5}Z_{6}Z_{7}, Z_{2}Z_{3}Y_{4}Y_{5}, Z_{1}Y_{3}X_{4}Y_{6}, Z_{0}X_{4}X_{5}X_{7},$ $Z_{1}Y_{3}X_{4}Y_{6}, Z_{0}X_{4}X_{5}X_{7}, Z_{0}X_{1}X_{2}Z_{4}Z_{7}, X_{0}Z_{1}Y_{2}X_{3}Y_{4}\rangle$ & $1 + 4x^{4} + 12x^{5} + 24x^{6} + 20x^{7} + 3x^{8}$\\ 
4337 & 1 &  $\langle Z_{0}Z_{2}X_{3}Z_{7}, Z_{1}Z_{5}X_{6}Z_{7}, X_{3}X_{4}Y_{5}Y_{6}, X_{0}Y_{1}X_{4}Y_{7},$ $X_{3}X_{4}Y_{5}Y_{6}, X_{0}Y_{1}X_{4}Y_{7}, X_{1}X_{2}Z_{3}Z_{6}Z_{7}, X_{0}Z_{1}Y_{2}Y_{4}Z_{5}\rangle$ & $1 + 4x^{4} + 12x^{5} + 24x^{6} + 20x^{7} + 3x^{8}$\\ 
4934 & 1 &  $\langle X_{0}Z_{6}Z_{7}, X_{1}X_{2}X_{3}Z_{7}, Y_{2}Y_{3}X_{4}Z_{5}Z_{6}, X_{0}Y_{1}Y_{2}Z_{4}X_{5},$ $Y_{2}Y_{3}X_{4}Z_{5}Z_{6}, X_{0}Y_{1}Y_{2}Z_{4}X_{5}, Y_{0}Z_{1}Y_{3}Z_{4}X_{6}, Z_{0}Z_{3}Z_{4}Z_{5}X_{7}\rangle$ & $1 + x^{3} + x^{4} + 14x^{5} + 26x^{6} + 17x^{7} + 4x^{8}$\\ 
5834 & 1 &  $\langle X_{2}Z_{5}Z_{6}, X_{0}Y_{1}Z_{3}Y_{4}, Z_{0}X_{2}Z_{4}X_{7}, Y_{1}Y_{2}Y_{6}Y_{7},$ $Z_{0}X_{2}Z_{4}X_{7}, Y_{1}Y_{2}Y_{6}Y_{7}, X_{0}X_{1}X_{3}Z_{5}Z_{7}, X_{0}Y_{2}Y_{3}Z_{4}X_{5}\rangle$ & $1 + x^{3} + 3x^{4} + 10x^{5} + 26x^{6} + 21x^{7} + 2x^{8}$\\ 
\bottomrule
\end{tabularx}
\endgroup
\caption{All indecomposable $[[8,2,3]]$ code classes sorted by automorphism group order. \label{tab:codes823}}
\end{table}

There is only one $[[8,3,3]]$ code \cite{CRSS98}. It has $w(x)=1+28x^6+3x^8$, $|\aut{S}|=168$, and has no parent subsystem code.

\subsection{Error-correcting codes on 9 qubits}

There are 3411 indecomposable $[[9,1,3]]$ codes. Of these, 1448 are nondegenerate. The 19 indecomposable CSS codes are listed in Table~\ref{tab:codes913}. Four of these CSS codes have planar Tanner graphs for their low weight generating sets as shown in Figure~\ref{fig:913planar}. All four arise from gauge fixing the Bacon-Shor code \cite{bacon06,ac07} and are the only CSS gauge fixings of that code. 

\begin{table}
\renewcommand{\arraystretch}{1.55}
\centering
\begingroup
\begin{tabularx}{\textwidth}{l|l|L{24}|>{\hsize=4cm\RaggedRight} X}
\toprule
$\mathrm{Idx}$ & $|\mathrm{Aut}(S)|$ & $S$ & $w(x)$ \\ 
\specialrule{1.5pt}{1pt}{1pt}
8802 & 82944 & $\langle Z_{0}Z_{6}$, $Z_{1}Z_{6}$, $Z_{2}Z_{7}$, $Z_{3}Z_{8}$, $Z_{4}Z_{8}$, $Z_{5}Z_{7}$, $X_{0}X_{1}X_{2}X_{5}X_{6}X_{7}$, $X_{0}X_{1}X_{3}X_{4}X_{6}X_{8}\rangle$ & $1 + 9x^{2} + 27x^{4} + 75x^{6} + 144x^{8}$ \\ 
4280 & 9216 & $\langle Z_{0}Z_{7}$, $Z_{1}Z_{8}$, $X_{2}X_{6}$, $X_{3}X_{6}$, $Z_{4}Z_{5}$, $X_{0}X_{4}X_{5}X_{7}$, $X_{1}X_{4}X_{5}X_{8}$, $Z_{2}Z_{3}Z_{4}Z_{6}Z_{7}Z_{8}\rangle$ & $1 + 6x^{2} + 24x^{4} + 90x^{6} + 135x^{8}$ \\ 
4079 & 3072 & $\langle Z_{0}Z_{8}$, $X_{1}X_{7}$, $X_{2}X_{7}$, $X_{3}X_{5}$, $X_{4}X_{6}$, $Z_{3}Z_{4}Z_{5}Z_{6}$, $X_{0}X_{5}X_{6}X_{8}$, $Z_{1}Z_{2}Z_{3}Z_{5}Z_{7}Z_{8}\rangle$ & $1 + 6x^{2} + 24x^{4} + 90x^{6} + 135x^{8}$ \\ 
4395 & 1152 & $\langle Z_{0}Z_{8}$, $Z_{1}Z_{8}$, $Z_{2}Z_{3}Z_{6}Z_{7}$, $Z_{2}Z_{4}Z_{7}Z_{8}$, $Z_{5}Z_{6}Z_{7}Z_{8}$, $X_{2}X_{4}X_{5}X_{6}$, $X_{3}X_{4}X_{5}X_{7}$, $X_{0}X_{1}X_{2}X_{3}X_{5}X_{8}\rangle$ & $1 + 3x^{2} + 21x^{4} + 105x^{6} + 126x^{8}$ \\ 
8519 & 1024 & $\langle Z_{0}Z_{8}$, $X_{1}X_{7}$, $Z_{2}Z_{6}$, $X_{3}X_{4}$, $Z_{3}Z_{4}Z_{5}Z_{6}$, $X_{2}X_{5}X_{6}X_{7}$, $Z_{1}Z_{5}Z_{7}Z_{8}$, $X_{0}X_{3}X_{5}X_{8}\rangle$ & $1 + 4x^{2} + 22x^{4} + 100x^{6} + 129x^{8}$ \\ 
4335 & 576 & $\langle Z_{0}Z_{8}$, $Z_{1}Z_{8}$, $Z_{2}Z_{3}Z_{4}$, $Z_{3}Z_{5}Z_{6}$, $Z_{4}Z_{5}Z_{7}$, $X_{3}X_{4}X_{6}X_{7}$, $X_{2}X_{3}X_{5}X_{7}$, $X_{0}X_{1}X_{2}X_{3}X_{6}X_{8}\rangle$ & $1 + 3x^{2} + 4x^{3} + 9x^{4} + 24x^{5} + 49x^{6} + 84x^{7} + 66x^{8} + 16x^{9}$ \\ 
7419 & 384 & $\langle Z_{0}Z_{8}$, $Z_{1}Z_{6}$, $Z_{2}Z_{7}$, $Z_{3}Z_{4}Z_{6}Z_{8}$, $Z_{3}Z_{5}Z_{7}Z_{8}$, $X_{1}X_{3}X_{5}X_{6}$, $X_{2}X_{3}X_{4}X_{7}$, $X_{0}X_{4}X_{5}X_{8}\rangle$ & $1 + 3x^{2} + 21x^{4} + 105x^{6} + 126x^{8}$ \\ 
8816 & 384 & $\langle Z_{0}Z_{8}$, $Z_{1}Z_{6}$, $Z_{2}Z_{7}$, $X_{3}X_{4}X_{5}$, $Z_{3}Z_{4}Z_{6}Z_{8}$, $Z_{3}Z_{5}Z_{7}Z_{8}$, $Y_{1}Y_{2}X_{3}Y_{6}Y_{7}$, $Y_{0}Y_{1}X_{5}Y_{6}Y_{8}\rangle$ & $1 + 3x^{2} + x^{3} + 15x^{4} + 27x^{5} + 37x^{6} + 87x^{7} + 72x^{8} + 13x^{9}$ \\ 
9709 & 384 & $\langle Y_{0}Y_{1}Y_{4}Y_{5}$, $Y_{0}Y_{2}Y_{4}Y_{6}$, $Y_{0}Y_{3}Y_{4}Y_{7}$, $X_{0}X_{1}X_{4}X_{5}$, $X_{0}X_{2}X_{4}X_{6}$, $X_{0}X_{3}X_{4}X_{7}$, $X_{0}X_{1}X_{2}X_{3}X_{8}$, $Z_{1}Z_{2}Z_{3}Z_{4}Z_{8}\rangle$ & $1 + 18x^{4} + 16x^{5} + 56x^{6} + 96x^{7} + 53x^{8} + 16x^{9}$ \\ 
5477 & 256 & $\langle Z_{0}Z_{7}$, $Z_{1}Z_{8}$, $Z_{2}Z_{3}Z_{5}Z_{8}$, $Z_{2}Z_{4}Z_{5}Z_{7}$, $X_{3}X_{4}X_{5}X_{6}$, $Z_{5}Z_{6}Z_{7}Z_{8}$, $X_{0}X_{2}X_{3}X_{6}X_{7}$, $X_{1}X_{2}X_{4}X_{6}X_{8}\rangle$ & $1 + 2x^{2} + 20x^{4} + 16x^{5} + 46x^{6} + 96x^{7} + 59x^{8} + 16x^{9}$ \\ 
11001 & 192 & $\langle Z_{0}Z_{1}$, $X_{0}X_{1}X_{4}$, $Z_{2}Z_{3}Z_{5}Z_{7}$, $X_{2}X_{4}X_{5}X_{6}$, $Z_{5}Z_{6}Z_{7}Z_{8}$, $X_{3}X_{4}X_{6}X_{7}$, $X_{2}X_{3}X_{6}X_{8}$, $Z_{1}Z_{2}Z_{4}Z_{7}Z_{8}\rangle$ & $1 + x^{2} + 2x^{3} + 13x^{4} + 24x^{5} + 47x^{6} + 90x^{7} + 66x^{8} + 12x^{9}$ \\ 
12079 & 192 & $\langle Z_{0}Z_{8}$, $Z_{1}Z_{2}Z_{5}Z_{6}$, $Z_{1}Z_{3}Z_{4}Z_{6}$, $X_{1}X_{2}X_{4}X_{7}$, $X_{1}X_{3}X_{5}X_{7}$, $X_{2}X_{3}X_{6}X_{7}$, $X_{0}X_{3}X_{4}X_{8}$, $Z_{4}Z_{5}Z_{6}Z_{7}Z_{8}\rangle$ & $1 + x^{2} + 19x^{4} + 16x^{5} + 51x^{6} + 96x^{7} + 56x^{8} + 16x^{9}$ \\ 
9897 & 144 & $\langle Z_{0}Z_{1}Z_{4}Z_{8}$, $Z_{0}Z_{2}Z_{5}Z_{7}$, $Z_{1}Z_{3}Z_{5}Z_{7}$, $X_{1}X_{2}X_{4}X_{5}$, $Z_{0}Z_{4}Z_{5}Z_{6}$, $X_{1}X_{3}X_{4}X_{6}$, $X_{0}X_{3}X_{4}X_{7}$, $X_{0}X_{3}X_{5}X_{8}\rangle$ & $1 + 18x^{4} + 120x^{6} + 117x^{8}$ \\ 
5781 & 128 & $\langle Z_{0}Z_{8}$, $X_{1}X_{5}$, $Y_{2}Y_{3}Y_{6}Y_{7}$, $Z_{2}Z_{3}Z_{4}Z_{8}$, $X_{3}X_{4}X_{5}X_{6}$, $X_{2}X_{4}X_{5}X_{7}$, $Y_{1}Y_{2}Y_{4}Z_{5}Y_{7}$, $Y_{0}Y_{4}Y_{6}Y_{7}X_{8}\rangle$ & $1 + 2x^{2} + 12x^{4} + 32x^{5} + 46x^{6} + 80x^{7} + 67x^{8} + 16x^{9}$ \\ 
5784 & 128 & $\langle Z_{0}Z_{8}$, $X_{1}X_{5}$, $Z_{2}Z_{3}Z_{6}Z_{8}$, $Z_{2}Z_{4}Z_{7}Z_{8}$, $Z_{1}Z_{5}Z_{6}Z_{7}$, $X_{2}X_{4}X_{5}X_{6}$, $X_{2}X_{3}X_{5}X_{7}$, $X_{0}X_{3}X_{4}X_{8}\rangle$ & $1 + 2x^{2} + 20x^{4} + 110x^{6} + 123x^{8}$ \\ 
12003 & 96 & $\langle X_{0}X_{5}$, $Z_{1}Z_{4}Z_{6}$, $Z_{2}Z_{4}Z_{7}$, $Z_{3}Z_{4}Z_{8}$, $Y_{1}Y_{2}Y_{6}Y_{7}$, $Y_{1}Y_{3}Y_{6}Y_{8}$, $Z_{0}Z_{5}Z_{6}Z_{7}Z_{8}$, $X_{2}X_{3}X_{4}X_{5}X_{6}\rangle$ & $1 + x^{2} + 3x^{3} + 9x^{4} + 25x^{5} + 55x^{6} + 85x^{7} + 62x^{8} + 15x^{9}$ \\ 
6038 & 64 & $\langle Z_{0}Z_{8}$, $X_{1}X_{2}$, $X_{3}X_{5}X_{7}$, $X_{4}X_{6}X_{7}$, $Z_{1}Z_{2}Z_{3}Z_{5}$, $Z_{1}Z_{2}Z_{4}Z_{6}$, $Z_{3}Z_{6}Z_{7}Z_{8}$, $X_{0}X_{2}X_{3}X_{4}X_{8}\rangle$ & $1 + 2x^{2} + 2x^{3} + 12x^{4} + 26x^{5} + 46x^{6} + 86x^{7} + 67x^{8} + 14x^{9}$ \\ 
8124 & 12 & $\langle X_{2}X_{4}X_{5}$, $X_{1}X_{3}X_{6}$, $Z_{1}Z_{2}Z_{3}Z_{4}$, $Z_{2}Z_{3}Z_{5}Z_{6}$, $X_{0}X_{2}X_{3}X_{7}$, $Z_{0}Z_{4}Z_{5}Z_{7}$, $X_{0}X_{1}X_{4}X_{8}$, $Z_{0}Z_{2}Z_{5}Z_{8}\rangle$ & $1 + 2x^{3} + 12x^{4} + 24x^{5} + 52x^{6} + 90x^{7} + 63x^{8} + 12x^{9}$ \\ 
7810 & 8 & $\langle Z_{1}Z_{2}Z_{4}$, $X_{1}X_{2}X_{7}$, $Y_{0}Y_{3}Y_{6}Y_{8}$, $X_{0}X_{3}X_{5}X_{7}$, $X_{2}X_{3}X_{4}X_{6}$, $Z_{0}Z_{4}Z_{5}Z_{6}$, $Z_{1}Z_{3}Z_{6}Z_{7}$, $X_{0}X_{2}X_{4}X_{8}\rangle$ & $1 + 2x^{3} + 12x^{4} + 24x^{5} + 52x^{6} + 90x^{7} + 63x^{8} + 12x^{9}$ \\ 
\bottomrule
\end{tabularx}
\endgroup
\caption{All 19 of the indecomposable $[[9,1,3]]$ CSS codes, sorted by automorphism group order. The code with the largest automorphism group is Shor's code \cite{shor95}. The rotated surface code \cite{bombin07rsc} has $|\aut{S}|=1024$. \label{tab:codes913}}
\end{table}

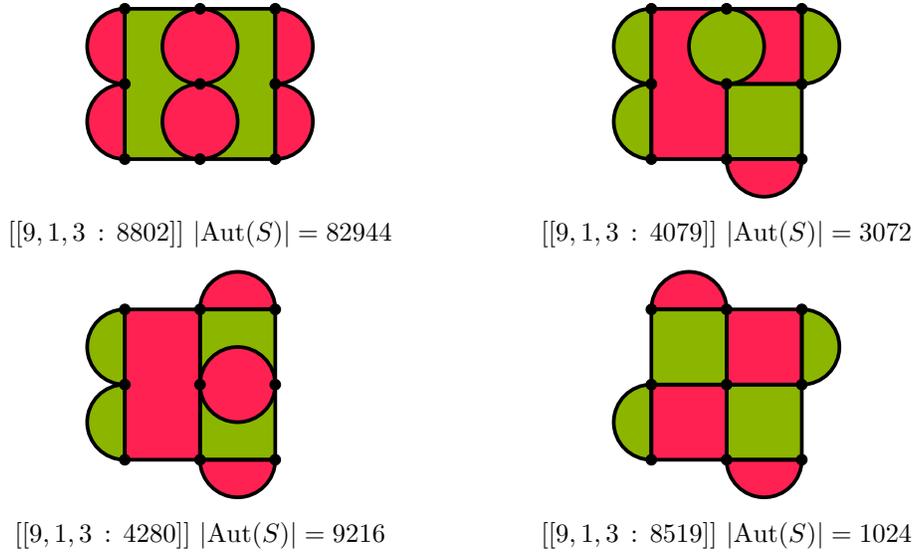
\begin{figure}[h!]
    \centering
    \begin{tikzpicture}

    \node at (1,3) {$\eclass{9}{1}{3}{8802}$ $|\aut{S}|=82944$};
    \begin{scope}[shift={(0,4)}]
        \coordinate (A) at (0,0);
        \coordinate (B) at (0,1);
        \coordinate (C) at (0,2);
        \coordinate (D) at (1,2);
        \coordinate (E) at (2,2);
        \coordinate (F) at (2,1);
        \coordinate (G) at (2,0);
        \coordinate (H) at (1,0);
        \coordinate (I) at (1,1);
        \coordinate (J) at ($(B)!0.5!(C)$);
        \coordinate (K) at ($(A)!0.5!(B)$);
        \coordinate (L) at ($(E)!0.5!(F)$);
        \coordinate (M) at ($(F)!0.5!(G)$);

        \fill[cXstab] (A) -- (C) -- (E) -- (G) -- cycle;
        \fill[cZstab] (B) arc (90:270:0.5) -- cycle;
        \fill[cZstab] (C) arc (90:270:0.5) -- cycle;
        \fill[cZstab] (E) arc(90:-90:0.5) -- cycle;
        \fill[cZstab] (F) arc(90:-90:0.5) -- cycle;
        \fill[cZstab] (D) arc(90:450:0.5) -- cycle;
        \fill[cZstab] (I) arc(90:450:0.5) -- cycle;

        \draw[strong] (A) -- (C) -- (E) -- (G) -- cycle;
        \draw[strong] (B) arc (90:270:0.5) -- cycle;
        \draw[strong] (C) arc (90:270:0.5) -- cycle;
        \draw[strong] (E) arc(90:-90:0.5) -- cycle;
        \draw[strong] (F) arc(90:-90:0.5) -- cycle;
        \draw[strong] (D) arc(90:450:0.5) -- cycle;
        \draw[strong] (I) arc(90:450:0.5) -- cycle;
        
        \foreach \corner in {A, B, C, D, E, F, G, H, I}
                \filldraw[black] (\corner) circle (2pt);
        
    \end{scope}

    \node at (1,-1) {$\eclass{9}{1}{3}{4280}$ $|\aut{S}|=9216$};
    \begin{scope}
        \coordinate (A) at (0,0);
        \coordinate (B) at (0,1);
        \coordinate (C) at (0,2);
        \coordinate (D) at (1,2);
        \coordinate (E) at (2,2);
        \coordinate (F) at (2,1);
        \coordinate (G) at (2,0);
        \coordinate (H) at (1,0);
        \coordinate (I) at (1,1);

        \fill[cXstab] (D) -- (E) -- (G) -- (H) -- cycle;
        \fill[cZstab] (C) -- (D) -- (H) -- (A) -- cycle;
        \fill[cXstab] (C) arc(90:270:0.5) -- cycle;
        \fill[cXstab] (B) arc(90:270:0.5) -- cycle;
        \fill[cZstab] (D) arc(180:0:0.5) -- cycle;
        \fill[cZstab] (I) arc(180:540:0.5) -- cycle;
        \fill[cZstab] (H) arc(180:360:0.5) --cycle;

        \draw[strong] (D) -- (E) -- (G) -- (H) -- cycle;
        \draw[strong] (C) -- (D) -- (H) -- (A) -- cycle;
        \draw[strong] (C) arc(90:270:0.5) -- cycle;
        \draw[strong] (B) arc(90:270:0.5) -- cycle;
        \draw[strong] (D) arc(180:0:0.5) -- cycle;
        \draw[strong] (I) arc(180:540:0.5) -- cycle;
        \draw[strong] (H) arc(180:360:0.5) --cycle;

        \foreach \corner in {A, B, C, D, E, F, G, H, I}
                \filldraw[black] (\corner) circle (2pt);  
        
    \end{scope}

    \node at (8,3) {$\eclass{9}{1}{3}{4079}$ $|\aut{S}|=3072$};
    \begin{scope}[shift={(7,4)}]
        \coordinate (A) at (0,0);
        \coordinate (B) at (0,1);
        \coordinate (C) at (0,2);
        \coordinate (D) at (1,2);
        \coordinate (E) at (2,2);
        \coordinate (F) at (2,1);
        \coordinate (G) at (2,0);
        \coordinate (H) at (1,0);
        \coordinate (I) at (1,1);

        \fill[cZstab] (A) -- (C) -- (D) -- (H) -- cycle;
        \fill[cZstab] (D) -- (E) -- (F) -- (I) -- cycle;
        \fill[cXstab] (I) -- (F) -- (G) -- (H) -- cycle;
        \fill[cXstab] (D) arc(90:450:0.5) -- cycle;
        \fill[cXstab] (C) arc(90:270:0.5) -- cycle;
        \fill[cXstab] (B) arc(90:270:0.5) -- cycle;
        \fill[cXstab] (E) arc(90:-90:0.5) -- cycle;
        \fill[cZstab] (H) arc(180:360:0.5) --cycle;

        \draw[strong] (I) -- (H) -- (A) -- (C) -- (D);
        \draw[strong] (D) -- (E) -- (F) -- (I);
        \draw[strong] (I) -- (F) -- (G) -- (H) -- cycle;

        \draw[strong] (D) arc(90:450:0.5) -- cycle;
        \draw[strong] (C) arc(90:270:0.5) -- cycle;
        \draw[strong] (B) arc(90:270:0.5) -- cycle;
        \draw[strong] (E) arc(90:-90:0.5) -- cycle;
        \draw[strong] (H) arc(180:360:0.5) --cycle;

        \foreach \corner in {A, B, C, D, E, F, G, H, I}
                \filldraw[black] (\corner) circle (2pt);  
        
    \end{scope}

    \node at (8,-1) {$\eclass{9}{1}{3}{8519}$ $|\aut{S}|=1024$};
    \begin{scope}[shift={(7,0)}]
        \coordinate (A) at (0,0);
        \coordinate (B) at (0,1);
        \coordinate (C) at (0,2);
        \coordinate (D) at (1,2);
        \coordinate (E) at (2,2);
        \coordinate (F) at (2,1);
        \coordinate (G) at (2,0);
        \coordinate (H) at (1,0);
        \coordinate (I) at (1,1);

        \fill[cXstab] (C) -- (D) -- (I) -- (B) -- cycle;
        \fill[cXstab] (I) -- (F) -- (G) -- (H) -- cycle;
        \fill[cZstab] (D) -- (E) -- (F) -- (I) -- cycle;
        \fill[cZstab] (B) -- (I) -- (H) -- (A) -- cycle;
        \fill[cXstab] (B) arc(90:270:0.5) -- cycle;
        \fill[cXstab] (E) arc(90:-90:0.5) -- cycle;
        \fill[cZstab] (D) arc(0:180:0.5) -- cycle;
        \fill[cZstab] (H) arc(180:360:0.5) -- cycle;

        \draw[strong] (C) -- (D) -- (I) -- (B) -- cycle;
        \draw[strong] (I) -- (F) -- (G) -- (H) -- cycle;
        \draw[strong] (D) -- (E) -- (F) -- (I) -- cycle;
        \draw[strong] (B) -- (I) -- (H) -- (A) -- cycle;
        \draw[strong] (B) arc(90:270:0.5) -- cycle;
        \draw[strong] (E) arc(90:-90:0.5) -- cycle;
        \draw[strong] (D) arc(0:180:0.5) -- cycle;
        \draw[strong] (H) arc(180:360:0.5) -- cycle;

        \foreach \corner in {A, B, C, D, E, F, G, H, I}
                \filldraw[black] (\corner) circle (2pt);  
        
    \end{scope}

    \end{tikzpicture}
    \caption{All planar $[[9,1,3]]$ codes classes.\label{fig:913planar}}
\end{figure}

There are 4425 indecomposbale $[[9,2,3]]$ code classes of which we found only three with planar Tanner graphs~\ref{codes923planar}. Of these 4425 indecomposable code classes only four have automorphism groups sizes at least 512 which are listed in Table~\ref{tab:codes923aut}. There are no indecomposable $[[9,2,3]]$ CSS codes and there are no $\mathrm{GF}(4)$-linear indecomposble $[[9,2,3]]$ code classes. The [[9,2,3]] code constructed in~\cite{Grassl_2017} is represented by the class $\eclass{9}{2}{3}{8842}$. 

\begin{table}
\renewcommand{\arraystretch}{1.55}
\centering
\begingroup
\begin{tabularx}{\textwidth}{l|l|L{24}|>{\hsize=4cm\RaggedRight} X}
\toprule
$\mathrm{Idx}$ & $|\mathrm{Aut}(S)|$ & $S$ & $w(x)$ \\ 
\specialrule{1.5pt}{1pt}{1pt}
15551 & 1152 & $\langle X_{0}Z_{8}$, $X_{1}Z_{7}$, $X_{2}Z_{7}$, $X_{3}X_{4}Z_{5}Z_{6}$, $Z_{3}Z_{4}Y_{5}Y_{6}$, $Z_{0}Z_{4}X_{5}Z_{6}X_{8}$, $Z_{1}Z_{2}Y_{3}Y_{5}Z_{6}X_{7}Z_{8}\rangle$ & $1 + 4x^{2} + 6x^{4} + 8x^{5} + 12x^{6} + 56x^{7} + 41x^{8}$ \\ 
80585 & 1152 & $\langle X_{0}X_{1}Z_{4}Z_{5}$, $X_{2}X_{3}Z_{6}Z_{7}$, $Y_{0}Y_{4}Z_{5}Z_{8}$, $X_{0}Z_{1}X_{5}Z_{8}$, $Y_{2}Y_{6}Z_{7}Z_{8}$, $X_{2}Z_{3}X_{7}Z_{8}$, $Y_{0}Z_{1}Y_{2}Z_{3}Z_{4}Z_{6}X_{8}\rangle$ & $1 + 14x^{4} + 16x^{6} + 64x^{7} + 33x^{8}$ \\ 
22646 & 768 & $\langle X_{0}Z_{7}$, $X_{1}Z_{8}$, $X_{3}X_{4}Z_{5}Z_{6}$, $Z_{2}Y_{3}Y_{5}Z_{6}$, $Z_{2}X_{3}Z_{4}X_{6}$, $Z_{0}Z_{1}Y_{7}Y_{8}$, $Z_{0}Y_{2}Y_{3}Z_{4}Z_{5}X_{7}Z_{8}\rangle$ & $1 + 2x^{2} + 12x^{4} + 14x^{6} + 64x^{7} + 35x^{8}$ \\ 
53565 & 512 & $\langle X_{0}Z_{8}$, $X_{1}Z_{6}$, $X_{2}Z_{7}$, $X_{3}X_{4}$, $Z_{0}Y_{3}Z_{4}Z_{5}Y_{8}$, $Z_{2}Z_{3}Z_{4}X_{5}Z_{6}X_{7}$, $Z_{0}Z_{1}X_{5}X_{6}Z_{7}X_{8}\rangle$ & $1 + 4x^{2} + 6x^{4} + 4x^{5} + 20x^{6} + 56x^{7} + 33x^{8} + 4x^{9}$ \\ 
\bottomrule
\end{tabularx}
\endgroup
\caption{All indecomposable $[[9,2,3]]$ code classes with automorphism group order at least 512. \label{tab:codes923aut}}
\end{table}

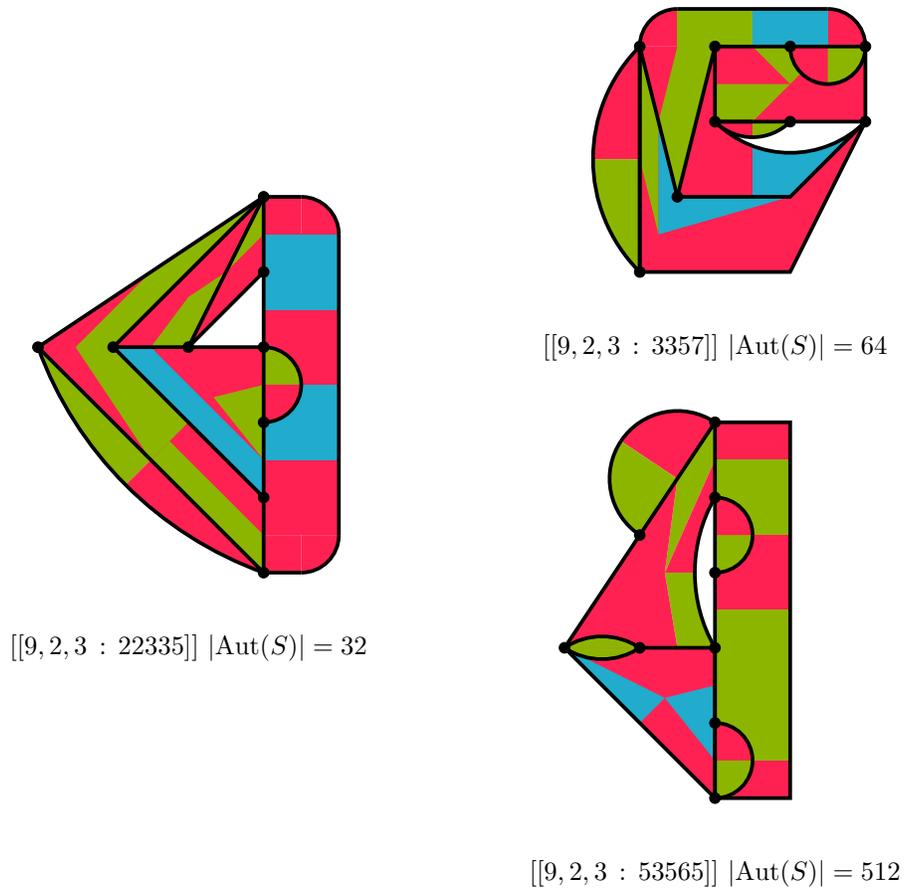
\begin{figure}
    \centering
    \begin{tikzpicture}
        \node at (8,-3) {$\eclass{9}{2}{3}{53565}$ $|\aut{S}|=512$};
        \begin{scope}[shift={(6,0)}]
            \coordinate (A) at (0,0); %
            \coordinate (C) at (2,3);
            \coordinate (B) at ($(A)!0.5!(C)$);
            \coordinate (D) at (2,2);
            \coordinate (E) at (2,1);
            \coordinate (F) at (2,0);
            \coordinate (G) at (2,-1);
            \coordinate (H) at (2,-2);
            \coordinate (I) at (1,0);
            \coordinate (J) at (3,3);
            \coordinate (K) at (3,-2);
            \coordinate (L) at ($(B)!0.5!(C)$);
            \coordinate (M) at ($(A)!0.5!(B)$);
            \coordinate (N) at ($(A)!0.5!(H)$);

            \coordinate (AB) at (3, 2.5);
            \coordinate (AC) at (3, 1.5);
            \coordinate (AD) at (3, 0.5);
            \coordinate (AE) at (3, -0.5);
            \coordinate (AF) at (3, -1.5);
            
            \coordinate (O) at ($(H)!0.5!(G)$);
            \coordinate (P) at ($(G)!0.5!(F)$);
            \coordinate (Q) at ($(F)!0.5!(E)$);
            \coordinate (R) at ($(E)!0.5!(D)$);
            \coordinate (S) at ($(D)!0.5!(C)$);
            
            \coordinate (T) at ($(G)!0.5!(H)$);
            \coordinate (U) at ($(F)!0.5!(G)$);
            \coordinate (V) at ($(I)!0.5!(F)$);
            \coordinate (W) at (barycentric cs:A=1,C=1,F=1);
            \coordinate (X) at (barycentric cs:A=1,F=1,H=1);
            \coordinate (Y) at ($(C)!0.5!(D)$);
            \coordinate (Z) at ($(D)!0.5!(E)$);
            \coordinate (AA) at ($(E)!0.5!(F)$);

            \fill[cZstab] (C) arc(56.31:146.31:0.9014) -- (L) -- cycle;
            \fill[cXstab] (B) arc(236.31:146.31:0.9014) -- (L) -- cycle;

            \fill[cZstab] (C) -- (J) -- (AB) -- (S) -- cycle;
            \fill[cXstab] (S) -- (AB) -- (AC) -- (R) -- cycle;

            \fill[cZstab] (AC) -- (R) -- (Q) -- (AD) -- cycle;
            
            \fill[cXstab] (Q) -- (AD) -- (AF) -- (O) -- cycle;

            \fill[cZstab] (AF) -- (O) -- (H) -- (K) -- cycle;

            \fill[cXstab] (C) -- (L) -- (W) -- (Y) -- cycle;
            
            \fill[cZstab] (D) arc(90:0:0.5) -- (R) -- cycle;
            \fill[cXstab] (E) arc(-90:0:0.5) -- (R) -- cycle;

            \fill[cZstab] (G) arc(90:0:0.5) -- (O) -- cycle;
            \fill[cXstab] (H) arc(-90:0:0.5) -- (O) -- cycle;

            \fill[cXstab] (W) -- (E) -- (F) -- (V) -- cycle;

            \fill[cZstab] (L) -- (A) -- (V) -- (W) -- cycle;
            \fill[cZstab] (Y) -- (W) -- (E) -- cycle;

            \fill[cZstab] (A) -- (F) -- (H) -- cycle;
            \fill[cYstab] (X) -- (A) -- (N) -- cycle;
            \fill[cYstab] (X) -- (U) -- (O) -- cycle;

            \fill[cXstab] (A) arc(236.31:303.69: 0.901) -- cycle;
            \fill[cXstab] (A) arc(123.69:56.31: 0.901) -- cycle;

            \fill[white] (D) arc(150:210:2) -- cycle;

            \draw[strong] (A) -- (C) -- (J) -- (K) -- (H) -- (C);
            \draw[strong] (I) -- (F) -- (H) -- (A);

            \draw[strong] (A) arc(236.31:303.69: 0.901);
            \draw[strong] (A) arc(123.69:56.31: 0.901);
            \draw[strong] (D) arc(150:210:2);
            \draw[strong] (C) arc(56.31:146.31:0.9014);
            \draw[strong] (B) arc(236.31:146.31:0.9014);
            \draw[strong] (D) arc(90:0:0.5);
            \draw[strong] (E) arc(-90:0:0.5);
            \draw[strong] (G) arc(90:0:0.5);
            \draw[strong] (H) arc(-90:0:0.5);

            \foreach \corner in {A, B, C, D, E, F, G, H, I}
                \filldraw[black] (\corner) circle (2pt);
                
        \end{scope}

        \node at (8,4) {$\eclass{9}{2}{3}{3357}$ $|\aut{S}|=64$};
        \begin{scope}[shift={(7,5)}]

        \coordinate (A) at (0,0);
        \coordinate (B) at (0,3);
        \coordinate (C) at (3,3);
        \coordinate (D) at (2,3);
        \coordinate (E) at (1,3);
        \coordinate (F) at (3,2);
        \coordinate (G) at (2,2);
        \coordinate (H) at (1,2);
        \coordinate (I) at (0.5,1);
        \coordinate (J) at ($(A)!0.5!(B)$);
        \coordinate (K) at ($(D)!0.5!(C)$);
        \coordinate (L) at ($(E)!0.5!(D)$);
        \coordinate (M) at ($(C)!0.5!(F)$);
        \coordinate (N) at ($(G)!0.5!(F)$);
        \coordinate (O) at ($(H)!0.5!(G)$);
        \coordinate (P) at ($(E)!0.5!(H)$);
        \coordinate (Q) at ($(I)!0.5!(B)$);
        \coordinate (R) at (2,0);
        \coordinate (S) at (2,1);
        \coordinate (T) at (0.25,0.5);
        \coordinate (U) at (0.5,3.5);
        \coordinate (V) at (2.5,3.5);
        \coordinate (W) at (0.5,3);
        \coordinate (X) at (1.5,3.5);
        \coordinate (Y) at (1.5,1);
        \coordinate (Z) at (2,2.5);

        \fill[cZstab] (A) -- (B) -- (C) -- (F) -- (R) -- cycle;

        \fill[cZstab] (B) arc(135:180:2.1213) -- (J) -- cycle;
        \fill[cXstab] (A) arc(225:180:2.1213) -- (J) -- cycle;     

        \fill[cXstab] (T) -- (J) -- (B) -- (Q) -- cycle;
        \fill[cYstab] (T) -- (Q) -- (I) -- (S) -- cycle;
        \fill[cYstab] (F) -- (S) -- (Y) -- (O) -- cycle;

        \fill[white] (H) arc(225:315:1.4142) -- cycle;
        \fill[cXstab] (G) arc(-45:-90:0.7071) -- (O) -- cycle;
        \fill[cZstab] (H) arc(225:270:0.7071) -- (O) -- cycle;

        \fill[cXstab] (Z) -- (O) -- (H) -- (P) -- cycle;
        \fill[cXstab] (Z) -- (L) -- (K) -- cycle;

        \fill[cZstab] (D) arc(180:270:0.5) -- (K) -- cycle;
        \fill[cXstab] (C) arc(0:-90:0.5) -- (K) -- cycle;

        \fill[cZstab] (B) arc(180:90:0.5) -- (W) -- cycle;
        \fill[cZstab] (C) arc(0:90:0.5) -- (K) -- cycle;
        \fill[cYstab] (V) -- (K) -- (L) -- (X) -- cycle;
        \fill[cXstab] (L) -- (X) -- (U) -- (W) -- cycle;
        \fill[cXstab] (E) -- (I) -- (Q) -- (W) -- cycle;

        \draw[strong] (B) -- (A) -- (R) -- (F);
        \draw[strong] (F) -- (S) -- (I) -- (B);
        \draw[strong] (I) -- (E) -- (C) -- (F) -- (H) -- (E);
        \draw[strong] (B) arc(135:225:2.1213) -- (J) -- cycle;
        \draw[strong] (H) arc(225:315:1.4142);
        \draw[strong] (G) arc(-45:-135:0.7071) -- (O);
        \draw[strong] (D) arc(180:360:0.5);

        \draw[strong] (B) arc(180:90:0.5);
        \draw[strong] (C) arc(0:90:0.5);
        \draw[strong] (U) -- (V);

        \foreach \corner in {A, B, C, D, E, F, G, H, I}
                \filldraw[black] (\corner) circle (2pt);
            
        \end{scope}

        \node at (1,0) {$\eclass{9}{2}{3}{22335}$ $|\aut{S}|=32$};
        \begin{scope}[shift={(-1,4)}]
            \coordinate (A) at (0,0);
            \coordinate (B) at (3,2);
            \coordinate (C) at (3,-3);
            \coordinate (D) at (3,1);
            \coordinate (E) at (3,0);
            \coordinate (F) at (3,-1);
            \coordinate (G) at (3,-2);
            \coordinate (H) at (1,0);
            \coordinate (I) at (2,0);
            \coordinate (J) at ($(A)!0.5!(C)$);
            \coordinate (K) at ($(G)!0.5!(H)$);
            \coordinate (L) at ($(C)!0.5!(G)$);
            \coordinate (M) at ($(F)!0.5!(G)$);
            \coordinate (N) at ($(E)!0.5!(F)$);
            \coordinate (O) at ($(D)!0.5!(E)$);
            \coordinate (P) at ($(B)!0.5!(D)$);
            \coordinate (Q) at ($(A)!0.5!(B)$);
            \coordinate (R) at (0.5,0);
            \coordinate (S) at ($(I)!0.5!(B)$);
            \coordinate (T) at (barycentric cs:B=1,H=1,I=1);
            \coordinate (U) at ($(H)!0.5!(I)$);
            \coordinate (V) at (barycentric cs:H=1,E=1,G=1);
            \coordinate (W) at (3.5,2);
            \coordinate (X) at (4,1.5);
            \coordinate (Y) at (4,0.5);
            \coordinate (Z) at (4,-0.5);
            \coordinate (AA) at (4,-1.5);
            \coordinate (AB) at (4,-2.5);
            \coordinate (AC) at (3.5,-3);
            \coordinate (AD) at ($(J)!0.5!(K)$);

            \fill[cZstab] (P) -- (X) -- (AB) -- (L) -- cycle;
            \fill[cZstab] (B) -- (W) -- (3.5,1.5) -- (P) -- cycle;
            \fill[cZstab] (C) -- (AC) -- (3.5,-2.5) -- (L) -- cycle;
            \fill[cZstab] (W) arc(90:0:0.5) -- (3.5,1.5) -- cycle;
            \fill[cZstab] (AC) arc(-90:0:0.5) -- (3.5,-2.5) -- cycle;
            \fill[cYstab] (X) -- (Y) -- (O) -- (P) -- cycle;
            \fill[cYstab] (Z) -- (AA) -- (M) -- (N) -- cycle;
            \fill[cXstab] (E) arc(90:0:0.5) -- (N) -- cycle;
            \fill[cZstab] (F) arc(-90:0:0.5) -- (N) -- cycle;

            \fill[cZstab] (A) -- (B) -- (C) -- cycle;
            \fill[white] (D) -- (I) -- (E) -- cycle;
            \fill[cXstab] (B) -- (S) -- (P) -- cycle;
            \fill[cXstab] (T) -- (S) -- (I) -- (U) -- cycle;
            \fill[cXstab] (B) -- (H) -- (K) -- (J) -- (R) -- (Q) -- cycle;
            \fill[cXstab] (AD) -- (L) -- (C) -- (J) -- cycle;
            \fill[cYstab] (H) -- (U) -- (M) -- (M) -- (G) -- cycle;

            \fill[cXstab] (A) arc(199.89:225:5) -- (J) -- cycle;
            \fill[cZstab] (C) arc(250.11:225:5) -- (J) -- cycle;
            
            \fill[cXstab] (V) -- (N) -- (M) -- cycle;

            \draw[strong] (A) -- (B) -- (C) -- cycle;
            \draw[strong] (B) -- (H) -- (E);
            \draw[strong] (H) -- (G);
            \draw[strong] (B) -- (W);
            \draw[strong] (X) -- (AB);
            \draw[strong] (AC) -- (C);
            \draw[strong] (B) -- (I);
            \draw[strong] (I) -- (D);
            \draw[strong] (W) arc(90:0:0.5);
            \draw[strong] (AC) arc(-90:0:0.5);
            \draw[strong] (E) arc(90:-90:0.5);
            \draw[strong] (A) arc(199.89:250.11:5);

            \foreach \corner in {A, B, C, D, E, F, G, H, I}
                \filldraw[black] (\corner) circle (2pt);
        \end{scope}
    \end{tikzpicture}
    \caption{All planar $[[9,2,3]]$ code classes.\label{codes923planar}}
\end{figure}

There are 221 indecomposable $[[9,3,3]]$ codes. Among these, there are no CSS codes, but there is one GF(4)-linear code with $|\mathrm{Aut}(S)|=1296$. The next largest automorphism group among $[[9,3,3]]$ codes is $|\aut{S}|=324$.

\subsection{Sporadic examples}

\paragraph{Smallest error-correcting subsystem code}
The only indecomposable $[[6,1,3]]$ code is well-known \cite{CRSS98}. It is the $\tilde{X}$ gauge-fixing of
\begin{equation}
G=\langle Z_1Y_2Z_3Z_4Y_6, Z_1Z_2Y_3Z_5Y_6, Z_1Z_3X_4Z_5X_6, Z_1Z_2Z_4X_5X_6, \tilde{X}=X_1Z_6, \tilde{Z}=Z_1 \rangle.
\end{equation}
Projecting onto the $+1$ eigenspace of $\tilde{Z}$ yields the 5-qubit code.

\paragraph{Smallest ``least symmetric'' code}
Codes that have only the identity automorphism are the ``least symmetric'' codes in this sense. The smallest occurrence of a stabilizer code with no automorphisms outside of the Pauli group is a $[[7,1,2]]$ code
\begin{equation}
S=\langle X_1X_2Z_5, X_3Z_6Z_7, X_4Z_5Z_6, Z_3Z_4X_6, Y_2Z_4Y_5Z_7, Z_1Z_2Z_3X_7\rangle, 
\end{equation}
with $w(x)=1+4x^3+9x^4+20x^5+22x^6+8x^7$. There are also 11 $[[7,2,2]]$ codes with this property.

\paragraph{Only indecomposable [[8,1,3]] CSS code}
This code is defined by
\begin{equation}
S=\langle Z_1Z_7,Z_2Z_6Z_7Z_8, Z_3Z_5Z_7Z_8, Z_4Z_5Z_6Z_7, X_2X_4X_5X_8, X_3X_4X_6X_8, X_1X_2X_3X_4X_7\rangle
\end{equation}
with $|\aut{S}|=192$ and $w(x)=1+x^2+17x^4+16x^5+39x^6+48x^7+6x^8$. This code appears on the descending part of the fault-tolerant code conversion path that was discovered in \cite{HFWH13}. The code is a gauge-fixing of two different $[[8,1,3]]$ subsystem codes
\begin{align}
G_1 = \langle & X_1X_4X_5X_6X_7, X_2X_4X_5X_8, X_3X_4X_6X_8, \\
& Z_1Z_2Z_3Z_4, Z_1Z_3Z_5Z_8, Z_1Z_2Z_6Z_8, \bar{X}_1=X_7, \bar{Z}_1=Z_1Z_7\rangle \\
G_2 = \langle & X_1X_4X_5X_6X_7, X_3X_4X_6X_8, Z_1Y_2Z_3Y_4X_5X_8, \\
& Z_1Z_3Z_5Z_8, Z_1Z_2Z_6Z_8, Z_1X_2X_4X_5Z_7X_8, \\
& \bar{X}_1=X_2X_4X_5X_8, \bar{Z}_1=Z_2X_4X_7\rangle.
\end{align}
Gauge-fixing the code $G_1$ in the X-basis gives the Steane code.

\paragraph{$GF(4)$-linear [[8,2,3]] code}

The $[[8,2,3]]$ code with the largest automorphism group, $|\aut{S}|=1728$, is at the top of Table~\ref{tab:codes823}. This code has an exceptionally short fault-tolerant syndrome measurement sequence \cite{dr20short}. It is also GF(4)-linear (but not CSS) and therefore $R^{\otimes 8}$ is a valid transversal gate. Choosing $\bar{X}_1=Z_1Z_2Z_6X_7$, $\bar{X}_2=Z_2Z_5Z_6X_8$, $\bar{Z}_1=Z_1Z_2Z_3Z_7$, and $\bar{Z}_2=Z_1Z_2Z_4Z_5Z_6Z_8$, the transversal $R$ gate applies the logical gate $\Lambda(X)_{2,1}\Lambda(X)_{1,2}$.

\paragraph{Only indecomposable [[8,2,3]] subsystem code}

Only three of the codes in Table~\ref{tab:codes823} have parent subsystem codes. In fact, the code with $|\aut{S}|=48$ and the two codes with $|\aut{S}|=6$ all have the same parent subsystem code
\begin{align}
G = \langle & X_1Z_2Z_3Z_4Z_6X_7, X_2Y_3Z_4Z_5X_7Y_8, Z_1X_2Z_3Y_4Y_7X_8, Z_1X_2Z_4X_5Z_6Z_8, Z_2Z_3Z_5X_6Z_7X_8 \\
& \tilde{X}_1=Z_1X_2Z_3Z_5Z_6X_8, \tilde{Z}_1=Z_2Z_3Z_4Z_5\rangle.
\end{align}

\paragraph {$GF(4)$-linear [[9,3,3]] code}
There is one GF(4)-linear code $[[9,3,3]]$ code defined by
\begin{gather}
\langle Y_1Z_2Z_5X_7X_8Z_9,Z_1X_2Z_4Z_6X_7X_8,Z_2Y_3Z_4Z_5Y_7X_9,\\
Z_1Z_3X_4X_7Y_8Z_9,Z_4X_5Z_6Z_7Z_8X_9,Z_1Z_3Z_4Z_5Y_6Z_7 \rangle.
\end{gather}
It is nondegenerate and has $w(x)=1+36x^6+27x^8$ and $|\mathrm{Aut}(S)|=1296$. The code has no parent subsystem code. This code differs from the $[[9,3,3]]$ code in \cite{Grassl:codetables} which has $|\mathrm{Aut}(S)|=336$ and $w(x)=1+x+28x^6+28x^7+3x^8+3x^9$.

\paragraph {``Most symmetric'' [[10,1,4]] code}
The indecomposable $[[10,1,4]]$ code with the largest automorphism group, $|\aut{S}|=20480$, is defined by
\begin{gather}
\langle X_1Z_9, X_2Z_8, X_3Z_7, X_4Z_{10}, X_5X_6,
Z_4Y_5Z_6Z_7Z_8Y_{10}, \\
Z_3Z_4X_7Z_8Z_9X_{10}, Z_2Z_5Z_6Z_7X_8Z_9, Z_1Z_5Z_6Z_8X_9Z_{10}\rangle
\end{gather}
and has $w(x)=1+5x^2+10x^4+170x^6+245x^8+81x^{10}$. The code is not equivalent to a CSS or GF(4)-linear code, and the low weight generating set given above does not have a planar Tanner graph. The parent subsystem code is
\begin{gather}
G = \langle X_1X_2Z_8Z_9, X_1X_3Z_7Z_9, X_1X_4Z_9Z_{10}, Z_4Y_5Z_6Z_7Z_8Y_{10}, \\
Z_4Z_5Y_6Z_7Z_8Y_{10}, Z_3Z_4X_7Z_8Z_9X_{10}, Z_2Z_5Z_6Z_7X_8Z_9, \\ Z_1Z_5Z_6Z_8X_9Z_{10}, \tilde{X}_1=X_1Z_9, \tilde{Z}_1=Z_1Z_2Z_3Z_4\rangle.
\end{gather}
This code differs from the $[[10,1,4]]$ code in \cite{Grassl:codetables}, which has $|\aut{S}|=9$ and $w(x)=1+36x^5+90x^6+120x^7+135x^8+100x^9+30x^{10}$.

\section{Applications}

\subsection{Gauge-fixing relationships among 7-qubit codes}

\begin{figure}
    \centering
    \begin{tikzpicture}[node distance=2cm and 2cm,square/.style={regular polygon,regular polygon sides=4}]
        \tikzstyle{greenrect} = [rectangle,thick,left color=green!20,right color=green!10,text=black,draw=green!10,minimum size=8mm]
        \tikzstyle{blueball} = [circle,thick,left color=blue!20,right color=blue!10,text=black,draw=blue!10,minimum size=8mm]
        \tikzstyle{redball} = [circle,thick,left color=red!20,right color=red!10,text=black,draw=red!10,minimum size=8mm]

        \node [blueball] (n642) {166};
        \node [greenrect] (nC1) [below of=n642] {C1};
        \node [blueball] (n322) [left of=nC1] {257};
        \node [blueball] (n321) [right of=nC1] {255};
        \node [greenrect] (nA1) [below of=n322] {A1};
        \node [blueball] (n768) [left of=nA1] {$\underset{\textrm{bare}}{108}$};
        \node [blueball] (n4) [right of=nA1] {185};
        \node [greenrect] (nB2) [right of=n4] {B2};
        \node [greenrect] (nbridge) [right of=nB2] {};
        \node [greenrect] (nA2) [right of=nbridge] {A2};
        \node [blueball] (n96) [above of=nA2] {239};
        \node [blueball] (n162) [below of=nA1] {240};
        \node [greenrect] (nD1) [below of=n4] {D1};
        \node [blueball] (n641) [below of=nB2] {164};
        \node [blueball] (n576) [below of=nA2] {$\underset{5-7}{115}$};
        \node [blueball] (n6) [below of=nD1] {200};
        \node [greenrect] (nB1) [right of=n6] {B1};
        \node [blueball] (n42) [right of=nB1] {227};
        \node [redball] (n1008) [left of=n6] {$\underset{\textrm{Steane}}{226}$};
        \node [redball] (n144) [left of=n1008] {228};
        \node [redball] (n48) [below of=n144] {209};
        \node [redball] (n161) [below of=n1008] {221};
        \node [blueball] (n192) [below of=nB1] {$\underset{\Delta}{190}$};

        \path[->]
            (nC1) edge[] node [scale=0.8,midway,right] {$\tilde{Z}$} (n642)
            (nC1) edge[] node [scale=0.8,midway,above] {$\tilde{X}$} (n322)
            (nC1) edge[] node [scale=0.8,midway,right] {$\tilde{Y}$} (n4)
            (nA1) edge[] node [scale=0.8,midway,right] {$\tilde{X}$} (n322)
            (nA1) edge[] node [scale=0.8,midway,above] {$\tilde{Z}$} (n768)
            (nA1) edge[] node [scale=0.8,midway,right] {$\tilde{Y}$} (n162)
            (nB2) edge[] node [scale=0.8,midway,right] {$\begin{smallmatrix}\tilde{Y}_1\tilde{X}_2,\\ \tilde{Z}_1\tilde{Z}_2\end{smallmatrix}$} (n321)
            (nB2) edge[] node [scale=0.8,midway,above] {$\begin{smallmatrix} \tilde{Y}_1\tilde{X}_2,\\ \tilde{X}_1\tilde{Z}_2\end{smallmatrix}$} (n4)
            (nB2) edge[] node [scale=0.8,midway,right] {$\begin{smallmatrix}\tilde{Y}_1\tilde{X}_2,\\ \tilde{Y}_1\end{smallmatrix}$} (n641)
            (nB2) edge[] node [scale=0.8,midway,above] {$\tilde{X}_1$} (nbridge)
            (nD1) edge[] node [scale=0.8,midway,right] {$\tilde{Y}$} (n4)
            (nD1) edge[] node [scale=0.8,midway,above] {$\tilde{Z}$} (n162)
            (nD1) edge[] node [scale=0.8,midway,right] {$\tilde{X}$} (n6)
            (nB1) edge[] node [scale=0.8,midway,above] {$\tilde{Y}$} (n6)
            (nB1) edge[] node [scale=0.8,midway,above] {$\tilde{X}$} (n42)
            (nB1) edge[] node [scale=0.8,midway,right] {$\tilde{Z}$} (n192)
            (nA2) edge[] node [scale=0.8,midway,right] {$\begin{smallmatrix}\tilde{Y}_1\tilde{X}_2,\\ \tilde{Z}_1\tilde{Y}_2\end{smallmatrix}$} (n96)
            (nA2) edge[] node [scale=0.8,midway,above] {$\tilde{X}_1$} (nbridge)
            (nA2) edge[] node [scale=0.8,midway,right] {$\begin{smallmatrix} \tilde{Y}_1\tilde{X}_2,\\ \tilde{X}_2\end{smallmatrix}$} (n576)
        ;
    \end{tikzpicture}
    \caption{Illustration of parent-child gauge-fixing relationships between $[[7,1,3]]$ subsystem and stabilizer codes. Each node represents an equivalence class of codes. Boxes correspond to subsystem codes, which are either unlabeled, or labeled as described in the text. Circles correspond to stabilizer codes. These are labeled by their index in Table~\ref{tab:codes7}. Each edge is labeled by the gauge operators that are fixed to have a $+1$ eigenvalue.
	  \label{fig:subsystem7}}
\end{figure}

\begin{table}
	\centering
	\begin{tabular}{ccc}
		label & $w(x)$ & $G$ \\ \midrule
        A2 & $1+2x+9x^2+4x^3+31x^4+$ &
		$\langle Y_1Z_2Y_3Z_4Z_5X_7, Y_1Z_2Z_3Y_4Z_6X_7, Y_1Z_2Z_4X_5Z_6Y_7,$ \\
        & $78x^5+119x^6+12x^7$ & $Y_1Z_2Z_3Z_5X_6Y_7, \tilde{X}_1=Z_2, \tilde{Z}_1=X_2Z_7, \tilde{X}_2=Z_1Z_7, \tilde{Z}_2=X_1Z_7 \rangle$ \\
        B2 & $1+2x+5x^2+4x^3+39x^4+$ &
		$\langle Y_2X_3Z_4Z_6Y_7, Y_2Z_3X_4Z_5Y_7, Y_2Z_3Y_5Z_6X_7,$ \\
        & $94x^5+83x^6+28x^7$ & $Z_1Y_2Z_4Z_5Y_6X_7,
		\tilde{X}_1=Z_1, \tilde{Z}_1=X_1Z_6, \tilde{X}_2=Z_2Z_7, \tilde{Z}_2=X_2Z_7\rangle$ \\
        A1 & $1+5x^2+2x^3+23x^4+$ & $\langle X_2Z_5, X_1X_3Z_5Z_6, X_1X_4Z_5Z_7, Z_1Z_2Z_4X_5Z_6X_7,$ \\
        & $24x^5+67x^6+6x^7$ & $Z_3Z_4Y_6Y_7, \tilde{X}_1=Z_1Z_3Z_4Z_5Z_6Z_7, \tilde{Z}_1=X_1Z_5\rangle$ \\
        B1 & $1+3x^2+2x^3+27x^4+$ & $\langle X_1X_2Z_6Z_7, X_1X_3Z_5Z_6, Z_2Y_4Z_6Y_7, Z_2Z_3Z_4X_5Z_6X_7,$ \\
        & $24x^5+65x^6+6x^7$& $Z_1Z_4Z_5X_6, \tilde{X}_1=Z_1Z_2Z_3Z_4Z_5Z_7, \tilde{Z}_1=X_1Z_6\rangle$ \\
        C1 & $1+3x^2+2x^3+27x^4+$ & $\langle X_1X_2Z_4Z_7, Y_1Y_3Z_5X_7, Z_1Z_2Z_3X_4Z_6X_7, Z_4X_5Z_6Z_7,$ \\
        & $24x^5+65x^6+6x^7$ & $Z_3Z_5X_6Z_7, \tilde{X}_1=Z_1Z_2Z_4Z_7, \tilde{Z}_1=X_1Z_7 \rangle$ \\
        D1 & $1+x^2+6x^3+19x^4+$ & $\langle Y_1Y_2Z_3Z_5X_7, Z_1Z_2Y_3Z_6Y_7, Z_1Z_2Z_3X_4Z_5Z_6X_7,$ \\
        & $44x^5+43x^6+14x^7$ & $X_1Z_2Z_4X_5Z_7, X_1Z_3Z_4X_6Z_7, \tilde{X}_1=Z_1Z_3Z_4Z_5Z_6Z_7, \tilde{Z}_1=X_1Z_7 \rangle$
    \end{tabular}
	\caption{Some $[[7,1,3]]$ subsystem codes.\label{tab:sub7}}
\end{table}

By fixing gauge operators of the $[[7,1,3]]$ subsystem codes in Table~\ref{tab:sub7} to their $+1$ eigenspaces, we produce 12 of the 16 codes in Table~\ref{tab:codes7}, as illustrated in Figure~\ref{fig:subsystem7}. Table~\ref{tab:sub7} is not an exhaustive list of all inequivalent $[[7,1,3]]$ subsystem codes.

\subsection{Transforming a color code into a surface code}\label{app:path}

As an application of the results in Table~\ref{tab:codes913}, we found a sequence of 8 error-correcting subsystem codes that links the 7-qubit color code (plus two ancillas) to the rotated surface code by gauge-fixing relationships.  The sequence does not include any codes with more than 9 qubits. It was constructed in an ad-hoc manner by inspecting the gauge-freedoms of $[[9,1,3]]$ CSS codes, so we have no reason to believe it is optimal. For example, a shorter path may exist through stabilizer codes that do not have CSS representatives. In principle an optimal sequence could be found using the complete database of 9-qubit codes.

Table~\ref{tab:path} lists the subsystem code sequence. The 7-qubit color code with two ancillas is the $X$-gauge of code 1. The second code in Table~\ref{tab:codes7} with $|\aut{S}|=128$ is both the $Z$-gauge of code 1 and the Y-gauge of code 2. The CSS code with $|\aut{S}|=8$ is $Z$-gauge of both codes 2 and 3. The first code in Table~\ref{tab:codes7} with $|\aut{S}|=192$ is both the $Y$-gauge of code 3 and the $X$-gauge of code 4. The CSS code with $|\aut{S}|=256$ is both the $Z$-gauge of codes 4 and 5. The CSS code with $|\aut{S}|=64$ is the $X$ gauge of both codes 5 and 6. The second code in Table~\ref{tab:codes7} with $|\aut{S}|=384$ is both the $Z$-gauge of code 6 and the $X$-gauge of code 7. The Shor code\footnote{Apply Hadamard to qubits 2-8 to reach the usual presentation.} is the $Z$-gauge of both codes 7 and 8. The rotated surface code is the $X$-gauge of code 8.

\begin{table}
    \renewcommand{\arraystretch}{1.55}
	\centering
	\begin{tabularx}{\textwidth}{l|L{30}|l}
		\toprule
		Seq & Stabilizer & $\bar{X}_g, \bar{Z}_g$ pairs \\ \midrule
        1 & $X_{1}X_{3}X_{5}X_{7}X_{9}, X_{3}X_{4}X_{7}X_{8}, X_{3}X_{4}X_{5}X_{6}$, $Z_{3}Z_{5}Z_{7}Z_{9}, Z_{2}Z_{3}Z_{4}Z_{7}Z_{8}, Z_{2}Z_{3}Z_{4}Z_{5}Z_{6}$ & $X_{1}, Z_{1}Z_{9}, Z_{2}, X_{2}X_{4}$ \\
        2 & $X_{5}X_{6}X_{7}X_{8}, X_{1}X_{2}X_{6}X_{7}X_{9}, X_{2}X_{3}X_{5}X_{6}, Y_{1}X_{4}X_{6}X_{7}Y_{9}$, $X_{1}Z_{3}X_{4}Z_{5}X_{6}Y_{7}Y_{9}, X_{1}Z_{3}X_{4}Y_{6}X_{7}Z_{8}Y_{9}, X_{1}Y_{2}Z_{4}Z_{5}X_{6}X_{7}Z_{8}Y_{9}$ & $X_{1}X_{2}Y_{4}X_{6}X_{7}, X_{1}Z_{4}X_{6}X_{7}$ \\
        3 & $X_{5}X_{6}X_{7}X_{8}, X_{1}Z_{4}X_{6}X_{7}, X_{2}X_{3}X_{5}X_{6}, Y_{1}X_{2}Y_{4}X_{6}X_{7}Z_{9}$, $Z_{3}Y_{4}Z_{5}Z_{7}Y_{9}, Z_{3}Y_{4}Z_{6}Z_{8}Y_{9}, Y_{2}Z_{5}Z_{8}Y_{9}$ & $X_{1}Y_{9}, X_{1}X_{2}X_{6}X_{7}X_{9}$ \\
        4 & $Z_{3}Y_{4}Z_{5}Z_{7}Y_{9}, Y_{1}Y_{2}Z_{3}Z_{5}Z_{6}, Z_{5}Z_{6}Z_{7}Z_{8}, X_{3}X_{6}X_{8}Z_{9}$, $X_{2}X_{5}X_{8}Z_{9}, X_{2}X_{6}X_{7}Z_{9}, Z_{1}X_{2}X_{4}Z_{9}$ & $Y_{1}Y_{4}, X_{4}Z_{9}$ \\
        5 & $Z_{3}Z_{4}Z_{5}Z_{7}X_{9}, X_{1}Z_{2}Z_{3}Z_{5}Z_{6}, Z_{5}Z_{6}Z_{7}Z_{8}, X_{2}X_{5}X_{8}Z_{9}$, $X_{2}X_{6}X_{7}Z_{9}, X_{3}X_{4}X_{6}X_{8}, Z_{1}X_{2}$ & $Z_{3}Z_{4}, X_{3}X_{6}X_{8}Z_{9}$ \\
        6 & $X_{1}Z_{2}Z_{4}Z_{5}Z_{6}, Z_{3}Z_{4}Z_{5}Z_{6}Z_{7}Z_{8}, Z_{5}Z_{7}X_{9}, X_{2}X_{5}X_{8}Z_{9}$, $X_{2}X_{6}X_{7}Z_{9}, X_{3}X_{4}X_{6}X_{8}, Z_{1}X_{2}$ & $Z_{3}Z_{4}, X_{4}X_{6}$ \\
        7 & $X_{1}Z_{2}Z_{3}Z_{5}Z_{8}X_{9}, X_{1}Z_{2}Z_{4}Z_{6}Z_{7}X_{9}, Z_{1}X_{5}X_{8}Z_{9}, Z_{1}X_{2}$, $X_{5}X_{6}X_{7}X_{8}, X_{3}X_{8}, X_{4}X_{6}$ & $Z_{5}Z_{7}X_{9}, X_{5}X_{8}$ \\
        8 & $X_{1}Z_{2}Z_{3}Z_{5}Z_{8}X_{9}, X_{1}Z_{2}Z_{4}Z_{6}Z_{7}X_{9}, Z_{1}X_{4}X_{5}X_{6}X_{8}Z_{9}$, $Z_{1}X_{2}X_{3}X_{5}X_{6}X_{7}, X_{6}X_{7}, X_{3}X_{8}$ & $Z_{2}Z_{4}, X_{4}X_{6}, Z_{5}X_{9}, X_{5}X_{8}$ \\
	\end{tabularx}
	\caption{Sequence of [[9,1,3]] subsystem codes that constitute a ``path'' from the 7-qubit color code to the 9-qubit rotated surface code. Each stabilizer code along the ``path'' is equivalent to a code in CSS form. \label{tab:path}}
\end{table}

\section{Discussion}

We have presented and enumerated all binary stabilizer codes up to local Clifford and permutation equivalence and have given some analysis of the codes contained within the database. We encourage the reader to explore these tables more deeply. Future work will explore improving naive aspects of the search approaches as well as finding a significantly faster method to determine if two codes are equivalent. 

\section{Acknowledgements}

We acknowledge helpful discussions with Theodore Yoder, Lev Bishop, Sergey Bravyi, Ken Brown, Markus Grassl, John Smolin, Andrew Nemec and Victor Albert. 

\bibliographystyle{unsrt}
\bibliography{paper}

\appendix

\include{details}
\include{enumtables}
\include{fulltables}

\end{document}

%% file: details.tex
\section{Code Set Equivalence}\label{ap:code-set-equivalence}

\begin{definition}
Let $G_1$ and $G_1$ be two subgroups of $\PauliN$. If $CS(G_1) = CS(G_2)$ then we say that $G_1$ and $G_2$ generate the same code sets or that $G_1$ is equivalent to $G_2$ in terms of code sets, $G_1\sim_{cs} G_2$.
\end{definition}

\noindent Subgroups that define the same code sets are easily characterized as shown in the following lemma:

\begin{lemma}\label{csiicen}
Let $G_1$ and $G_1$ be two subgroups of $\PauliN$, then $G_1\sim_{cs} G_2$ if, and only if, $\langle iI, Z(G_1)\rangle = \langle iI, Z(G_2)\rangle$.
\end{lemma}

\begin{proof} The proof follows from the fact that if $W\in\PauliN$ is of order $2$ then the eigenvectors of $iW$ are $i$ multiples of the eigenvectors of $W$.
\end{proof}

\begin{corollary} \label{cor:geqiig}
Let $G\leq\PauliN$ then $CS(G)=CS(\langle iI, G\rangle)$.
\end{corollary}

It will be useful to work only with gauge groups $G$ have a generating sets $\langle iI, g_1,\dots,g_k\rangle$ where the $g_l$ are order $2$ Pauli operators. We make this assumption without loss of generality from above. We will also define $S=S_G$ such that
\[
Z(G) = \langle iI, S_G\rangle = \langle iI, S\rangle 
\]
\noindent where $-I\not\in S$ (so $S$ is generated by operators of order $2$). In this situation the above character definition simplifies to 
\[
CS(G) = \{\mathcal{H}_\chi : \chi\in\hat{S}\}
\]
\noindent if $S=\langle s_1,\dots,s_k\rangle$ is a minimal generating set then $|\hat{S}|=|S|=|CS(G)|=2^k$ and so $\textrm{dim}_\bbC(\mathcal{H}_\chi)=2^{n-k}$ as expected.

\begin{definition}
Let $G$ be a gauge group of $\PauliN$ and  $M=\langle iI, G\rangle$ with $Z(M)=\langle iI, S\rangle$ with $-I\not\in S$. Let $\chi_0$ be the principal character of $\widehat{S}$. The the principal code of $G$ is defined as
\[
C(G) = \mathcal{H}_{\chi_0} = \{\ket{\psi}\in\mathcal{H}: P\ket{\psi}=\ket{\psi}\, \forall P\in S\}
\]
\end{definition}

\begin{lemma}\label{csconj}
Let $W\in\PauliN$ and $G\leq\PauliN$ then $G\sim_{cs} WGW^{-1}$.
\end{lemma}
\begin{proof}
By Corollary~\ref{cor:geqiig} we may assume without loss of generality that $iI\in G$ and we define $S$ such that $Z(G)=\langle iI,S\rangle$ with $-I\not\in S$. For each $g\in G$ there exists a $\epsilon\in\{-1,+1\}$ such that $WgW^{-1}=\epsilon g$ hence $\langle iI, WgW^{-1}\rangle=\langle iI,g\rangle$. Thus
$Z(WGW^{-1})=WZ(G)W^{-1}=Z(G)$ and so $CS(WGW^{-1}) = CS(Z(WGW^{-1}))=CS(Z(G))=CS(G)$.
\end{proof}

\begin{lemma}\label{deltaP}
Let $S\leq \PauliN$ be an abelian subgroup with $-I\not\in S$ and $S=\langle g_1,g_2,\dots,g_k\rangle$ a minimal generating set. Let $\delta=(\delta_1,\delta_2,\dots,\delta_k)\in\bbZ_2^k$ then there exists a $P\in\PauliN$ such that $Pg_vP^{-1}=(-1)^{\delta_v}g_v$ for all $v$.
\end{lemma}
\begin{proof}
For each $g_v$ let $w_v\in\Pauli_n$ be chosen such that $w_vg_vw^{-1}_v = (-1)^{\delta_v}g_v$ and $w_vg_jw^{-1}_v = g_j$ for all $j\not=v$. Such $w_v$ exist for $\delta_v = 1$ by Proposition~10.4 of~\cite{nielsen00} otherwise $w_v$ can be set to $I$. Setting $P=w_1w_2\dots{w_k}$ proves the lemma.
\end{proof}

\begin{lemma}
Let $S_1,S_2\leq\PauliN$ be two stabilizer subgroups. Then $S_1\sim_{cs} S_2$ if, and only if, there exists a $W\in\PauliN$ such that $S_1=WS_2W^{-1}$.
\end{lemma}\label{appendix:lemma:LPeqSE}
\begin{proof}
Suppose that $S_1=WS_2W^{-1}$ for some $W\in\PauliN$. Then by Lemma~\ref{csconj} $CS(S_1)=CS(WS_2W^{-1})=CS(S_2)$ and so $S_1\sim_{cs}S_2$. Now suppose that
$S_1\sim_{cs}S_2$ and that $S_1=\langle g_1,g_2,\dots,g_k\rangle$ and $S_2=\langle h_1,h_2,\dots,h_t\rangle$ are minimal generating sets with each of the $g_v$ and $h_v$ being of order two. By Lemma~\ref{csiicen} we must have $k=t$ and that there exists some invertible matrix  $\lambda$ over $\bbZ_2$ and a  $\delta=(\delta_1,\delta_2,\dots,\delta_k)\in\bbZ_2^k$ such that 
\begin{align*}
    g_1 &= (-1)^{\delta_1}h_1^{\lambda_{11}}h_2^{\lambda_{12}}\dots h_k^{\lambda_{1k}} \\
    g_2 &= (-1)^{\delta_2}h_1^{\lambda_{21}}h_2^{\lambda_{22}}\dots h_k^{\lambda_{2k}} \\
    &\vdots \\
    g_k &= (-1)^{\delta_k}h_1^{\lambda_{k1}}h_2^{\lambda_{k2}}\dots h_k^{\lambda_{kk}} \\
\end{align*}
\noindent Let $w=(w_1,w_2,\dots, w_k)^T = \lambda \delta^T$ and let $P\in\PauliN$ such that  $Ph_vP^{-1}=(-1)^{w_v}h_v$
for all $v$. Such a $P$ exists by Lemma~\ref{deltaP}. Then 
$
Ph_1^{\lambda_{v1}}h_2^{\lambda_{v2}}\dots h_k^{\lambda_{vk}}P^{-1} = (-1)^{\delta_v}h_1^{\lambda_{v1}}h_2^{\lambda_{v2}}\dots h_k^{\lambda_{vk}}.
$ Thus
\[
S_1=\langle g_1,g_2,\dots,g_k\rangle = P\langle h_1,h_2,\dots,h_t\rangle P^{-1} = PS_2P^{-1}.
\]
\end{proof}

\section{Permutation Equivalence}

Code set equivalence is not the only natural condition for two codes to be equivalent. It is also natural to define two codes to be equivalent if one code can be obtained from another by simply permuting the qubit labels.

\begin{definition}
Let $g\in\PauliN$ and suppose that $g=P_1^{\delta_1}P_2^{\delta_2}\cdots{}P_n^{\delta_n}$ where $\delta_j\in\bbZ_2$, $P\in\{I,X,Z,Y\}$ and
\[
P_j = I^{\otimes (j-1)}\otimes P \otimes I^{\otimes (n-j)}
\]
\noindent is the operator $P$ acting on qubit $j$. Suppose that $\pi\in S_n$ then the action of $\pi$ on $g$ is defined by
\[
g^\pi = P_{\pi(1)}^{\delta_1}P_{\pi(2)}^{\delta_2}\cdots{}P_{\pi(n)}^{\delta_n}.
\]
\end{definition}

\begin{lemma}
Let $G\leq \PauliN$ and $\pi\in S_n$. Then $Z(G^\pi)=Z(G)^\pi$ and $CS(G^\pi)=CS(H)^\pi$.
\end{lemma}
\begin{proof}
This follows almost by definition.
\end{proof}

\begin{definition}
Let $G$ and $H$ be two gauge groups of $\PauliN$. The groups $G$ and $H$ are said to be permutation equivalent if there exists a $\pi\in S_n$ such that $G=H^\pi$ where $H^\pi=\langle h^\pi: h\in H\rangle$. In this case we write $G\sim_\pi H$.
\end{definition}

\begin{lemma}
If $G,H\leq \PauliN$ and if $G\sim_\pi H$ then $G\sim_{cs}H^\pi$.
\end{lemma}
\begin{proof}
If $G\sim_\pi H$ then $G=H^\pi$ for some $\pi\in S_n$. Hence $CS(G)=CS(H^\pi)$ and so $G\sim_{cs} H^\pi$.
\end{proof}

\section{LC$\Pi$ Equivalent Codes}
This appendix provides the details on defining LC$\Pi$ equivalence over symplectic matrices. Since $\Cl{1}/\langle\zeta_8I,\Pauli_1\rangle\simeq  \sym{3}$ we can define the group action of $\sym{3}^n$ on gauge groups of $\PauliN$ in the following way. Let $W\in \sym{3}^n$ and let $C_W$ be the corresponding operator in $(\Cl{1}/\langle\zeta_8I,\Pauli_1\rangle)^{\otimes n}$ under the isomorphism $\Cl{1}/\langle\zeta_8I,\Pauli_1\rangle\simeq  \sym{3}$. Now define the action of $W$ on a gauge group via the conjugation action of $W_C$. The combined group action 
\[
(W,\pi)\calL = W\calL^\pi
\]
\noindent then turns the set $\sym{3}^n\times \sym{n}$  into the group $\sym{3}^n \rtimes_{\psi} \sym{n} = \sym{3}\wr\sym{n}$.

\begin{definition}[LCP Equivalence for Gauge Groups]
Let $\calG$ and $\calL$ be two gauge groups of $\PauliN$. The groups $\mathcal{G}$ and $\mathcal{L}$ are said to be locally permutation Clifford (LCP) equivalent, written as $\calG\sim \calL$ if, and only if,
\[
\calG = (W, \pi) \calL
\]
\noindent for some $(W, \pi)\in \sym{3}\wr \sym{n}$
\end{definition}

The above equivalence definition must be written in terms of symplectic matrices. For this we need three group actions, melded into one, that represent the action of matrices representing a local Clifford conjugation, the action of matrices representing a change of basis (since using a symplectic matrix representation requires fixing a basis) and the action of matrices representing a permutation of qubits.

Change of basis can be achieved by left multiplication of a symplectic matrix by an matrix from $\textrm{GL}_v(\bbF_2)$. That is the group $\textrm{GL}_v(\bbF_2)$ acts on $\bbF_2^{v\times 2n, v}$ via multiplication on the left. For the permutation of qubits let $\pi\in\sym{n}$ and set
\[
[\pi] = \textrm{diag}([\pi]_c,[\pi]_c) \in \textrm{GL}_{2n}(\bbF_2)
\]
\noindent where $[\pi]_c\in\textrm{GL}_{n}(\bbF_2)$ is the matrix representation of $\pi$ such that multiplication by $[\pi]_c$ on the right of a matrix $M$ will permute the columns of $M$ according to the permutation $\pi$. We now define the action of $\sym{n}$ on $\bbF_2^{v\times 2n, v}$ by
\[
\pi M := M [\pi^{-1}]
\]
\noindent for $M\in\bbF_2^{k\times 2n, v}$ and $\pi\in\sym{n}$.

As before we reduce the complexity and size of the local Clifford group $\Cln{1}{n}$ via the reduction 
\[
\Cln{1}{n} \rightarrow (\Cl{1}/\langle\zeta_8I,\mathcal{P}_1\rangle)^{\otimes n} \simeq \sym{3}^n
\]
\noindent as coefficients of the Paulis are not important in this situation. Since $\Cl{1}/\langle\zeta_8I,\mathcal{P}_1\rangle$ acts as $\textrm{Sym}(X,Z,Y)$ we need only construct matrices in $\textrm{GL}_{2n}(\bbF_2)$ that perform the same action on the symplectic matrix representation of the Pauli subgroup. This can be done via the isomorphism:
\begin{align*}
    (X) &\xrightarrow{\tau} \begin{bmatrix} 0 & 1 \\ 1 & 0 \end{bmatrix} \\
    (XZY) &\xrightarrow{\tau} \begin{bmatrix} 0 & 1 \\ 1 & 1 \end{bmatrix}
\end{align*}
\noindent which works for a single qubit. This idea can easily be extended to $n$ qubits. Let
$M=\begin{bmatrix} a & b \\ c & d \end{bmatrix}$ be a matrix in the above matrix group representation for $\textrm{Sym}(X,Z,Y)$.  Then for some integer $j$ the matrix

\[
\begin{tikzpicture}[
style1/.style={
  matrix of math nodes,
  every node/.append style={text width=#1,align=center,minimum height=#1},
  nodes in empty cells,
  left delimiter=[,
  right delimiter=],
  }
]
\matrix[style1=0.45cm] (1mat)
{
  & & & & & & & & & \\
  & & & & & & & & & \\
  & & & & & & & & & \\
  & & & & & & & & & \\
  & & & & & & & & & \\
  & & & & & & & & & \\
  & & & & & & & & & \\
  & & & & & & & & & \\
  & & & & & & & & & \\
  & & & & & & & & & \\
  & & & & & & & & & \\
  & & & & & & & & & \\
  & & & & & & & & & \\
};

\draw[dashed]
  (1mat-1-3.north) -- (1mat-13-3.south);
\draw[dashed]
  (1mat-1-8.north) -- (1mat-13-8.south);

\draw[dashed]
  (1mat-3-1.west) -- (1mat-3-10.east);
\draw[dashed]
  (1mat-10-1.west) -- (1mat-10-10.east);

\draw[dashed]
  (1mat-3-4.center) -- (1mat-10-4.center);
\draw[dashed]
  (1mat-3-7.center) -- (1mat-10-7.center);

\draw[dashed]
  (1mat-4-3.south) -- (1mat-4-8.south);
\draw[dashed]
  (1mat-9-3.north) -- (1mat-9-8.north);

\node[font=\Large] 
  at (1mat-2-2.north west) {$I_{i-1}$};
  
\node[font=\Large] 
  at (1mat-4-4.north west) {$a$};

\node[font=\Large] 
  at (1mat-4-7.north east) {$b$};  
  
\node[font=\Large] 
  at (1mat-9-3.east) {$c$};
  
\node[font=\Large] 
  at (1mat-9-7.east) {$d$};
  
\node[font=\Large] 
  at (1mat-6-6.west) {$I_{n-1}$};
  
\node[font=\Large] 
  at (1mat-12-10.west) {$I_{n-i}$};
\end{tikzpicture}
\]

\noindent represents the matrix in $\textrm{GL}_{2n}(\bbF_2)$ that applies $M$ to qubit $j$. Thus, the action of any Clifford operator in $\Cl{1}/\langle\zeta_8I,\mathcal{P}_1\rangle$, via conjugation, can be represented by a matrix in  $\textrm{GL}_{2n}(\bbF_2)$ via a product of the matrices defined above. This action is multiplication on the right.


These three actions can be put together to form a single group. Let $G=(G_1,G_2,\dots,G_n)\in\mathcal{S}_3^n$, $M\in\textrm{GL}_{v}(\bbF_2)$. We define the group action of $\textrm{GL}_{v}(\bbF_2)\times \sym{3}^n$ on $\bbF_2^{v\times 2n, v}$ as 
\[
(m,G)\Gamma = M\Gamma[G]
\]
\noindent where $\Gamma\in\bbF_2^{v\times 2n, v}$ where the group law on $\textrm{GL}_{v}(\bbF_2)\times \sym{3}^n$ is coordinate-wise multiplication. Let $\pi\in \sym{n}$ and define multiplication on the set $\textrm{GL}_{v}(\bbF_2)\times \sym{3}^n \times \sym{n}$ as
\[
((M_2,H);\pi_2)\circ((M_2,G);\pi_1) := ((M_1M_2, H\psi_{\pi_2}(G));\pi_2\pi_1)
\]
\noindent where  $\psi:\sym{n}\rightarrow \textrm{Aut}(\textrm{GL}_k(\bbF_2)\times \sym{3}^{n})$ is the homomorphism that acts as the identity on $\textrm{GL}_k(\bbF_2)$ and that is defined by, abusing notation slightly, 
\[
\psi(\pi)(M,G) = \psi_{\pi}(M, G) = (M, \psi_\pi(G))
\]
\noindent where 
\[
\psi_{\pi}(G_1,G_2,\dots,G_n) = (G_{\pi^{-1}(1)},\dots,G_{\pi^{-1}(n)}).
\]
\noindent The set $\sym{n} \times \sym{3}^n \times \textrm{GL}_k(\bbF_2)$  together with the multiplication $\circ$ define the group
$(\textrm{GL}_k(\bbF_2) \times \sym{3}^n) \rtimes_\psi \sym{n}$.  We then define its action on $\bbF_2^{v\times 2n,v}$ as
\[
((M,G);\pi) \Gamma := (M, G) (\pi\Gamma) = M \Gamma [\pi^{-1}] [G]
\]
\noindent where $[\pi^{-1}]$ and $[G]$ are the matrix representations of $\pi^{-1}$ and $G$ in $\textrm{GL}_{2n}(\bbF_2)$ as defined above.

We can now give the definition of equivalence in terms of symplectic matrices.

\begin{definition}[LC$\Pi$ for Symplectic Matrices]
Let $\Gamma$ and $\Omega$ be matrices of $\bbF_2^{k\times 2n, v}$. Then $\Gamma$ and $\Omega$ are said to be locally permutation Clifford equivalent, written as $\Gamma\sim\Omega$ if, and only if, 
\[
\Gamma = ((m,G);\pi) \Omega,
\]
\noindent for some $((m,G);\pi) \in (\textrm{GL}_k(\bbF_2) \times \sym{3}^n) \rtimes_\psi \sym{n}$.
\end{definition}

\noindent These two equivalence relations are then related as $\mathcal{A},\mathcal{B}\leq\mathcal{P}_n$ then $\mathcal{A}\sim\mathcal{B}$ if, and only if, $\kappa(\mathcal{A})\sim\kappa(\mathcal{B})$ where $\kappa(\calA)$ is the symplectic matrix representation of $\calA$ relative to the standard basis for $\PauliN$.

\section{Qiskit QEC database}\label{qiskit-qec}

The full catalogue of equivalence classes for $n=1,2,...,9$ are available in Qiskit-QEC~\cite{qiskitqec} as a code library. 
In Qiskit-QEC varous code libraries are managed by a \texttt{CodeLibrarian}.  These databases are meant to be analogous
to GAP's~\cite{GAP4} small group database or to Sloan's encyclopedia of integer sequences~\cite{oeis}. The \texttt{CodeLibrarian} is started by 
issuing the python command  (from within a Python shell or Juypter Notebook):

\begin{verbatim}
import qiskit_qec.codes.codebase as cb
\end{verbatim}

The syntax for the codebases is similar for GAP's small group database. You can search for a given code or set of code
and return the codes themselves or just the information about the codes. To find a specific code you can use the \texttt{small\_code} 
method and provide it with suitable parameters. 
\begin{verbatim}
# Load the code with n=5, k=0 and index=4
code = cb.small_code(5, 0, 4)
code.generators
PauliList(['Z1Z2', 'Z0Z4', 'Y1Y2X3', 'Z0Z1Z3', 'Y0X3Y4'])
code.generators.matrix.astype(int)
array([[0, 0, 0, 0, 0, 0, 1, 1, 0, 0],
       [0, 0, 0, 0, 0, 1, 0, 0, 0, 1],
       [0, 1, 1, 1, 0, 0, 1, 1, 0, 0],
       [0, 0, 0, 0, 0, 1, 1, 0, 1, 0],
       [1, 0, 0, 1, 1, 1, 0, 0, 0, 1]])
\end{verbatim}
\noindent If you only want the information for the code then set the \texttt{info\_only} flag to True.
\begin{verbatim}
# Load the code with n=5, k=0 and index=4 but only fetch the information on it
code = cb.small_code(5, 0, 4, info_only=True)
print(code)
aut_group_generators : ['S1S2', '(1,2)', 'S0S4', '(0,1)(2,4)']
aut_group_size       : 32
code_type            : StabSubSystemCode
d                    : 2
index                : 4
is_css               : 1
is_decomposable      : 0
is_degenerate        : 0
is_gf4linear         : 0
is_subsystem         : 1
isotropic_generators : ['Z0Z4', 'Z1Z2', 'Z1Z3Z4', 'X1X2X3', 'X0X3X4']
k                    : 0
logical_ops          : []
n                    : 5
uuid                 : ee84af51-7ef1-4540-b635-a6377aa49596
weight_enumerator    : [1, 0, 2, 8, 13, 8]
\end{verbatim}

\noindent There is also a method to find sets of codes using the available fields for a database. It uses the \texttt{all\_small\_codes} method as follows:
\begin{verbatim}
# Load all small codes information with n=9, k=3, d=3 that are CSS, are not decomposable
# and then print only the codes with automorphism groups sizes at least 300.
codes = cb.all_small_codes(9, 3, d=3, info_only=True, is_decomposable=False, is_css=False)
for code in codes:
    if code['aut_group_size'] > 300:
        print(
            f" [[{code['n'],code['k'],code['d']}]]:{code['index']} \
                aut_group_size={code['aut_group_size']}"
        )
 [[(9, 3, 3)]]:170234 aut_group_size=324
 [[(9, 3, 3)]]:170235 aut_group_size=1296
\end{verbatim}
\noindent Codes classes are stored in groups associated to a given $n$ and $k$. If you want to find codes over different $n$ and $k$ simply load those separately. For example, if you want to find the number of codes for each $n$ in the database:

\begin{verbatim}
print(f"Total number of codes for each n:")
summary = 0
for n in range(1,10):
    total = 0
    num = 0
    for k in range(n+1):
        codes = cb.all_small_codes(n, k, info_only=True, list_only=True)
        num = len(codes)
        total += num
    print(f" n = {n} : {total}")
    summary += total
print(f"The database contains {summary} number of code equivalence classes")

Total number of codes for each n:
 n = 1 : 2
 n = 2 : 5
 n = 3 : 12
 n = 4 : 35
 n = 5 : 112
 n = 6 : 474
 n = 7 : 2757
 n = 8 : 28642
 n = 9 : 721967
The database contains 754006 number of code equivalence classes.
\end{verbatim}

%% file: enumtables.tex
\section{Tables}

This section contains detailed tables from the database for small values of $n$ that not included in the main section.

\begin{table}[!htbp]
    \centering
    \begin{tabular}{c|cccccccccc}
        $n\backslash k$ & 0 & 1 & 2 & 3 & 4 & 5 & 6 & 7 & 8 & 9 \\ 
        \midrule
        1 & 0 & 0 & 0 & 0 & 0 & 0 & 0 & 0 & 0 & 0 \\ 
        2 & 1 & 1 & 1 & 0 & 0 & 0 & 0 & 0 & 0 & 0 \\ 
        3 & 2 & 3 & 2 & 1 & 0 & 0 & 0 & 0 & 0 & 0 \\ 
        4 & 4 & 7 & 7 & 3 & 1 & 0 & 0 & 0 & 0 & 0 \\ 
        5 & 7 & 19 & 22 & 13 & 4 & 1 & 0 & 0 & 0 & 0 \\ 
        6 & 15 & 52 & 78 & 56 & 22 & 5 & 1 & 0 & 0 & 0 \\ 
        7 & 33 & 164 & 321 & 296 & 136 & 35 & 6 & 1 & 0 & 0 \\ 
        8 & 81 & 604 & 1682 & 2005 & 1129 & 308 & 52 & 7 & 1 & 0 \\ 
        9 & 235 & 2985 & 13444 & 22273 & 15623 & 4593 & 676 & 75 & 8 & 1 \\ 
    \end{tabular}
    \caption{Number of equivalence classes of each $[[n,k]]$ that are decomposable.}\label{table:totalcountfinedecom}
\end{table}

\begin{table}[!htbp]
\centering
\begin{tabular}{c|llllllllll}
\toprule
$n/k$ & 0 & 1 & 2 & 3 & 4 & 5 & 6 & 7 & 8 & 9 \\ 
\midrule
 1 &  1 & 1 & - & - & - & - & - & - & - & - \\ 
 2 &  2 & 1 & 1 & - & - & - & - & - & - & - \\ 
 3 &  2 & 1 & 1 & 1 & - & - & - & - & - & - \\ 
 4 &  2 & 2 & 2 & 1 & 1 & - & - & - & - & - \\ 
 5 &  3 & 3 & 2 & 1 & 1 & 1 & - & - & - & - \\ 
 6 &  4 & 3 & 2 & 2 & 2 & 1 & 1 & - & - & - \\ 
 7 &  3 & 3 & 2 & 2 & 2 & 1 & 1 & 1 & - & - \\ 
 8 &  4 & 3 & 3 & 3 & 2 & 2 & 2 & 1 & 1 & - \\ 
 9 &  4 & 3 & 3 & 3 & 2 & 2 & 2 & 1 & 1 & 1 \\ 
\end{tabular}
\caption{Maximum distance for all $[[n,k,d]]$ codes.}
\end{table}

\begin{table}[!htbp]
\centering
\begin{tabular}{c|llllllllll}
\toprule
$n/k$ & 0 & 1 & 2 & 3 & 4 & 5 & 6 & 7 & 8 & 9 \\ 
\midrule
 1 &  1 & 1 & - & - & - & - & - & - & - & - \\ 
 2 &  2 & 1 & - & - & - & - & - & - & - & - \\ 
 3 &  2 & 1 & 1 & - & - & - & - & - & - & - \\ 
 4 &  2 & 2 & 2 & 1 & - & - & - & - & - & - \\ 
 5 &  3 & 3 & 2 & 1 & 1 & - & - & - & - & - \\ 
 6 &  4 & 3 & 2 & 2 & 2 & 1 & - & - & - & - \\ 
 7 &  3 & 3 & 2 & 2 & 2 & 1 & 1 & - & - & - \\ 
 8 &  4 & 3 & 3 & 3 & 2 & 2 & 2 & 1 & - & - \\ 
 9 &  4 & 3 & 3 & 3 & 2 & 2 & 2 & 1 & 1 & - \\ 
\end{tabular}
\caption{Maximum distance for all indecomposable $[[n,k,d]]$ codes.}
\end{table}

\begin{table}[!htbp]
\centering
\begin{tabular}{c|llllllllll}
\toprule
$n/k$ & 0 & 1 & 2 & 3 & 4 & 5 & 6 & 7 & 8 & 9 \\ 
\midrule
 1 &  1 & 1 & - & - & - & - & - & - & - & - \\ 
 2 &  2 & 1 & 1 & - & - & - & - & - & - & - \\ 
 3 &  2 & 1 & 1 & 1 & - & - & - & - & - & - \\ 
 4 &  2 & 2 & 2 & 1 & 1 & - & - & - & - & - \\ 
 5 &  2 & 2 & 2 & 1 & 1 & 1 & - & - & - & - \\ 
 6 &  3 & 2 & 2 & 2 & 2 & 1 & 1 & - & - & - \\ 
 7 &  3 & 3 & 2 & 2 & 2 & 1 & 1 & 1 & - & - \\ 
 8 &  4 & 3 & 2 & 2 & 2 & 2 & 2 & 1 & 1 & - \\ 
 9 &  3 & 3 & 2 & 2 & 2 & 2 & 2 & 1 & 1 & 1 \\ 
\end{tabular}
\caption{Maximum distance for all CSS $[[n,k,d]]$ codes.}
\end{table}

\begin{table}[!htbp]
\centering
\begin{tabular}{c|llllllllll}
\toprule
$n/k$ & 0 & 1 & 2 & 3 & 4 & 5 & 6 & 7 & 8 & 9 \\ 
\midrule
 1 &  1 & 1 & - & - & - & - & - & - & - & - \\ 
 2 &  2 & 1 & - & - & - & - & - & - & - & - \\ 
 3 &  2 & 1 & 1 & - & - & - & - & - & - & - \\ 
 4 &  2 & 2 & 2 & 1 & - & - & - & - & - & - \\ 
 5 &  2 & 2 & 2 & 1 & 1 & - & - & - & - & - \\ 
 6 &  3 & 2 & 2 & 2 & 2 & 1 & - & - & - & - \\ 
 7 &  3 & 3 & 2 & 2 & 2 & 1 & 1 & - & - & - \\ 
 8 &  4 & 3 & 2 & 2 & 2 & 2 & 2 & 1 & - & - \\ 
 9 &  3 & 3 & 2 & 2 & 2 & 2 & 2 & 1 & 1 & - \\ 
\end{tabular}
\caption{Maximum distance for all indecomposable CSS $[[n,k,d]]$ codes.}
\end{table}

\begin{table}[!htbp]
\centering
\begin{tabular}{c|llllllllll}
\toprule
$n/k$ & 0 & 1 & 2 & 3 & 4 & 5 & 6 & 7 & 8 & 9 \\ 
\midrule
 1 &  - & - & - & - & - & - & - & - & - & - \\ 
 2 &  2 & - & - & - & - & - & - & - & - & - \\ 
 3 &  - & 1 & - & - & - & - & - & - & - & - \\ 
 4 &  2 & - & 2 & - & - & - & - & - & - & - \\ 
 5 &  - & 3 & - & 1 & - & - & - & - & - & - \\ 
 6 &  4 & - & 2 & - & 2 & - & - & - & - & - \\ 
 7 &  - & 3 & - & 2 & - & 1 & - & - & - & - \\ 
 8 &  4 & - & 3 & - & 2 & - & 2 & - & - & - \\ 
 9 &  - & 3 & - & 3 & - & 2 & - & 1 & - & - \\ 
\end{tabular}
\caption{Maximum distance for all GF(4)-linear $[[n,k,d]]$ codes.}
\end{table}

\begin{table}[!htbp]
\centering
\begin{tabular}{c|llllllllll}
\toprule
$n/k$ & 0 & 1 & 2 & 3 & 4 & 5 & 6 & 7 & 8 & 9 \\ 
\midrule
 1 &  - & - & - & - & - & - & - & - & - & - \\ 
 2 &  2 & - & - & - & - & - & - & - & - & - \\ 
 3 &  - & - & - & - & - & - & - & - & - & - \\ 
 4 &  - & - & 2 & - & - & - & - & - & - & - \\ 
 5 &  - & 3 & - & - & - & - & - & - & - & - \\ 
 6 &  4 & - & 2 & - & 2 & - & - & - & - & - \\ 
 7 &  - & 3 & - & 2 & - & - & - & - & - & - \\ 
 8 &  4 & - & 3 & - & 2 & - & 2 & - & - & - \\ 
 9 &  - & 3 & - & 3 & - & 2 & - & - & - & - \\ 
\end{tabular}
\caption{Maximum distance for all indecomposable GF(4)-linear $[[n,k,d]]$ codes.}
\end{table}

\begin{landscape}
\begin{table}
\centering
\begingroup
    \footnotesize
\begin{tabular}{l|llll|lll|lll|lll|ll|ll|ll|l|l|l}
\toprule
$k$ & \multicolumn{4}{c|}{0} & \multicolumn{3}{c|}{1} & \multicolumn{3}{c|}{2} & \multicolumn{3}{c|}{3} & \multicolumn{2}{c|}{4} & \multicolumn{2}{c|}{5} & \multicolumn{2}{c|}{6} & \multicolumn{1}{c|}{7} & \multicolumn{1}{c|}{8} & \multicolumn{1}{c}{9} \\ 
$n/d$ & 1 & 2 & 3 & 4 & 1 & 2 & 3 & 1 & 2 & 3 & 1 & 2 & 3 & 1 & 2 & 1 & 2 & 1 & 2 & 1 & 1 & 1 \\ 
\midrule
1 & 1 & - & - & - & 1 & - & - & - & - & - & - & - & - & - & - & - & - & - & - & - & - & - \\ 
2 & 1 & 1 & - & - & 2 & - & - & 1 & - & - & - & - & - & - & - & - & - & - & - & - & - & - \\ 
3 & 2 & 1 & - & - & 5 & - & - & 3 & - & - & 1 & - & - & - & - & - & - & - & - & - & - & - \\ 
4 & 3 & 3 & - & - & 11 & 2 & - & 10 & 1 & - & 4 & - & - & 1 & - & - & - & - & - & - & - & - \\ 
5 & 6 & 4 & 1 & - & 29 & 6 & 1 & 37 & 3 & - & 19 & - & - & 5 & - & 1 & - & - & - & - & - & - \\ 
6 & 11 & 13 & 1 & 1 & 78 & 35 & 2 & 156 & 29 & - & 104 & 5 & - & 31 & 1 & 6 & - & 1 & - & - & - & - \\ 
7 & 26 & 29 & 4 & - & 260 & 169 & 19 & 834 & 241 & - & 785 & 67 & - & 260 & 7 & 48 & - & 7 & - & 1 & - & - \\ 
8 & 59 & 107 & 11 & 5 & 1023 & 1170 & 178 & 6266 & 3724 & 20 & 9304 & 2117 & 1 & 3699 & 264 & 603 & 11 & 70 & 1 & 8 & 1 & - \\ 
9 & 182 & 416 & 69 & 8 & 5777 & 10742 & 3609 & 78567 & 98027 & 4445 & 222749 & 130598 & 222 & 122541 & 24117 & 17677 & 768 & 1331 & 13 & 99 & 9 & 1 \\ 
\bottomrule
\end{tabular}
\caption{Number of equivalent $[[n,k,d]]$ codes.}
\endgroup
\end{table}

\begin{table}
\centering
\begingroup
    \footnotesize
\begin{tabular}{l|llll|lll|lll|lll|ll|ll|ll|l|l|l}
\toprule
$k$ & \multicolumn{4}{c|}{0} & \multicolumn{3}{c|}{1} & \multicolumn{3}{c|}{2} & \multicolumn{3}{c|}{3} & \multicolumn{2}{c|}{4} & \multicolumn{2}{c|}{5} & \multicolumn{2}{c|}{6} & \multicolumn{1}{c|}{7} & \multicolumn{1}{c|}{8} & \multicolumn{1}{c}{9} \\ 
$n/d$ & 1 & 2 & 3 & 4 & 1 & 2 & 3 & 1 & 2 & 3 & 1 & 2 & 3 & 1 & 2 & 1 & 2 & 1 & 2 & 1 & 1 & 1 \\ 
\midrule
1 & 1 & - & - & - & 1 & - & - & - & - & - & - & - & - & - & - & - & - & - & - & - & - & - \\ 
2 & - & 1 & - & - & 1 & - & - & - & - & - & - & - & - & - & - & - & - & - & - & - & - & - \\ 
3 & - & 1 & - & - & 2 & - & - & 1 & - & - & - & - & - & - & - & - & - & - & - & - & - & - \\ 
4 & - & 2 & - & - & 4 & 2 & - & 3 & 1 & - & 1 & - & - & - & - & - & - & - & - & - & - & - \\ 
5 & - & 3 & 1 & - & 12 & 4 & 1 & 16 & 2 & - & 6 & - & - & 1 & - & - & - & - & - & - & - & - \\ 
6 & - & 9 & 1 & 1 & 35 & 27 & 1 & 82 & 25 & - & 48 & 5 & - & 9 & 1 & 1 & - & - & - & - & - & - \\ 
7 & - & 22 & 4 & - & 140 & 128 & 16 & 545 & 209 & - & 494 & 62 & - & 125 & 6 & 13 & - & 1 & - & - & - & - \\ 
8 & - & 85 & 11 & 5 & 646 & 964 & 157 & 4858 & 3450 & 20 & 7373 & 2043 & 1 & 2579 & 255 & 295 & 11 & 18 & 1 & 1 & - & - \\ 
9 & - & 363 & 69 & 8 & 4337 & 9395 & 3411 & 69122 & 94048 & 4425 & 202670 & 128405 & 221 & 107191 & 23844 & 13095 & 757 & 656 & 12 & 24 & 1 & - \\ 
\bottomrule
\end{tabular}
\caption{Number of equivalent indecomposable $[[n,k,d]]$ codes.}
\endgroup
\end{table}
\end{landscape}

\begin{landscape}

\begin{table}
\centering
\begingroup
    \footnotesize
\begin{tabular}{l|llll|lll|ll|ll|ll|ll|ll|l|l|l}
\toprule
$k$ & \multicolumn{4}{c|}{0} & \multicolumn{3}{c|}{1} & \multicolumn{2}{c|}{2} & \multicolumn{2}{c|}{3} & \multicolumn{2}{c|}{4} & \multicolumn{2}{c|}{5} & \multicolumn{2}{c|}{6} & \multicolumn{1}{c|}{7} & \multicolumn{1}{c|}{8} & \multicolumn{1}{c}{9} \\ 
$n/d$ & 1 & 2 & 3 & 4 & 1 & 2 & 3 & 1 & 2 & 1 & 2 & 1 & 2 & 1 & 2 & 1 & 2 & 1 & 1 & 1 \\ 
\midrule
1 & 1 & - & - & - & 1 & - & - & - & - & - & - & - & - & - & - & - & - & - & - & - \\ 
2 & 1 & 1 & - & - & 2 & - & - & 1 & - & - & - & - & - & - & - & - & - & - & - & - \\ 
3 & 2 & 1 & - & - & 5 & - & - & 3 & - & 1 & - & - & - & - & - & - & - & - & - & - \\ 
4 & 3 & 3 & - & - & 11 & 1 & - & 10 & 1 & 4 & - & 1 & - & - & - & - & - & - & - & - \\ 
5 & 6 & 4 & - & - & 27 & 4 & - & 32 & 2 & 18 & - & 5 & - & 1 & - & - & - & - & - & - \\ 
6 & 10 & 11 & 1 & - & 65 & 17 & - & 114 & 13 & 79 & 2 & 29 & 1 & 6 & - & 1 & - & - & - & - \\ 
7 & 22 & 20 & 1 & - & 175 & 62 & 1 & 417 & 52 & 392 & 15 & 168 & 3 & 43 & - & 7 & - & 1 & - & - \\ 
8 & 43 & 58 & 2 & 1 & 492 & 248 & 2 & 1691 & 311 & 2082 & 132 & 1153 & 36 & 326 & 3 & 61 & 1 & 8 & 1 & - \\ 
9 & 104 & 142 & 4 & - & 1539 & 1031 & 22 & 7494 & 1843 & 12627 & 1233 & 8886 & 345 & 3055 & 43 & 592 & 4 & 83 & 9 & 1 \\ 
\bottomrule
\end{tabular}
\caption{Number of equivalent CSS $[[n,k,d]]$ codes.}
\endgroup
\end{table}

\begin{table}
\centering
\begingroup
    \footnotesize
\begin{tabular}{l|llll|lll|ll|ll|ll|ll|ll|l|l|l}
\toprule
$k$ & \multicolumn{4}{c|}{0} & \multicolumn{3}{c|}{1} & \multicolumn{2}{c|}{2} & \multicolumn{2}{c|}{3} & \multicolumn{2}{c|}{4} & \multicolumn{2}{c|}{5} & \multicolumn{2}{c|}{6} & \multicolumn{1}{c|}{7} & \multicolumn{1}{c|}{8} & \multicolumn{1}{c}{9} \\ 
$n/d$ & 1 & 2 & 3 & 4 & 1 & 2 & 3 & 1 & 2 & 1 & 2 & 1 & 2 & 1 & 2 & 1 & 2 & 1 & 1 & 1 \\ 
\midrule
1 & 1 & - & - & - & 1 & - & - & - & - & - & - & - & - & - & - & - & - & - & - & - \\ 
2 & - & 1 & - & - & 1 & - & - & - & - & - & - & - & - & - & - & - & - & - & - & - \\ 
3 & - & 1 & - & - & 2 & - & - & 1 & - & - & - & - & - & - & - & - & - & - & - & - \\ 
4 & - & 2 & - & - & 4 & 1 & - & 3 & 1 & 1 & - & - & - & - & - & - & - & - & - & - \\ 
5 & - & 3 & - & - & 10 & 3 & - & 12 & 1 & 5 & - & 1 & - & - & - & - & - & - & - & - \\ 
6 & - & 7 & 1 & - & 25 & 12 & - & 50 & 10 & 30 & 2 & 8 & 1 & 1 & - & - & - & - & - & - \\ 
7 & - & 14 & 1 & - & 74 & 41 & 1 & 208 & 37 & 185 & 13 & 63 & 2 & 10 & - & 1 & - & - & - & - \\ 
8 & - & 40 & 2 & 1 & 229 & 168 & 1 & 953 & 244 & 1176 & 114 & 572 & 31 & 119 & 3 & 14 & 1 & 1 & - & - \\ 
9 & - & 106 & 4 & - & 796 & 717 & 19 & 4688 & 1475 & 8198 & 1082 & 5396 & 305 & 1531 & 40 & 211 & 3 & 17 & 1 & - \\ 
\bottomrule
\end{tabular}
\caption{Number of equivalent indecomposable CSS $[[n,k,d]]$ codes.}
\endgroup
\end{table}
\end{landscape}

\begin{landscape}
    \begin{table}
    \centering
    \begingroup
    \footnotesize
\begin{tabular}{l|llll|lll|ll|ll|ll|ll|ll|l|l|l}
\toprule
$k$ & \multicolumn{4}{c|}{0} & \multicolumn{3}{c|}{1} & \multicolumn{2}{c|}{2} & \multicolumn{2}{c|}{3} & \multicolumn{2}{c|}{4} & \multicolumn{2}{c|}{5} & \multicolumn{2}{c|}{6} & \multicolumn{1}{c|}{7} & \multicolumn{1}{c|}{8} & \multicolumn{1}{c}{9} \\ 
$n/d$ & 1 & 2 & 3 & 4 & 1 & 2 & 3 & 1 & 2 & 1 & 2 & 1 & 2 & 1 & 2 & 1 & 2 & 1 & 1 & 1 \\ 
\midrule
1 & - & - & - & - & - & - & - & - & - & - & - & - & - & - & - & - & - & - & - & - \\ 
2 & - & 1 & - & - & - & - & - & - & - & - & - & - & - & - & - & - & - & - & - & - \\ 
3 & - & - & - & - & 1 & - & - & - & - & - & - & - & - & - & - & - & - & - & - & - \\ 
4 & - & 1 & - & - & - & - & - & 1 & 1 & - & - & - & - & - & - & - & - & - & - & - \\ 
5 & - & - & - & - & 1 & - & 1 & - & - & 2 & - & - & - & - & - & - & - & - & - & - \\ 
6 & - & 1 & - & 1 & - & - & - & 2 & 2 & - & - & 2 & 1 & - & - & - & - & - & - & - \\ 
7 & - & - & - & - & 2 & - & 2 & - & - & 4 & 1 & - & - & 3 & - & - & - & - & - & - \\ 
8 & - & 2 & - & 1 & - & - & - & 4 & 5 & - & - & 5 & 4 & - & - & 3 & 1 & - & - & - \\ 
9 & - & - & - & - & 3 & - & 4 & - & - & 10 & 5 & - & - & 9 & 2 & - & - & 4 & - & - \\ 
\bottomrule
\end{tabular}
    \caption{Number of equivalent GF(4)-linear $[[n,k,d]]$ codes.}\label{table:GF4countfineindecom}
\endgroup
\end{table}

\begin{table}
    \centering
    \begingroup
    \footnotesize
\begin{tabular}{l|llll|lll|ll|ll|ll|ll|ll|l|l|l}
\toprule
$k$ & \multicolumn{4}{c|}{0} & \multicolumn{3}{c|}{1} & \multicolumn{2}{c|}{2} & \multicolumn{2}{c|}{3} & \multicolumn{2}{c|}{4} & \multicolumn{2}{c|}{5} & \multicolumn{2}{c|}{6} & \multicolumn{1}{c|}{7} & \multicolumn{1}{c|}{8} & \multicolumn{1}{c}{9} \\ 
$n/d$ & 1 & 2 & 3 & 4 & 1 & 2 & 3 & 1 & 2 & 1 & 2 & 1 & 2 & 1 & 2 & 1 & 2 & 1 & 1 & 1 \\ 
\midrule
1 & - & - & - & - & - & - & - & - & - & - & - & - & - & - & - & - & - & - & - & - \\ 
2 & - & 1 & - & - & - & - & - & - & - & - & - & - & - & - & - & - & - & - & - & - \\ 
3 & - & - & - & - & - & - & - & - & - & - & - & - & - & - & - & - & - & - & - & - \\ 
4 & - & - & - & - & - & - & - & - & 1 & - & - & - & - & - & - & - & - & - & - & - \\ 
5 & - & - & - & - & - & - & 1 & - & - & - & - & - & - & - & - & - & - & - & - & - \\ 
6 & - & - & - & 1 & - & - & - & - & 1 & - & - & - & 1 & - & - & - & - & - & - & - \\ 
7 & - & - & - & - & - & - & 1 & - & - & - & 1 & - & - & - & - & - & - & - & - & - \\ 
8 & - & - & - & 1 & - & - & - & - & 3 & - & - & - & 2 & - & - & - & 1 & - & - & - \\ 
9 & - & - & - & - & - & - & 2 & - & - & - & 3 & - & - & - & 2 & - & - & - & - & - \\ 
\bottomrule
\end{tabular}
    \caption{Number of equivalent indecomposable GF(4)-linear $[[n,k,d]]$ codes.}
\endgroup
\end{table}
\end{landscape}

\begin{landscape}

\begin{table}[!htbp]
\centering
\begingroup
    \footnotesize
\begin{tabular}{l|lllllllll|lllllll|llll|l}
\toprule
$d$ & \multicolumn{9}{c|}{1} & \multicolumn{7}{c|}{2} & \multicolumn{4}{c|}{3} & \multicolumn{1}{c}{4} \\ 
$n/k$ & 0 & 1 & 2 & 3 & 4 & 5 & 6 & 7 & 8 & 0 & 1 & 2 & 3 & 4 & 5 & 6 & 0 & 1 & 2 & 3 & 0 \\ 
\midrule
0 & - & - & - & - & - & - & - & - & - & - & - & - & - & - & - & - & - & - & - & - & - \\ 
1 & 1 & - & - & - & - & - & - & - & - & - & - & - & - & - & - & - & - & - & - & - & - \\ 
2 & - & 1 & - & - & - & - & - & - & - & 1 & - & - & - & - & - & - & - & - & - & - & - \\ 
3 & - & 2 & 1 & - & - & - & - & - & - & 1 & - & - & - & - & - & - & - & - & - & - & - \\ 
4 & - & 4 & 3 & 1 & - & - & - & - & - & 2 & 2 & 1 & - & - & - & - & - & - & - & - & - \\ 
5 & - & 12 & 16 & 6 & 1 & - & - & - & - & 3 & 4 & 2 & - & - & - & - & 1 & 1 & - & - & - \\ 
6 & - & 35 & 82 & 48 & 9 & 1 & - & - & - & 9 & 27 & 25 & 5 & 1 & - & - & 1 & 1 & - & - & 1 \\ 
7 & - & 140 & 545 & 494 & 125 & 13 & 1 & - & - & 22 & 128 & 209 & 62 & 6 & - & - & 4 & 16 & - & - & - \\ 
8 & - & 646 & 4858 & 7373 & 2579 & 295 & 18 & 1 & - & 85 & 964 & 3450 & 2043 & 255 & 11 & 1 & 11 & 157 & 20 & 1 & 5 \\ 
9 & - & 4337 & 69122 & 202670 & 107191 & 13095 & 656 & 24 & 1 & 363 & 9395 & 94048 & 128405 & 23844 & 757 & 12 & 69 & 3411 & 4425 & 221 & 8 \\ 
\bottomrule
\end{tabular}
\caption{Number of classes of indecomposable codes of given distances}
\endgroup
\end{table}

\begin{table}[!htbp]
\centering
\begingroup
    \footnotesize
\begin{tabular}{l|lllllllll|lllllll|ll|l}
\toprule
$d$ & \multicolumn{9}{c|}{1} & \multicolumn{7}{c|}{2} & \multicolumn{2}{c|}{3} & \multicolumn{1}{c}{4} \\ 
$n/k$ & 0 & 1 & 2 & 3 & 4 & 5 & 6 & 7 & 8 & 0 & 1 & 2 & 3 & 4 & 5 & 6 & 0 & 1 & 0 \\ 
\midrule
0 & - & - & - & - & - & - & - & - & - & - & - & - & - & - & - & - & - & - & - \\ 
1 & 1 & - & - & - & - & - & - & - & - & - & - & - & - & - & - & - & - & - & - \\ 
2 & - & 1 & - & - & - & - & - & - & - & 1 & - & - & - & - & - & - & - & - & - \\ 
3 & - & 2 & 1 & - & - & - & - & - & - & 1 & - & - & - & - & - & - & - & - & - \\ 
4 & - & 4 & 3 & 1 & - & - & - & - & - & 2 & 1 & 1 & - & - & - & - & - & - & - \\ 
5 & - & 10 & 12 & 5 & 1 & - & - & - & - & 3 & 3 & 1 & - & - & - & - & - & - & - \\ 
6 & - & 25 & 50 & 30 & 8 & 1 & - & - & - & 7 & 12 & 10 & 2 & 1 & - & - & 1 & - & - \\ 
7 & - & 74 & 208 & 185 & 63 & 10 & 1 & - & - & 14 & 41 & 37 & 13 & 2 & - & - & 1 & 1 & - \\ 
8 & - & 229 & 953 & 1176 & 572 & 119 & 14 & 1 & - & 40 & 168 & 244 & 114 & 31 & 3 & 1 & 2 & 1 & 1 \\ 
9 & - & 796 & 4688 & 8198 & 5396 & 1531 & 211 & 17 & 1 & 106 & 717 & 1475 & 1082 & 305 & 40 & 3 & 4 & 19 & - \\ 
\bottomrule
\end{tabular}
\caption{Number of classes of indecomposable CSS codes of given distances}
\endgroup
\end{table}

\begin{table}[!htbp]
\centering
\begingroup
    \footnotesize
\begin{tabular}{l|lllllll|llll|l}
\toprule
$d$ & \multicolumn{7}{c|}{2} & \multicolumn{4}{c|}{3} & \multicolumn{1}{c}{4} \\ 
$n/k$ & 0 & 1 & 2 & 3 & 4 & 5 & 6 & 0 & 1 & 2 & 3 & 0 \\ 
\midrule
0 & - & - & - & - & - & - & - & - & - & - & - & - \\ 
1 & - & - & - & - & - & - & - & - & - & - & - & - \\ 
2 & 1 & - & - & - & - & - & - & - & - & - & - & - \\ 
3 & - & - & - & - & - & - & - & - & - & - & - & - \\ 
4 & - & - & 1 & - & - & - & - & - & - & - & - & - \\ 
5 & - & - & - & - & - & - & - & - & 1 & - & - & - \\ 
6 & - & - & 1 & - & 1 & - & - & - & - & - & - & 1 \\ 
7 & - & - & - & 1 & - & - & - & - & 1 & - & - & - \\ 
8 & - & - & 3 & - & 2 & - & 1 & - & - & 1 & - & 1 \\ 
9 & - & - & - & 3 & - & 2 & - & - & 2 & - & 1 & - \\ 
\bottomrule
\end{tabular}
\caption{Number of classes of indecomposable GF(4)-linear codes of given distances}
\endgroup
\end{table}

\end{landscape}

%% file: fulltables.tex
\begin{figure}
\centering
    \centering
    \includegraphics[width=\textwidth]{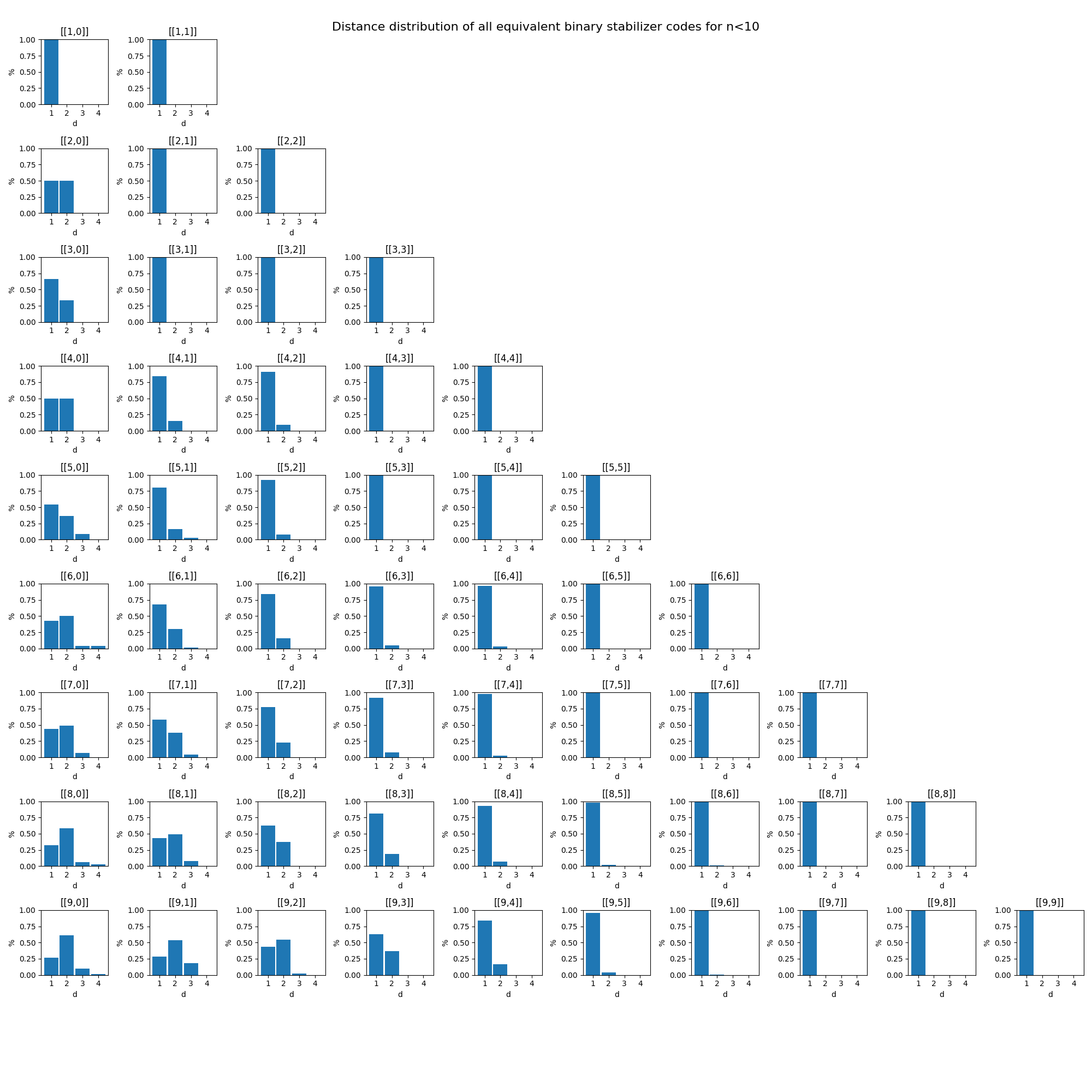}
    \caption{Distributions of minimum distances of all non-equivalent binary stabilizer code for $n<10$.
    \label{fig:distdistall}}
\end{figure}

\begin{figure}
\centering
    \centering
    \includegraphics[width=\textwidth]{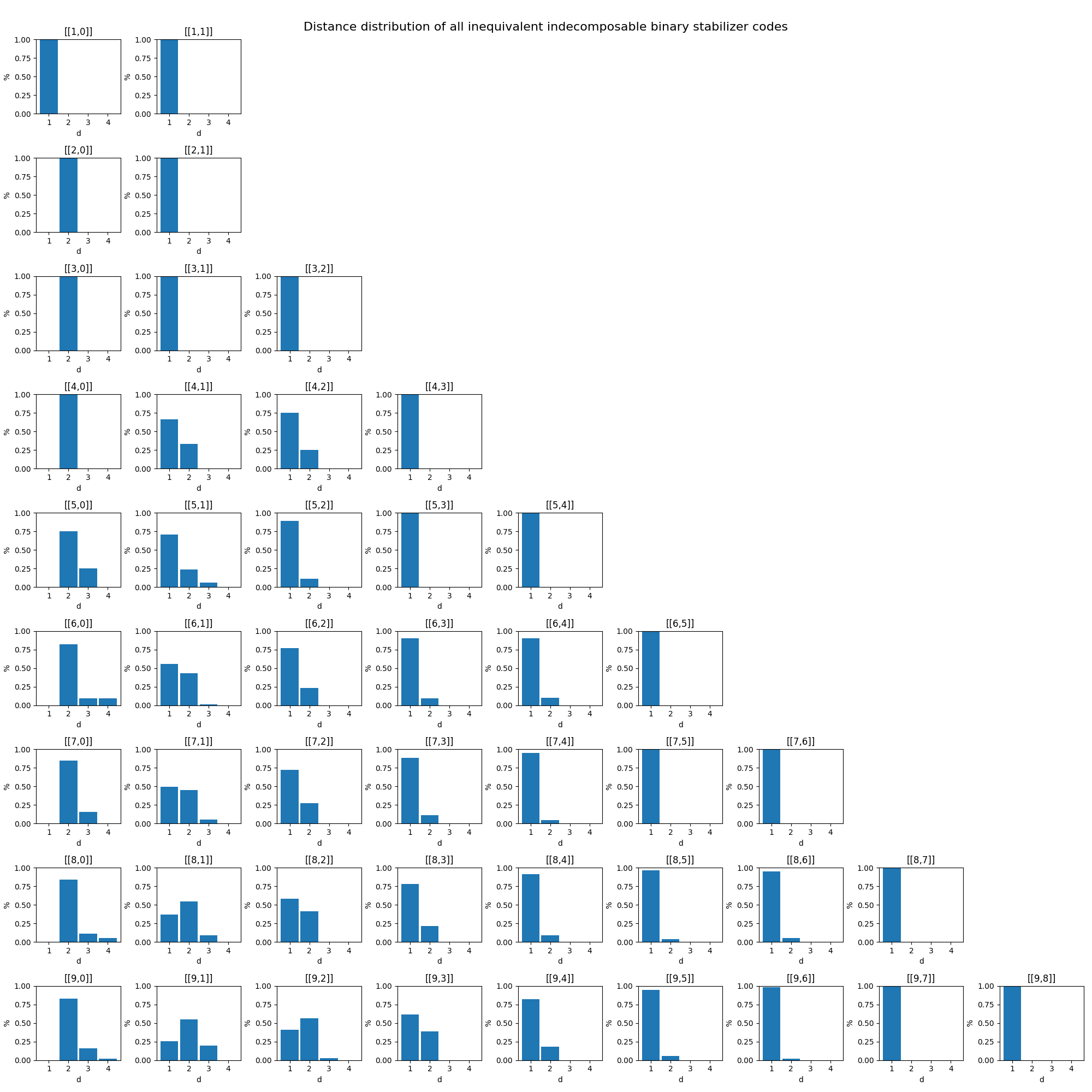}
    \caption{Distributions of minimum distances of all non-equivalent indecomposable binary stabilizer code for $n<10$.
    \label{fig:distdistindec}}
\end{figure}

\begin{landscape}
\begin{table}
\centering
\begingroup
\footnotesize 
\begin{tabular}{l|l|l|l|l|l}
\toprule
$[[n,k,d]]$ & $\mathrm{Idx}$  & $|\mathrm{Aut}(S)|$ & $S$ & Logicals & $w(x)$ \\ 
\specialrule{1.5pt}{1pt}{1pt} 
\midrule $[[1,1,1]]$  & 0 & 6 & $I_{}$ & $X_{0}, Z_{0}$ & $1$ \\ 
\midrule $[[2,0,2]]$  & 1 & 12 & $Z_{0}Z_{1}, X_{0}X_{1}$ &  & $1 + 3x^{2}$ \\ 
\midrule $[[2,1,1]]$  & 0 & 8 & $X_{0}X_{1}$ & $X_{1}, Z_{0}Z_{1}$ & $1 + x^{2}$ \\ 
\midrule $[[3,0,2]]$  & 1 & 24 & $X_{0}X_{1}, X_{0}X_{2}, Z_{0}Z_{1}Z_{2}$ &  & $1 + 3x^{2} + 4x^{3}$ \\ 
\midrule $[[3,1,1]]$  & 4 & 8 & $X_{0}X_{2}, Z_{0}Z_{1}Z_{2}$ & $X_{1}X_{2}, Z_{0}Z_{2}$ & $1 + x^{2} + 2x^{3}$ \\ 
 & 2 & 48 & $X_{0}X_{1}, X_{0}X_{2}$ & $X_{2}, Z_{0}Z_{1}Z_{2}$ & $1 + 3x^{2}$ \\ 
\midrule $[[3,2,1]]$  & 1 & 48 & $Z_{0}Z_{1}Z_{2}$ & $X_{0}X_{1}, X_{0}X_{2}, Z_{1}, Z_{2}$ & $1 + x^{3}$ \\
\bottomrule
\end{tabular}
\captionsetup{font=small} 
\caption{Representatives for every equivalence class of indecomposable codes, sorted by $n$ then $k$ and then by automorphism group order for $n=1,2,3,4$. \label{tab:all_code_list_ind_1_1_to_3_4}}
\endgroup
\end{table}

\begin{table}
\centering
\begingroup
\footnotesize 
\begin{tabular}{l|l|l|l|l|l}
\toprule
$[[n,k,d]]$ & $\mathrm{Idx}$ & $|\mathrm{Aut}(S)|$ & $S$ & Logicals & $w(x)$ \\ 
\midrule $[[5,0,2]]$  & 4 & 32 & $Z_{0}Z_{4}, Z_{1}Z_{2}, Z_{1}Z_{3}Z_{4}, X_{1}X_{2}X_{3}, X_{0}X_{3}X_{4}$ &  & $1 + 2x^{2} + 8x^{3} + 13x^{4} + 8x^{5}$ \\ 
 & 2 & 96 & $Z_{0}Z_{4}, Z_{1}Z_{4}, X_{2}X_{3}, Z_{2}Z_{3}Z_{4}, X_{0}X_{1}X_{3}X_{4}$ &  & $1 + 4x^{2} + 6x^{3} + 11x^{4} + 10x^{5}$ \\ 
 & 6 & 1920 & $Z_{0}Z_{4}, Z_{1}Z_{4}, Z_{2}Z_{4}, Z_{3}Z_{4}, X_{0}X_{1}X_{2}X_{3}X_{4}$ &  & $1 + 10x^{2} + 5x^{4} + 16x^{5}$ \\ 
\midrule $[[5,0,3]]$  & 7 & 120 & $X_{0}Z_{1}Z_{2}, Z_{0}X_{1}Z_{3}, Z_{0}X_{2}Z_{4}, Z_{1}X_{3}Z_{4}, Z_{2}Z_{3}X_{4}$ &  & $1 + 10x^{3} + 15x^{4} + 6x^{5}$ \\  
\bottomrule
\end{tabular}
\captionsetup{font=small} 
\caption{Representatives for every equivalence class of indecomposable codes, sorted by $n$ then $k$ and then by automorphism group order for $[[5,0,2]]$ and $[[5,0,3]]$. \label{tab:all_code_list_ind_5_0_to_5_0}}
\endgroup
\end{table}
\end{landscape}

\begin{landscape}
\begin{table}
\centering
\begingroup
\footnotesize 
\begin{tabular}{l|l|l|l|l|l}
\toprule
$[[n,k,d]]$ & $\mathrm{Idx}$ & $|\mathrm{Aut}(S)|$ & $S$ & Logicals & $w(x)$ \\ 
\specialrule{1.5pt}{1pt}{1pt}  
\midrule $[[5,1,1]]$  & 28 & 8 & $Z_{1}Z_{2}, Z_{1}Z_{3}Z_{4}, X_{1}X_{2}X_{3}, X_{0}X_{3}X_{4}$ & $X_{3}X_{4}, Z_{0}Z_{4}$ & $1 + x^{2} + 5x^{3} + 6x^{4} + 3x^{5}$ \\ 
 & 27 & 16 & $Z_{0}X_{1}Z_{3}, Z_{0}X_{2}Z_{4}, Z_{1}X_{3}Z_{4}, Z_{2}Z_{3}X_{4}$ & $X_{0}Z_{1}Z_{2}, Z_{0}$ & $1 + 6x^{3} + 7x^{4} + 2x^{5}$ \\ 
 & 7 & 32 & $Z_{0}Z_{4}, Z_{1}Z_{2}, Z_{1}Z_{3}Z_{4}, X_{1}X_{2}X_{3}$ & $X_{0}X_{3}X_{4}, Z_{4}$ & $1 + 2x^{2} + 6x^{3} + 5x^{4} + 2x^{5}$ \\ 
 & 33 & 32 & $Z_{1}Z_{4}, X_{2}X_{3}, Z_{2}Z_{3}Z_{4}, X_{0}X_{1}X_{3}X_{4}$ & $X_{1}X_{3}X_{4}, Z_{0}Z_{4}$ & $1 + 2x^{2} + 4x^{3} + 5x^{4} + 4x^{5}$ \\ 
 & 29 & 48 & $Z_{0}Z_{4}, Z_{1}Z_{4}, Z_{2}Z_{3}Z_{4}, X_{0}X_{1}X_{3}X_{4}$ & $X_{2}X_{3}, Z_{0}Z_{3}$ & $1 + 3x^{2} + 3x^{3} + 4x^{4} + 5x^{5}$ \\ 
 & 31 & 64 & $X_{0}Z_{4}, X_{1}X_{2}, Y_{1}Z_{2}Y_{3}Z_{4}, Z_{0}Z_{1}Z_{2}X_{4}$ & $Z_{1}Z_{2}X_{3}, Z_{1}Z_{2}Z_{3}$ & $1 + 2x^{2} + 13x^{4}$ \\ 
 & 32 & 64 & $Z_{0}Z_{4}, Z_{1}Z_{2}, X_{1}X_{2}X_{3}, X_{0}X_{3}X_{4}$ & $X_{3}, Z_{0}Z_{1}Z_{3}$ & $1 + 2x^{2} + 4x^{3} + 5x^{4} + 4x^{5}$ \\
\bottomrule
\end{tabular}
\captionsetup{font=small}
\caption{Representatives for every equivalence class of indecomposable codes, sorted by $n$ then $k$ and then by automorphism group order for $[[5,1,1]]$, $[[5,1,2]]$, $[[5,1,3]]$, 
$[[5,2,1]]$, $[[5,3,1]]$ and $[[5,4,1]]$. \label{tab:all_code_list_ind_5_1_to_5_4}}
\endgroup
\end{table}
\end{landscape}

\begin{landscape}
\begin{table}
    \centering
\begingroup
\footnotesize 
\begin{tabular}{l|l|l|l|l|l}
\toprule
$[[n,k,d]]$ & $\mathrm{Idx}$ & $|\mathrm{Aut}(S)|$ & $S$ & Logicals & $w(x)$ \\ 
\specialrule{1.5pt}{1pt}{1pt}
\midrule $[[6,0,2]]$  & 7 & 32 & $X_{0}X_{4}, Z_{1}Z_{5}, X_{2}X_{3}X_{4}, Z_{2}Z_{3}Z_{5}, Z_{0}Z_{3}Z_{4}, X_{1}X_{2}X_{5}$ &  & $1 + 2x^{2} + 8x^{3} + 17x^{4} + 24x^{5} + 12x^{6}$ \\ 
 & 20 & 32 & $X_{0}Z_{5}, X_{1}Z_{3}Z_{5}, X_{2}Z_{4}Z_{5}, Z_{1}X_{3}Z_{4}, Z_{2}Z_{3}X_{4}, Z_{0}Z_{1}Z_{2}X_{5}$ &  & $1 + x^{2} + 8x^{3} + 19x^{4} + 24x^{5} + 11x^{6}$ \\ 
 & 5 & 96 & $Z_{0}Z_{5}, Z_{1}Z_{5}, Z_{2}Z_{4}, Z_{3}Z_{4}Z_{5}, X_{2}X_{3}X_{4}, X_{0}X_{1}X_{3}X_{5}$ &  & $1 + 4x^{2} + 8x^{3} + 13x^{4} + 24x^{5} + 14x^{6}$ \\ 
 & 3 & 128 & $Z_{0}Z_{5}, X_{1}X_{3}, X_{2}X_{4}, Z_{1}Z_{3}Z_{5}, Z_{2}Z_{4}Z_{5}, X_{0}X_{3}X_{4}X_{5}$ &  & $1 + 3x^{2} + 8x^{3} + 15x^{4} + 24x^{5} + 13x^{6}$ \\ 
 & 11 & 384 & $Z_{0}Z_{4}, Z_{1}Z_{5}, Z_{2}Z_{3}, Z_{2}Z_{4}Z_{5}, X_{0}X_{2}X_{3}X_{4}, X_{1}X_{2}X_{3}X_{5}$ &  & $1 + 3x^{2} + 8x^{3} + 15x^{4} + 24x^{5} + 13x^{6}$ \\ 
 & 12 & 384 & $X_{0}Z_{3}, X_{1}Z_{4}, X_{2}Z_{5}, Z_{0}X_{3}Z_{4}Z_{5}, Z_{1}Z_{3}X_{4}Z_{5}, Z_{2}Z_{3}Z_{4}X_{5}$ &  & $1 + 3x^{2} + 39x^{4} + 21x^{6}$ \\ 
 & 9 & 768 & $Z_{0}Z_{5}, Z_{1}Z_{5}, Z_{2}Z_{5}, X_{3}X_{4}, Z_{3}Z_{4}Z_{5}, X_{0}X_{1}X_{2}X_{4}X_{5}$ &  & $1 + 7x^{2} + 8x^{3} + 7x^{4} + 24x^{5} + 17x^{6}$ \\
\bottomrule
\end{tabular}
\captionsetup{font=small}
\caption{Representatives for every equivalence class of indecomposable codes, sorted by $n$ then $k$ and then by automorphism group order for $[[6,0,2]]$. \label{tab:all_code_list_ind_6_0_to_6_0}}
\endgroup
\end{table}

\begin{table}
    \centering
\begingroup
\footnotesize 
\begin{tabular}{l|l|l|l|l|l}
\toprule
$[[n,k,d]]$ & $\mathrm{Idx}$ & $|\mathrm{Aut}(S)|$ & $S$ & Logicals & $w(x)$ \\ 
\specialrule{1.5pt}{1pt}{1pt}
\midrule $[[6,1,1]]$  & 90 & 8 & $Z_{0}Z_{1}Z_{4}, Z_{0}Z_{2}Z_{5}, Z_{3}Z_{4}Z_{5}, X_{1}X_{3}X_{4}, X_{2}X_{3}X_{5}$ & $X_{0}X_{4}X_{5}, Z_{2}Z_{5}$ & $1 + 6x^{3} + 11x^{4} + 10x^{5} + 4x^{6}$ \\ 
 & 97 & 8 & $X_{0}Z_{5}, X_{2}Z_{4}Z_{5}, Z_{1}X_{3}Z_{4}, Z_{2}Z_{3}X_{4}, Z_{0}Z_{1}Z_{2}X_{5}$ & $X_{1}Z_{3}Z_{5}, Z_{1}$ & $1 + x^{2} + 4x^{3} + 11x^{4} + 12x^{5} + 3x^{6}$ \\ 
 & 102 & 8 & $Z_{1}Z_{5}, X_{2}X_{3}X_{4}, Z_{2}Z_{3}Z_{5}, Z_{0}Z_{3}Z_{4}, X_{1}X_{2}X_{5}$ & $X_{0}X_{4}, Z_{1}Z_{2}Z_{4}$ & $1 + x^{2} + 6x^{3} + 9x^{4} + 10x^{5} + 5x^{6}$ \\ 
 & 94 & 16 & $X_{3}Z_{4}Z_{5}, X_{1}X_{2}X_{3}, Z_{1}Z_{3}X_{4}, Z_{2}Z_{3}X_{5}, Y_{0}Y_{1}Z_{2}Z_{4}$ & $Z_{0}X_{2}Z_{5}, Z_{0}Z_{1}Z_{2}$ & $1 + 4x^{3} + 13x^{4} + 12x^{5} + 2x^{6}$ \\ 
 & 98 & 16 & $X_{0}Z_{5}, X_{1}Z_{3}Z_{5}, X_{2}Z_{4}Z_{5}, Z_{2}Z_{3}X_{4}, Z_{0}Z_{1}Z_{2}X_{5}$ & $Z_{1}X_{3}Z_{4}, Z_{3}$ & $1 + x^{2} + 6x^{3} + 9x^{4} + 10x^{5} + 5x^{6}$ \\ 
 & 107 & 16 & $X_{1}Z_{3}Z_{5}, X_{2}Z_{4}Z_{5}, Z_{1}X_{3}Z_{4}, Z_{2}Z_{3}X_{4}, Z_{0}Z_{1}Z_{2}X_{5}$ & $X_{0}Z_{5}, Z_{0}$ & $1 + 6x^{3} + 11x^{4} + 10x^{5} + 4x^{6}$ \\ 
 & 21 & 32 & $X_{0}X_{4}, Z_{1}Z_{5}, X_{2}X_{3}X_{4}, Z_{2}Z_{3}Z_{5}, X_{1}X_{2}X_{5}$ & $X_{4}, Z_{0}Z_{1}Z_{2}Z_{4}$ & $1 + 2x^{2} + 6x^{3} + 11x^{4} + 10x^{5} + 2x^{6}$ \\
\bottomrule
\end{tabular}
\captionsetup{font=small}
\caption{Representatives for every equivalence class of indecomposable codes, sorted by $n$ then $k$ and then by automorphism group order for $[[6,1,1,]]$. \label{tab:all_code_list_ind_6_1_to_6_1_d1}}
\endgroup
\end{table}

\end{landscape}

\begin{landscape}
\begin{table}
    \centering
\begingroup
\footnotesize 
\begin{tabular}{l|l|l|l|l|l}
\toprule
$[[n,k,d]]$ & $\mathrm{Idx}$ & $|\mathrm{Aut}(S)|$ & $S$ & Logicals & $w(x)$ \\ 
\specialrule{1.5pt}{1pt}{1pt}
\midrule $[[6,1,2]]$  & 60 & 2 & $X_{0}X_{1}Z_{3}, X_{2}Z_{4}Z_{5}, Z_{1}X_{3}Z_{4}, Z_{2}Z_{3}X_{4}, Z_{0}Z_{1}Z_{2}X_{5}$ & $X_{1}Z_{3}Z_{5}, Z_{0}Z_{1}$ & $1 + 4x^{3} + 9x^{4} + 12x^{5} + 6x^{6}$ \\ 
 & 70 & 4 & $Z_{0}Z_{5}, Z_{2}Z_{4}Z_{5}, Z_{1}Z_{2}Z_{3}, X_{2}X_{3}X_{4}, X_{0}X_{1}X_{2}X_{5}$ & $X_{1}X_{3}, Z_{0}Z_{2}Z_{3}$ & $1 + x^{2} + 4x^{3} + 7x^{4} + 12x^{5} + 7x^{6}$ \\ 
 & 52 & 8 & $X_{1}Z_{3}Z_{5}, X_{2}Z_{4}Z_{5}, X_{0}Y_{1}Y_{3}Z_{4}, X_{0}X_{1}Z_{2}X_{4}, Z_{0}Z_{1}Z_{2}X_{5}$ & $Z_{2}Z_{3}X_{4}, Z_{0}Z_{3}Z_{4}$ & $1 + 2x^{3} + 11x^{4} + 14x^{5} + 4x^{6}$ \\ 
 & 65 & 8 & $X_{0}Z_{5}, Z_{2}Z_{3}X_{4}, X_{1}Y_{2}Y_{4}, Y_{1}Y_{3}Z_{4}Z_{5}, Z_{0}Z_{1}Z_{2}X_{5}$ & $Z_{1}X_{3}Z_{4}, Z_{1}Z_{2}Z_{3}$ & $1 + x^{2} + 2x^{3} + 9x^{4} + 14x^{5} + 5x^{6}$ \\ 
 & 66 & 8 & $X_{0}Z_{4}, Z_{2}X_{3}Z_{5}, Z_{0}X_{3}X_{4}, Y_{1}X_{2}Y_{5}, X_{1}Y_{2}Y_{3}Z_{4}$ & $X_{2}Z_{3}Z_{4}, Z_{1}Z_{2}$ & $1 + x^{2} + 4x^{3} + 7x^{4} + 12x^{5} + 7x^{6}$ \\ 
 & 37 & 12 & $X_{1}Z_{3}Z_{5}, X_{2}Z_{4}Z_{5}, Z_{2}Z_{3}X_{4}, X_{0}Y_{1}Y_{3}Z_{4}, Z_{0}Z_{1}Z_{2}X_{5}$ & $Z_{1}X_{3}Z_{4}, Z_{0}Z_{3}$ & $1 + 4x^{3} + 9x^{4} + 12x^{5} + 6x^{6}$ \\ 
 & 10 & 16 & $X_{0}Z_{4}, X_{1}Z_{5}, Z_{2}X_{3}Z_{5}, Z_{0}X_{2}Y_{3}Y_{4}, Z_{1}Z_{3}Z_{4}X_{5}$ & $Z_{0}Z_{2}X_{4}Z_{5}, Z_{2}Z_{4}$ & $1 + 2x^{2} + 2x^{3} + 7x^{4} + 14x^{5} + 6x^{6}$ \\
\bottomrule
\end{tabular}
\captionsetup{font=small}
\caption{Representatives for every equivalence class of indecomposable codes, sorted by $n$ then $k$ and then by automorphism group order for $[[6,1,2]]$ and $[[6,1,3]]$. \label{tab:all_code_list_ind_6_1_to_6_1_d23}}
\endgroup
\end{table}
\end{landscape}

\begin{landscape}
\begin{table}
    \centering
\begingroup
\footnotesize 
\begin{tabular}{l|l|l|l|l|l}
\toprule
$[[n,k,d]]$ & $\mathrm{Idx}$ & $|\mathrm{Aut}(S)|$ & $S$ & Logicals & $w(x)$ \\ 
\specialrule{1.5pt}{1pt}{1pt}
\midrule $[[6,2,1]]$  & 73 & 4 & $X_{2}Z_{4}Z_{5}, Z_{1}X_{3}Z_{4}, Z_{2}Z_{3}X_{4}, Z_{0}Z_{1}Z_{2}X_{5}$ & $X_{0}Z_{5}, X_{1}Z_{3}Z_{5}, Z_{0}, Z_{1}$ & $1 + 3x^{3} + 6x^{4} + 5x^{5} + x^{6}$ \\ 
 & 155 & 4 & $X_{2}Z_{4}Z_{5}, X_{0}Z_{1}X_{2}X_{3}, X_{0}Y_{2}Z_{3}Y_{4}, Z_{0}Z_{1}Z_{2}X_{5}$ & $Z_{2}Z_{3}X_{4}, X_{1}Z_{3}Z_{5}, Z_{0}Z_{3}Z_{4}, Z_{1}$ & $1 + x^{3} + 6x^{4} + 7x^{5} + x^{6}$ \\ 
 & 172 & 4 & $X_{0}X_{1}Z_{3}, X_{2}Z_{4}Z_{5}, Z_{2}Z_{3}X_{4}, Z_{0}Y_{2}X_{3}Y_{5}$ & $X_{1}, Z_{2}X_{5}, Z_{0}Z_{1}, Z_{3}Z_{5}$ & $1 + 3x^{3} + 4x^{4} + 5x^{5} + 3x^{6}$ \\ 
 & 174 & 4 & $X_{2}Z_{4}Z_{5}, Z_{1}X_{3}Z_{4}, X_{0}Y_{2}Z_{3}Y_{4}, Z_{0}Z_{1}Z_{2}X_{5}$ & $Z_{2}Z_{3}X_{4}, X_{1}Z_{3}Z_{5}, Z_{0}Z_{4}, Z_{1}$ & $1 + 2x^{3} + 5x^{4} + 6x^{5} + 2x^{6}$ \\ 
 & 176 & 4 & $Z_{2}Z_{3}X_{4}, X_{1}Y_{2}Y_{4}, Y_{1}Y_{3}Z_{4}Z_{5}, Z_{0}Z_{1}Z_{2}X_{5}$ & $X_{0}Z_{5}, Z_{1}X_{3}Z_{4}, Z_{0}, Z_{1}Z_{2}Z_{3}$ & $1 + 2x^{3} + 5x^{4} + 6x^{5} + 2x^{6}$ \\ 
 & 183 & 4 & $Z_{2}Z_{4}Z_{5}, Z_{1}Z_{2}Z_{3}, X_{2}X_{3}X_{4}, X_{0}X_{1}X_{2}X_{5}$ & $X_{1}X_{3}, X_{4}X_{5}, Z_{0}Z_{2}Z_{3}, Z_{0}Z_{5}$ & $1 + 3x^{3} + 4x^{4} + 5x^{5} + 3x^{6}$ \\ 
 & 78 & 8 & $X_{0}X_{1}X_{3}, Z_{2}Z_{4}Z_{5}, Z_{1}Z_{3}Z_{4}, X_{2}X_{3}X_{4}$ & $X_{1}X_{3}, X_{1}X_{4}X_{5}, Z_{0}Z_{2}Z_{3}, Z_{5}$ & $1 + 4x^{3} + 5x^{4} + 4x^{5} + 2x^{6}$ \\
\bottomrule
\end{tabular}
\captionsetup{font=small}
\caption{Part I: Representatives for every equivalence class of indecomposable codes, sorted by $n$ then $k$ and then by automorphism group order for $[[6,2,1]]$. \label{tab:all_code_list_ind_6_2_to_6_2_d1_I}}
\endgroup
\end{table}
\end{landscape}

\begin{landscape}
\begin{table}
    \centering
\begingroup
\footnotesize 
\begin{tabular}{l|l|l|l|l|l}
\toprule
$[[n,k,d]]$ & $\mathrm{Idx}$ & $|\mathrm{Aut}(S)|$ & $S$ & Logicals & $w(x)$ \\ 
\specialrule{1.5pt}{1pt}{1pt}
$[[6,2,1]]$ & 45 & 48 & $Z_{0}Z_{1}Z_{4}, Z_{0}Z_{2}Z_{5}, Z_{3}Z_{4}Z_{5}, X_{2}X_{3}X_{5}$ & $X_{1}X_{3}X_{4}, X_{0}X_{1}X_{3}X_{5}, Z_{4}, Z_{2}Z_{5}$ & $1 + 5x^{3} + 6x^{4} + 3x^{5} + x^{6}$ \\ 
 & 169 & 48 & $X_{1}X_{2}Z_{3}Z_{4}, Y_{1}Y_{3}Z_{4}Z_{5}, X_{1}Z_{2}X_{4}Z_{5}, Z_{0}Y_{1}Z_{2}Z_{3}Y_{5}$ & $X_{0}Z_{1}Z_{2}Z_{3}Z_{4}, Z_{1}Z_{2}X_{5}, Z_{0}, Z_{1}Z_{2}Z_{3}Z_{4}Z_{5}$ & $1 + 7x^{4} + 8x^{5}$ \\ 
 & 23 & 64 & $X_{0}X_{4}, Z_{1}Z_{5}, Z_{2}Z_{3}Z_{5}, X_{1}X_{3}X_{4}X_{5}$ & $X_{4}, X_{2}X_{3}, Z_{0}Z_{1}Z_{4}, Z_{1}Z_{3}$ & $1 + 2x^{2} + 2x^{3} + 5x^{4} + 6x^{5}$ \\ 
 & 24 & 64 & $X_{0}Z_{4}, X_{1}Z_{5}, Y_{2}Y_{3}Z_{4}Z_{5}, Z_{1}Z_{3}Z_{4}X_{5}$ & $Z_{0}Z_{2}X_{4}Z_{5}, Z_{2}X_{3}Z_{5}, Z_{4}, Z_{2}Z_{3}$ & $1 + 2x^{2} + 11x^{4} + 2x^{6}$ \\ 
 & 25 & 64 & $Z_{2}Z_{5}, Z_{3}Z_{4}, X_{1}X_{2}X_{5}, Z_{0}Z_{1}Z_{4}Z_{5}$ & $X_{0}X_{3}X_{4}, X_{0}X_{2}X_{5}, Z_{4}, Z_{1}Z_{5}$ & $1 + 2x^{2} + 2x^{3} + 5x^{4} + 6x^{5}$ \\ 
 & 26 & 64 & $Z_{0}Z_{1}, Z_{3}Z_{5}, Z_{2}Z_{4}Z_{5}, X_{0}X_{1}X_{2}X_{4}$ & $X_{2}X_{4}, X_{2}X_{3}X_{5}, Z_{0}Z_{4}, Z_{5}$ & $1 + 2x^{2} + 2x^{3} + 3x^{4} + 6x^{5} + 2x^{6}$ \\ 
 & 27 & 64 & $Z_{1}Z_{5}, Z_{2}Z_{3}, Z_{2}Z_{4}Z_{5}, X_{0}X_{2}X_{3}X_{4}$ & $X_{2}X_{3}X_{4}, X_{1}X_{2}X_{3}X_{5}, Z_{0}Z_{4}, Z_{5}$ & $1 + 2x^{2} + 4x^{3} + 3x^{4} + 4x^{5} + 2x^{6}$ \\
\bottomrule
\end{tabular}
\captionsetup{font=small}
\caption{Part II: Representatives for every equivalence class of indecomposable codes, sorted by $n$ then $k$ and then by automorphism group order for $[[6,2,1]]$. \label{tab:all_code_list_ind_6_2_to_6_2_d1_II}}
\endgroup
\end{table}
\end{landscape}

\begin{landscape}
\begin{table}
    \centering
\begingroup
\footnotesize 
\begin{tabular}{l|l|l|l|l|l|l}
\toprule
$[[n,k,d]]$ & $\mathrm{Idx}$ & $|\mathrm{Aut}(S)|$ & $S$ & Logicals & $w(x)$ \\ 
\specialrule{1.5pt}{1pt}{1pt}
\midrule $[[6,2,2]]$  & 133 & 2 & $Z_{2}Z_{3}X_{4}, X_{0}Y_{1}Y_{3}Z_{4}, Z_{1}X_{2}X_{3}Z_{5}, Z_{0}Z_{1}Z_{2}X_{5}$ & $Z_{0}X_{1}Z_{2}Z_{5}, Z_{0}X_{3}Z_{4}, Z_{0}Z_{1}, Z_{0}Z_{2}Z_{3}$ & $1 + x^{3} + 4x^{4} + 7x^{5} + 3x^{6}$ \\ 
 & 75 & 4 & $X_{0}X_{1}Z_{3}, Z_{2}Z_{3}X_{4}, Z_{1}X_{2}X_{3}Z_{5}, Z_{0}Y_{3}Y_{4}X_{5}$ & $X_{1}Z_{2}, Z_{2}X_{5}, Z_{0}Z_{1}, Z_{2}Z_{3}Z_{5}$ & $1 + 2x^{3} + 3x^{4} + 6x^{5} + 4x^{6}$ \\ 
 & 99 & 4 & $Z_{2}Z_{4}Z_{5}, X_{3}X_{4}X_{5}, X_{0}X_{1}X_{2}X_{4}, Z_{0}Z_{1}Z_{3}Z_{5}$ & $X_{2}X_{4}, X_{1}X_{2}X_{5}, Z_{0}Z_{3}Z_{4}, Z_{3}Z_{5}$ & $1 + 2x^{3} + 3x^{4} + 6x^{5} + 4x^{6}$ \\ 
 & 77 & 8 & $X_{3}Z_{4}, Z_{1}Z_{2}X_{5}, X_{0}X_{1}Z_{4}Z_{5}, Z_{0}X_{2}Z_{3}X_{4}Z_{5}$ & $Z_{0}Z_{3}X_{4}, X_{1}Z_{5}, Z_{2}Z_{4}, Z_{0}Z_{1}$ & $1 + x^{2} + x^{3} + 2x^{4} + 7x^{5} + 4x^{6}$ \\ 
 & 117 & 8 & $X_{0}Y_{1}Y_{3}Z_{4}, Z_{1}X_{2}X_{3}Z_{5}, X_{0}X_{1}Z_{2}X_{4}, Y_{0}Y_{1}Z_{2}Z_{3}X_{5}$ & $Z_{0}X_{1}Z_{2}Z_{4}, Z_{2}X_{5}, Z_{0}Z_{1}, Z_{0}Z_{2}Z_{3}Z_{4}Z_{5}$ & $1 + 5x^{4} + 8x^{5} + 2x^{6}$ \\ 
 & 85 & 12 & $X_{0}X_{3}X_{4}, X_{1}X_{2}X_{5}, Z_{0}Z_{1}Z_{4}Z_{5}, Z_{2}Z_{3}Z_{4}Z_{5}$ & $X_{2}X_{3}, X_{4}X_{5}, Z_{0}Z_{3}, Z_{1}Z_{5}$ & $1 + 2x^{3} + 3x^{4} + 6x^{5} + 4x^{6}$ \\ 
 & 108 & 12 & $X_{0}X_{1}Z_{3}, X_{2}Z_{4}Z_{5}, Z_{1}Z_{2}Y_{3}Y_{4}, Z_{0}Y_{2}X_{3}Y_{5}$ & $X_{1}Z_{4}, Z_{2}X_{5}, Z_{0}Z_{1}, Z_{3}Z_{4}Z_{5}$ & $1 + 2x^{3} + 3x^{4} + 6x^{5} + 4x^{6}$ \\
\bottomrule
\end{tabular}
\captionsetup{font=small}
\caption{Representatives for every equivalence class of indecomposable codes, sorted by $n$ then $k$ and then by automorphism group order for $[[6,2,2]]$. \label{tab:all_code_list_ind_6_2_to_6_2_d2}}
\endgroup
\end{table}
\end{landscape}

\begin{landscape}
\begin{table}
    \centering
\begingroup
\footnotesize 
\begin{tabular}{l|l|l|l|l|l}
\toprule
$[[n,k,d]]$ & $\mathrm{Idx}$ & $|\mathrm{Aut}(S)|$ & $S$ & Logicals & $w(x)$ \\ 
\specialrule{1.5pt}{1pt}{1pt}
\midrule $[[6,3,1]]$  & 80 & 8 & $Z_{2}Z_{3}X_{4}, Z_{1}X_{2}X_{3}Z_{5}, Z_{0}Y_{3}Y_{4}X_{5}$ & $X_{1}Z_{2}, X_{0}Z_{2}Z_{3}, Z_{2}X_{5}, Z_{1}, Z_{0}, Z_{2}Z_{3}Z_{5}$ & $1 + x^{3} + 2x^{4} + 3x^{5} + x^{6}$ \\ 
 & 104 & 8 & $Z_{3}Z_{4}X_{5}, X_{0}Z_{1}X_{3}Z_{5}, Z_{0}X_{1}Z_{2}X_{4}Z_{5}$ & $X_{3}Z_{5}, X_{4}Z_{5}, Z_{1}X_{2}, Z_{0}Z_{3}, Z_{1}Z_{4}, Z_{2}$ & $1 + x^{3} + x^{4} + 3x^{5} + 2x^{6}$ \\ 
 & 105 & 8 & $X_{0}Z_{1}X_{3}Z_{5}, Z_{2}Z_{3}Y_{4}Y_{5}, Z_{0}X_{1}Z_{2}X_{4}Z_{5}$ & $Z_{1}X_{3}Z_{4}, X_{2}Z_{4}, Z_{0}Z_{4}X_{5}, Z_{0}Z_{3}, Z_{2}, Z_{1}Z_{4}Z_{5}$ & $1 + 2x^{4} + 4x^{5} + x^{6}$ \\ 
 & 38 & 16 & $X_{3}X_{4}X_{5}, X_{0}X_{1}X_{2}X_{4}, Z_{0}Z_{1}Z_{3}Z_{5}$ & $X_{2}, X_{4}, X_{1}X_{5}, Z_{0}Z_{2}, Z_{0}Z_{3}Z_{4}, Z_{3}Z_{5}$ & $1 + x^{3} + 2x^{4} + 3x^{5} + x^{6}$ \\ 
 & 53 & 16 & $Z_{2}Z_{4}Z_{5}, Z_{1}Z_{3}Z_{4}, X_{0}X_{3}X_{4}X_{5}$ & $X_{1}X_{3}, X_{1}X_{2}X_{4}, X_{2}X_{5}, Z_{0}Z_{3}, Z_{0}Z_{4}, Z_{0}Z_{5}$ & $1 + 2x^{3} + 2x^{4} + 2x^{5} + x^{6}$ \\ 
 & 54 & 16 & $Z_{2}Z_{3}X_{4}, Z_{1}X_{2}X_{3}Z_{5}, Z_{0}Z_{1}Z_{2}X_{5}$ & $X_{3}Z_{4}, X_{1}Z_{2}Z_{5}, X_{0}Z_{5}, Z_{2}Z_{3}, Z_{1}, Z_{0}$ & $1 + x^{3} + 3x^{4} + 3x^{5}$ \\ 
 & 55 & 16 & $Z_{2}Z_{4}Z_{5}, X_{3}X_{4}X_{5}, Z_{0}Z_{1}Z_{3}Z_{5}$ & $X_{0}X_{1}, X_{2}X_{4}, X_{1}X_{2}X_{5}, Z_{0}, Z_{3}Z_{4}, Z_{3}Z_{5}$ & $1 + 2x^{3} + 2x^{4} + 2x^{5} + x^{6}$ \\
\bottomrule
\end{tabular}
\captionsetup{font=small}
\caption{Part I: Representatives for every equivalence class of indecomposable codes, sorted by $n$ then $k$ and then by automorphism group order for $[[6,3,1]]$. \label{tab:all_code_list_ind_6_3_to_6_3_d1_I}}
\endgroup
\end{table}
\end{landscape}

\begin{landscape}
\begin{table}
    \centering
\begingroup
\footnotesize 
\begin{tabular}{l|l|l|l|l|l}
\toprule
$[[n,k,d]]$ & $\mathrm{Idx}$ & $|\mathrm{Aut}(S)|$ & $S$ & Logicals & $w(x)$ \\ 
\specialrule{1.5pt}{1pt}{1pt}
\midrule $[[6,3,1]]$ & 108 & 96 & $Z_{2}Z_{5}, Z_{1}Z_{3}Z_{4}Z_{5}, X_{0}X_{1}X_{2}X_{3}X_{4}X_{5}$ & $X_{1}X_{3}, X_{1}X_{4}, X_{1}X_{2}X_{5}, Z_{0}Z_{3}, Z_{0}Z_{4}, Z_{0}Z_{5}$ & $1 + x^{2} + 2x^{4} + 4x^{6}$ \\ 
 & 3 & 128 & $Z_{3}Z_{5}, Z_{2}Z_{4}Z_{5}, X_{0}X_{1}X_{2}X_{4}$ & $X_{1}, X_{2}X_{4}, X_{2}X_{3}X_{5}, Z_{0}Z_{1}, Z_{0}Z_{4}, Z_{5}$ & $1 + x^{2} + 2x^{3} + x^{4} + 2x^{5} + x^{6}$ \\ 
 & 11 & 128 & $X_{2}Z_{5}, X_{0}X_{3}Z_{4}Z_{5}, X_{1}Z_{3}X_{4}Z_{5}$ & $Z_{1}X_{3}, Z_{0}X_{4}, Z_{0}Z_{1}Z_{2}X_{5}, Z_{0}Z_{3}, Z_{1}Z_{4}, Z_{5}$ & $1 + x^{2} + 5x^{4} + x^{6}$ \\ 
 & 25 & 128 & $X_{0}X_{1}, Z_{2}Z_{4}Z_{5}, Z_{0}Z_{1}Z_{3}Z_{5}$ & $X_{1}X_{3}, X_{2}X_{4}, X_{1}X_{2}X_{5}, Z_{3}, Z_{4}, Z_{5}$ & $1 + x^{2} + x^{3} + 2x^{4} + 3x^{5}$ \\ 
 & 40 & 128 & $Z_{3}Z_{4}X_{5}, X_{2}Y_{3}Y_{5}, Z_{0}X_{1}X_{2}Z_{4}$ & $X_{3}Z_{5}, Z_{2}X_{4}Z_{5}, X_{0}Z_{1}, Z_{1}Z_{2}Z_{3}, Z_{4}, Z_{0}$ & $1 + 2x^{3} + 3x^{4} + 2x^{5}$ \\ 
 & 63 & 128 & $Z_{0}Z_{1}Z_{2}Z_{4}, Z_{0}Z_{1}Z_{3}Z_{5}, Y_{2}Y_{3}Y_{4}Y_{5}$ & $X_{0}X_{1}, X_{1}X_{3}X_{4}, X_{3}X_{5}, Z_{0}, Z_{2}Z_{4}, Z_{2}Z_{5}$ & $1 + 5x^{4} + 2x^{6}$ \\ 
 & 64 & 128 & $X_{0}Z_{1}X_{2}Z_{4}, X_{0}Z_{1}X_{3}Z_{5}, Z_{0}X_{1}Z_{3}Z_{4}X_{5}$ & $Z_{1}X_{3}, Z_{1}Z_{2}X_{4}, Z_{0}Z_{2}X_{5}, Z_{0}Z_{2}Z_{3}, Z_{4}, Z_{1}Z_{5}$ & $1 + 3x^{4} + 4x^{5}$ \\
\bottomrule
\end{tabular}
\captionsetup{font=small}
\caption{Part II: Representatives for every equivalence class of indecomposable codes, sorted by $n$ then $k$ and then by automorphism group order for $[[6,3,1]]$ and $[[6,3,2]]$. \label{tab:all_code_list_ind_6_3_to_6_3_d1_II}}
\endgroup
\end{table}
\end{landscape}

\begin{landscape}
\begin{table}
    \centering
\begingroup
\footnotesize 
\begin{tabular}{l|l|l|l|l|l}
\toprule
$[[n,k,d]]$ & $\mathrm{Idx}$ & $|\mathrm{Aut}(S)|$ & $S$ & Logicals & $w(x)$ \\ 
\specialrule{1.5pt}{1pt}{1pt}
\midrule $[[6,4,1]]$  & 25 & 128 & $Z_{2}Z_{3}Y_{4}Y_{5}, Z_{0}X_{1}Z_{2}X_{4}Z_{5}$ & $X_{2}Z_{4}, Z_{1}X_{3}Z_{4}, X_{0}Z_{1}, Z_{4}X_{5}, Z_{2}, Z_{3}, Z_{0}, Z_{1}Z_{4}Z_{5}$ & $1 + x^{4} + 2x^{5}$ \\ 
 & 17 & 192 & $X_{3}X_{4}X_{5}, Z_{0}Z_{1}Z_{2}Z_{4}Z_{5}$ & $X_{1}X_{2}, X_{0}X_{1}, X_{1}X_{4}, X_{1}X_{5}, Z_{2}, Z_{0}, Z_{3}Z_{4}, Z_{3}Z_{5}$ & $1 + x^{3} + x^{5} + x^{6}$ \\ 
 & 26 & 192 & $Z_{0}Z_{1}Z_{2}Z_{4}Z_{5}, X_{0}X_{1}X_{3}X_{4}X_{5}$ & $X_{1}X_{2}, X_{3}, X_{1}X_{4}, X_{1}X_{5}, Z_{2}, Z_{0}Z_{3}, Z_{0}Z_{4}, Z_{0}Z_{5}$ & $1 + 2x^{5} + x^{6}$ \\ 
 & 21 & 256 & $Z_{0}Z_{1}Z_{3}Z_{5}, X_{2}X_{3}X_{4}X_{5}$ & $X_{0}X_{1}, X_{1}X_{3}, X_{4}, X_{1}X_{5}, Z_{0}, Z_{2}Z_{3}, Z_{2}Z_{4}, Z_{2}Z_{5}$ & $1 + 2x^{4} + x^{6}$ \\ 
 & 28 & 384 & $X_{2}X_{3}X_{4}X_{5}, Z_{0}Z_{1}Z_{2}Z_{3}Z_{4}Z_{5}$ & $X_{0}X_{1}, X_{1}X_{3}, X_{1}X_{4}, X_{1}X_{5}, Z_{0}, Z_{2}Z_{3}, Z_{2}Z_{4}, Z_{2}Z_{5}$ & $1 + x^{4} + 2x^{6}$ \\ 
 & 16 & 768 & $Z_{2}Z_{4}Z_{5}, Z_{0}Z_{1}Z_{3}Z_{5}$ & $X_{0}X_{1}, X_{0}X_{3}, X_{2}X_{4}, X_{0}X_{2}X_{5}, Z_{1}, Z_{3}, Z_{4}, Z_{5}$ & $1 + x^{3} + x^{4} + x^{5}$ \\ 
 & 24 & 1536 & $Z_{3}Z_{5}, X_{0}X_{1}X_{2}Y_{3}X_{4}Y_{5}$ & $X_{2}, X_{1}, X_{4}, X_{3}X_{5}, Z_{0}Z_{2}, Z_{0}Z_{1}, Z_{0}Z_{4}, Z_{0}Z_{5}$ & $1 + x^{2} + 2x^{6}$ \\
\bottomrule
\end{tabular}
\captionsetup{font=small}
\caption{Representatives for every equivalence class of indecomposable codes, sorted by $n$ then $k$ and then by automorphism group order for $[[6,4,1]]$, $[[6,4,2]]$ and $[[6,5,1]]$. \label{tab:all_code_list_ind_6_4_to_6_5_d1_II}}
\endgroup
\end{table}
\end{landscape}